\documentclass[a4paper,UKenglish,cleveref, autoref, thm-restate]{lipics-v2021}
\title{Space-time deterministic graph rewriting}
\author{Pablo Arrighi}{Université Paris-Saclay, Inria, CNRS, LMF, 91190 Gif-sur-Yvette, France}{}{}{}%
\author{Marin Costes}{Université Paris-Saclay, Inria, CNRS, LMF, 91190 Gif-sur-Yvette, France}{}{}{}%
\author{Gilles Dowek}{Université Paris-Saclay, Inria, CNRS, LMF, 91190 Gif-sur-Yvette, France}{}{}{}%
\author{Luidnel Maignan}{Univ Paris Est Creteil, LACL, 94000, Creteil, France}{}{}{}%
\authorrunning{P. Arrighi et al.}
\Copyright{Pablo Arrighi, Marin Costes, Gilles Dowek and Luidnel Maignan}


\nolinenumbers
\usepackage{graphicx}
\usepackage{todonotes}
\usepackage{amssymb}
\usepackage{amsmath}
\usepackage{amsfonts}
\usepackage{latexsym}
\usepackage{txfonts}
\usepackage{skak}
\usepackage[createShortEnv]{proof-at-the-end}
\usepackage[utf8]{inputenc}
\usepackage{enumitem}
\usepackage{subcaption}
\usepackage{forloop}
\usepackage{fp}
\usepackage{scalefnt}
\setitemize{itemsep=1pt,topsep=1pt,parsep=1pt,partopsep=1pt,leftmargin=*}

\newcommand{\MC}[1]{{\color{lowgreen}#1}}
\renewcommand{\MC}[1]{{#1}}

\newcommand{\LM}[1]{{\color{purple}#1}}
\renewcommand{\LM}[1]{{#1}}

\newcommand{\position}{position}
\newcommand{\Z}{\mathbb{Z}}
\newcommand{\Past}{\mathrm{Past}}
\newcommand{\I}{\mathrm{I}}
\newcommand{\B}{\mathrm{B}}
\newcommand{\V}{\mathrm{V}}

\newcommand{\E}{\mathrm{E}}
\newcommand{\Em}{\mathrm{E}}
\newcommand{\pof}{\pi.}
\renewcommand{\P}{\pof\E}
\newcommand{\Px}[1]{\pof\E_{#1}(x)}

\newcommand{\restr}[2]{\mathcal{N}_{#1}(#2)}
\newcommand{\restri}[1]{{\mathcal{N}_{#1}}}
\newcommand{\crestr}[2]{{{\overline{\mathcal{N}}}_{#1}(#2)}}
\newcommand{\crestri}[1]{{{\overline{\mathcal{N}}}_{#1}}}

\newcommand{\precursor}{-}
\renewcommand{\prec}[2]{{\mathcal{N}^{\precursor}_{#1}(#2)}}
\newcommand{\preci}[1]{{\mathcal{N}^{\precursor}_{#1}}}

\newcommand{\cpreci}[1]{{\overline{\mathcal{N}^{\precursor}_{#1}}}}
\newcommand{\zee}[1]{{\scriptstyle \mathbb Z}.#1}
\renewcommand{\restriction}{\mathrel{\!\raisebox{-.5ex}{$\vert$}}}

\newcommand{\VL}[1]{}
\newcommand{\withproofs}[1]{}

\newcommand{\calG}{{\ensuremath{\mathcal{G}}}}

\newcommand{\calV}{{\ensuremath{\mathcal{V}}}}

\newcommand{\calX}{{\ensuremath{\mathcal{X}}}}

\newcommand{\X}{\calX}

\definecolor{lowgreen}{rgb}{0.40390625,0.6109375,0.09109375}
\definecolor{myblack}{RGB}{0,0,0}
\definecolor{border}{RGB}{206,206,206}
\definecolor{port}{RGB}{155,155,155}
\definecolor{setBorder}{RGB}{80,227,194}
\definecolor{internal}{RGB}{38,105,185}
\definecolor{greenstate}{RGB}{80,183,100}
\definecolor{greenstategray}{RGB}{158,220,170}
\definecolor{redstate}{RGB}{201,109,76}
\DeclareRobustCommand{\sizeFone}[1]{\Huge{#1}} 
\DeclareRobustCommand{\sizeFonex}[1]{\Large{#1}} 
\DeclareRobustCommand{\sizeFoneAx}[1]{\Huge{#1}} 
\DeclareRobustCommand{\sizeFtwo}[0]{\Large} 
\DeclareRobustCommand{\sizeFtwop}[0]{\normalsize} 
\DeclareRobustCommand{\sizeFthree}[0]{\LARGE} 
\DeclareRobustCommand{\sizeFthreep}[0]{\normalsize} 
\definecolor{internalStates}{RGB}{240,240,240}
\DeclareRobustCommand{\sizeFfour}[0]{\LARGE \scalefont{0.9}} 
\DeclareRobustCommand{\sizeFfourp}[0]{\normalsize} 
\DeclareRobustCommand{\sizeFfours}[0]{\LARGE} 
\DeclareRobustCommand{\sizeFfive}[0]{\Huge} 
\DeclareRobustCommand{\sizeFfivep}[0]{\Huge} 
\DeclareRobustCommand{\sizeFsixp}[0]{\LARGE} 
\DeclareRobustCommand{\sizeFsixs}[0]{\Huge} 
\DeclareRobustCommand{\sizeFsix}[0]{\Huge \scalefont{2}}
\DeclareRobustCommand{\sizeFhuitp}[0]{\Huge} 
\DeclareRobustCommand{\sizeFhuit}[0]{\Huge \scalefont{2}}
\DeclareRobustCommand{\sizeFdixap}[0]{\LARGE} 
\DeclareRobustCommand{\sizeFdixa}[0]{\Huge \scalefont{1.3}}
\DeclareRobustCommand{\sizeFdixas}[0]{\LARGE}
\DeclareRobustCommand{\sizeFdixcp}[0]{\Huge} 
\DeclareRobustCommand{\sizeFdixc}[0]{\Huge \scalefont{2}}
\DeclareRobustCommand{\sizehalfCA}[0]{\LARGE \scalefont{1.1}} 
\DeclareRobustCommand{\sizehalfCAb}[0]{\Huge \scalefont{1.5}} 
\DeclareRobustCommand{\sizebighalfCA}[0]{\LARGE} 
\DeclareRobustCommand{\sizebighalfCAb}[0]{\LARGE} 
\DeclareRobustCommand{\sizesimu}[0]{\Large} 
\DeclareRobustCommand{\sizesimub}[0]{\LARGE} 
\DeclareRobustCommand{\sizesimup}[0]{\large} 
\DeclareRobustCommand{\sizebigsimu}[0]{\Huge \scalefont{0.95}} 
\DeclareRobustCommand{\sizebigsimub}[0]{\Huge \scalefont{1.2}} 

\newcommand{\harpvecsign}{\scriptscriptstyle\rightharpoonup}
\newcommand{\harpoonvec}[2]{%
  \ifx\displaystyle#1\doalign{$\harpvecsign$}{#1#2}\fi
  \ifx\textstyle#1\doalign{$\harpvecsign$}{#1#2}\fi
  \ifx\scriptstyle#1\doalign{\scalebox{.6}[.9]{$\harpvecsign$}}{#1#2}\fi
  \ifx\scriptscriptstyle#1\doalign{\scalebox{.5}[.8]{$\harpvecsign$}}{#1#2}\fi
}
\usetikzlibrary{patterns}
\tikzset{every picture/.style={line width=0.75pt}} 

\newcommand{\emptycell}[4]{
    \tikzstyle{roundnode}=[circle,draw = #3,fill=white, minimum size = 56];
    \node[roundnode] (#4) at (#1,#2) {};
    \draw[color = #3] (#1,#2 -1) -- (#1,#2 +1);
}

\newcommand{\leftmoover}[4]{
    \tikzstyle{roundnode}=[circle,draw = #3,fill=white, minimum size = 56];
    \node[roundnode] (#4) at (#1,#2) {};
    \draw[color = #3] (#1,#2 -1) -- (#1,#2 +1);
    \begin{scope}
        \clip (#1,#2-1) rectangle (#1-1,#2+1);
        \draw[color = #3][fill = #3] (#1,#2) circle(1);
    \end{scope}
}

\newcommand{\rightmoover}[4]{
    \tikzstyle{roundnode}=[circle,draw = #3,fill=white, minimum size = 56];
    \node[roundnode] (#4) at (#1,#2) {};
    \draw[color = #3] (#1,#2 -1) -- (#1,#2 +1);
    \begin{scope}
        \clip (#1,#2-1) rectangle (#1+1,#2+1);
        \draw[color = #3][fill = #3] (#1,#2) circle(1);
    \end{scope}
}

\newcommand{\rightleft}[4]{
    \tikzstyle{roundnode}=[circle,draw = #3,fill=#3, minimum size = 56];
    \node[roundnode] (#4) at (#1,#2) {};
}

\newcommand{\redcell}[3]{
    \tikzstyle{roundnode}=[circle,draw = myblack,fill=redstate, minimum size = 56];
    \node[roundnode] (#3) at (#1,#2) {};
}

\newcommand{\redcellinternal}[3]{
    \tikzstyle{roundnode}=[circle,draw = internal,fill=redstate, minimum size = 56];
    \node[roundnode] (#3) at (#1,#2) {};
}

\newcommand{\greencellgray}[3]{
    \tikzstyle{roundnode}=[circle,draw = border,fill=greenstategray, minimum size = 56];
    \node[roundnode] (#3) at (#1,#2) {};
}

\newcommand{\greencellinternal}[3]{
    \tikzstyle{roundnode}=[circle,draw = internal,fill=greenstate, minimum size = 56];
    \node[roundnode] (#3) at (#1,#2) {};
}

\newcommand{\edge}[3]{
    \draw[-stealth,color=#3,line width=1.75mm] (#1) to (#2);
}

\newcommand{\vertex}[6]{
    \ifnum #4=0
    {\emptycell{#1}{#2}{#3}{#6}}
    \fi
    \ifnum #4=1
    {\rightmoover{#1}{#2}{#3}{#6}}
    \fi
    \ifnum #4=2
    {\leftmoover{#1}{#2}{#3}{#6}}
    \fi
    \ifnum #4=3
    {\rightleft{#1}{#2}{#3}{#6}}
    \fi
    \draw (#1+1,#2) node[right]{{\Huge \color{#3} #5}};
}
\begin{document}
%
%
%
%
\maketitle              
\begin{abstract}
We study non-terminating graph rewriting models, whose local rules are applied non-deterministically---and yet enjoy a strong form of determinism, namely space-time determinism. In the case of terminating computation, it is well-known that the property of confluence may ensure a deterministic end result. In the context of distributed, non-terminating computation however, confluence alone is too weak a property. 
Here we provide sufficient conditions so that asynchronous local rule applications produce well-determined events in the space-time unfolding of the graph, regardless of their application orders. Our first two examples are asynchronous simulation of a dynamical systems. Our third example features time dilation, in the spirit of general relativity.
\end{abstract}
\newpage
\section{Introduction}


\subsection{Context}

{\em Dynamical systems} refer to the global and synchronous (or almost synchronous \cite{MairesseIPS}) evolution of an entire configuration at time $t$ into another at time $t+1$, $t+2$, etc., iteratively.
They often have a spatial dimension too, typically grid-based (e.g.  representing particles \cite{Zuse,Wolf-Gladrow}, fluids \cite{ChopardDroz,RothmanZaleski}, traffic jams \cite{NagelSchreckenberg}, demographics and regional development or consumption \cite{Bruun,WhiteEngelen}) or graph-based (e.g. representing physical systems \cite{MeyerLove}, computer processes \cite{PapazianRemila}, biochemical agents \cite{MurrayDicksonVol2}, economical agents \cite{KozmaBarrat}, users of social networks, etc.).
One can think of cellular automata, lattice-gas automata, parallel graph rewriting, causal graph dynamics~\cite{ArrighiIC,ArrighiRCGD,ArrighiBRCGD} or so-called global transformations~\cite{DBLP:conf/gg/MaignanS15,DBLP:conf/mfcs/FernandezMS22,DBLP:conf/gg/MaignanS24} for instance.

{\em Synchronism} is often criticised however, on the basis that: 1/ In some contexts synchronization mechanisms are considered a costly resource, which hinders the use of parallelism to achieve high performance computing~\cite{AsynchronousSimsSync}. 2/ It is often dubbed as physically unrealistic. Some authors argue that nature has no central clock, and thus cannot apply the same rule everywhere at once~\cite{AsynchronousPhysical}. This is a fair point as relativistic physics clearly departs from the idea of a global time across the universe. In particular the `time covariance' symmetry entails that, in any real system, it is perfectly legitimate to evolve just a small region of space, whilst keeping the rest of it unchanged.

{\em Asynchronism}---the application of the local rule at totally arbitrary places, non-deterministically---fits this picture better
and is also well studied, especially when 
it is more compelling to leave the evaluation strategy under-determined.
This is typically the case in rewriting theory, when the rewrite rules arise as a computationally-oriented version of an equality, e.g. $1+1\to 2$. Then, term $1+1+1$ may evolve into $2+1$, but it may also evolve into $1+2$, non-deterministically. This feels right, because: 1/ an underlying symmetry tells us there is no reason to favour one over the other and 2/ ultimately, if what matters is the end result, we are reassured by the fact that $2+1\to 3$ and $1+2\to 3$.

More generally, {\em confluence} of a set of rewrite rules is the property that if the evolutions $a\to^* b$ and $a\to^* c$ are possible, then there is some $d$ such that the evolutions $b\to^* d$ and $c\to^* d$ are also possible. In the context of terminating computation, this ensures that the non-determinism introduced by the order of application of the rewrite rules, will not matter to the end result, yielding a well-determined, unique normal form, e.g. $3$ in the above example.
Contributions \cite{StrongConfluenceEhrig,LocalConfluenceLambers,GraphRewritingConfluencePlump2,GraphRewritingConfluenceGaducci,GraphRewritingConfluenceJouannaud} deal with the case when the computational process eventually produces a graph as a result, through successive rewrites.

\subsection{Motivation}

\begin{figure}[t]
\centering
\begin{subfigure}{0.49\textwidth}
    \centering
    \resizebox{.5\textwidth}{!}{\tikzset{every picture/.style={line width=0.75pt}} 

\begin{tikzpicture}[x=0.75pt,y=0.75pt,yscale=-1,xscale=1]

\draw [color=myblack  ,draw opacity=1 ][line width=1.5]    (100,30) -- (200,30) ;
\draw [color=myblack  ,draw opacity=1 ][line width=1.5]    (300,30) -- (200,30) ;
\draw  [color=myblack  ,draw opacity=1 ][fill={rgb, 255:red, 255; green, 255; blue, 255 }  ,fill opacity=1 ][line width=1.5]  (90,30) .. controls (90,24.48) and (94.48,20) .. (100,20) .. controls (105.52,20) and (110,24.48) .. (110,30) .. controls (110,35.52) and (105.52,40) .. (100,40) .. controls (94.48,40) and (90,35.52) .. (90,30) -- cycle ;
\draw  [color=myblack  ,draw opacity=1 ][fill={rgb, 255:red, 255; green, 255; blue, 255 }  ,fill opacity=1 ][line width=1.5]  (290,30) .. controls (290,24.48) and (294.48,20) .. (300,20) .. controls (305.52,20) and (310,24.48) .. (310,30) .. controls (310,35.52) and (305.52,40) .. (300,40) .. controls (294.48,40) and (290,35.52) .. (290,30) -- cycle ;
\draw  [color=myblack  ,draw opacity=1 ][fill={rgb, 255:red, 255; green, 255; blue, 255 }  ,fill opacity=1 ][line width=1.5]  (190,30) .. controls (190,24.48) and (194.48,20) .. (200,20) .. controls (205.52,20) and (210,24.48) .. (210,30) .. controls (210,35.52) and (205.52,40) .. (200,40) .. controls (194.48,40) and (190,35.52) .. (190,30) -- cycle ;
\draw [color=border  ,draw opacity=1 ][line width=0.75]  [dash pattern={on 0.84pt off 2.51pt}]  (30,30) -- (90,30) ;
\draw [color=border  ,draw opacity=1 ][line width=0.75]  [dash pattern={on 0.84pt off 2.51pt}]  (310,30) -- (370,30) ;
\draw [line width=2.25]    (200,150) .. controls (170.88,123.18) and (169.72,108.29) .. (196.51,83.19) ;
\draw [shift={(200,80)}, rotate = 138.26] [fill=myblack  ][line width=0.08]  [draw opacity=0] (14.29,-6.86) -- (0,0) -- (14.29,6.86) -- cycle    ;
\draw [color=myblack  ,draw opacity=1 ][line width=1.5]    (100,180) -- (200,180) ;
\draw [color=myblack  ,draw opacity=1 ][line width=1.5]    (300,180) -- (200,180) ;
\draw  [color=myblack  ,draw opacity=1 ][fill={rgb, 255:red, 255; green, 255; blue, 255 }  ,fill opacity=1 ][line width=1.5]  (90,180) .. controls (90,174.48) and (94.48,170) .. (100,170) .. controls (105.52,170) and (110,174.48) .. (110,180) .. controls (110,185.52) and (105.52,190) .. (100,190) .. controls (94.48,190) and (90,185.52) .. (90,180) -- cycle ;
\draw  [color=myblack  ,draw opacity=1 ][fill={rgb, 255:red, 255; green, 255; blue, 255 }  ,fill opacity=1 ][line width=1.5]  (290,180) .. controls (290,174.48) and (294.48,170) .. (300,170) .. controls (305.52,170) and (310,174.48) .. (310,180) .. controls (310,185.52) and (305.52,190) .. (300,190) .. controls (294.48,190) and (290,185.52) .. (290,180) -- cycle ;
\draw  [color=myblack  ,draw opacity=1 ][fill={rgb, 255:red, 255; green, 255; blue, 255 }  ,fill opacity=1 ][line width=1.5]  (190,180) .. controls (190,174.48) and (194.48,170) .. (200,170) .. controls (205.52,170) and (210,174.48) .. (210,180) .. controls (210,185.52) and (205.52,190) .. (200,190) .. controls (194.48,190) and (190,185.52) .. (190,180) -- cycle ;
\draw [color=border  ,draw opacity=1 ][line width=0.75]  [dash pattern={on 0.84pt off 2.51pt}]  (30,180) -- (90,180) ;
\draw [color=border  ,draw opacity=1 ][line width=0.75]  [dash pattern={on 0.84pt off 2.51pt}]  (310,180) -- (370,180) ;
\draw  [color=myblack  ,draw opacity=1 ][fill=myblack  ,fill opacity=1 ][line width=1.5]  (200.56,190) .. controls (200.54,190) and (200.52,190) .. (200.5,190) .. controls (194.7,190) and (190,185.52) .. (190,180) .. controls (190,174.48) and (194.7,170) .. (200.5,170) .. controls (200.53,170) and (200.57,170) .. (200.6,170) -- cycle ;
\draw  [color=myblack  ,draw opacity=1 ][line width=1.5]  (200.6,170) .. controls (200.62,170) and (200.63,170) .. (200.65,170) .. controls (206.45,170) and (211.15,174.48) .. (211.15,180) .. controls (211.15,185.52) and (206.45,190) .. (200.65,190) .. controls (200.62,190) and (200.59,190) .. (200.56,190) -- cycle ;

\draw  [color=myblack  ,draw opacity=1 ][fill=myblack  ,fill opacity=1 ][line width=1.5]  (100.56,40) .. controls (100.54,40) and (100.52,40) .. (100.5,40) .. controls (94.7,40) and (90,35.52) .. (90,30) .. controls (90,24.48) and (94.7,20) .. (100.5,20) .. controls (100.53,20) and (100.57,20) .. (100.6,20) -- cycle ;
\draw  [color=myblack  ,draw opacity=1 ][line width=1.5]  (100.6,20) .. controls (100.62,20) and (100.63,20) .. (100.65,20) .. controls (106.45,20) and (111.15,24.48) .. (111.15,30) .. controls (111.15,35.52) and (106.45,40) .. (100.65,40) .. controls (100.62,40) and (100.59,40) .. (100.56,40) -- cycle ;

\draw (111,42) node [anchor=north west][inner sep=0.75pt]   [align=left] {\textbf{{\sizeFonex{x-1}}}};
\draw (211,42) node [anchor=north west][inner sep=0.75pt]   [align=left] {\textbf{{\sizeFonex{x}}}};
\draw (311,42) node [anchor=north west][inner sep=0.75pt]   [align=left] {\textbf{{\sizeFonex{x+1}}}};
\draw (211,92.4) node [anchor=north west][inner sep=0.75pt]  [font=\LARGE]  {\sizeFoneAx{$A_{x}$}};
\draw (111,192) node [anchor=north west][inner sep=0.75pt]   [align=left] {\textbf{{\sizeFonex{x-1}}}};
\draw (211,192) node [anchor=north west][inner sep=0.75pt]   [align=left] {\textbf{{\sizeFonex{x}}}};
\draw (311,192) node [anchor=north west][inner sep=0.75pt]   [align=left] {\textbf{{\sizeFonex{x+1}}}};

\end{tikzpicture}}\\
    {\footnotesize $A_x$ on a left moving particle.}
\end{subfigure}
\hfill
\begin{subfigure}{0.49\textwidth}
    \centering
    \resizebox{.5\textwidth}{!}{\tikzset{every picture/.style={line width=0.75pt}} 

\begin{tikzpicture}[x=0.75pt,y=0.75pt,yscale=-1,xscale=1]

\draw [color=myblack  ,draw opacity=1 ][line width=1.5]    (100,30) -- (200,30) ;
\draw [color=myblack  ,draw opacity=1 ][line width=1.5]    (300,30) -- (200,30) ;
\draw  [color=myblack  ,draw opacity=1 ][fill={rgb, 255:red, 255; green, 255; blue, 255 }  ,fill opacity=1 ][line width=1.5]  (90,30) .. controls (90,24.48) and (94.48,20) .. (100,20) .. controls (105.52,20) and (110,24.48) .. (110,30) .. controls (110,35.52) and (105.52,40) .. (100,40) .. controls (94.48,40) and (90,35.52) .. (90,30) -- cycle ;
\draw  [color=myblack  ,draw opacity=1 ][fill={rgb, 255:red, 255; green, 255; blue, 255 }  ,fill opacity=1 ][line width=1.5]  (290,30) .. controls (290,24.48) and (294.48,20) .. (300,20) .. controls (305.52,20) and (310,24.48) .. (310,30) .. controls (310,35.52) and (305.52,40) .. (300,40) .. controls (294.48,40) and (290,35.52) .. (290,30) -- cycle ;
\draw  [color=myblack  ,draw opacity=1 ][fill={rgb, 255:red, 255; green, 255; blue, 255 }  ,fill opacity=1 ][line width=1.5]  (190,30) .. controls (190,24.48) and (194.48,20) .. (200,20) .. controls (205.52,20) and (210,24.48) .. (210,30) .. controls (210,35.52) and (205.52,40) .. (200,40) .. controls (194.48,40) and (190,35.52) .. (190,30) -- cycle ;
\draw [color=border  ,draw opacity=1 ][line width=0.75]  [dash pattern={on 0.84pt off 2.51pt}]  (30,30) -- (90,30) ;
\draw [color=border  ,draw opacity=1 ][line width=0.75]  [dash pattern={on 0.84pt off 2.51pt}]  (310,30) -- (370,30) ;
\draw [color=myblack  ,draw opacity=1 ][line width=1.5]    (100,175) -- (200,175) ;
\draw [color=myblack  ,draw opacity=1 ][line width=1.5]    (300,175) -- (200,175) ;
\draw  [color=myblack  ,draw opacity=1 ][fill={rgb, 255:red, 255; green, 255; blue, 255 }  ,fill opacity=1 ][line width=1.5]  (90,175) .. controls (90,169.48) and (94.48,165) .. (100,165) .. controls (105.52,165) and (110,169.48) .. (110,175) .. controls (110,180.52) and (105.52,185) .. (100,185) .. controls (94.48,185) and (90,180.52) .. (90,175) -- cycle ;
\draw  [color=myblack  ,draw opacity=1 ][fill={rgb, 255:red, 255; green, 255; blue, 255 }  ,fill opacity=1 ][line width=1.5]  (290,175) .. controls (290,169.48) and (294.48,165) .. (300,165) .. controls (305.52,165) and (310,169.48) .. (310,175) .. controls (310,180.52) and (305.52,185) .. (300,185) .. controls (294.48,185) and (290,180.52) .. (290,175) -- cycle ;
\draw  [color=myblack  ,draw opacity=1 ][fill={rgb, 255:red, 255; green, 255; blue, 255 }  ,fill opacity=1 ][line width=1.5]  (190,175) .. controls (190,169.48) and (194.48,165) .. (200,165) .. controls (205.52,165) and (210,169.48) .. (210,175) .. controls (210,180.52) and (205.52,185) .. (200,185) .. controls (194.48,185) and (190,180.52) .. (190,175) -- cycle ;
\draw [color=border  ,draw opacity=1 ][line width=0.75]  [dash pattern={on 0.84pt off 2.51pt}]  (30,175) -- (90,175) ;
\draw [color=border  ,draw opacity=1 ][line width=0.75]  [dash pattern={on 0.84pt off 2.51pt}]  (310,175) -- (370,175) ;
\draw  [color=myblack  ,draw opacity=1 ][line width=1.5]  (200.56,185) .. controls (200.54,185) and (200.52,185) .. (200.5,185) .. controls (194.7,185) and (190,180.52) .. (190,175) .. controls (190,169.48) and (194.7,165) .. (200.5,165) .. controls (200.53,165) and (200.57,165) .. (200.6,165) -- cycle ;
\draw  [color=myblack  ,draw opacity=1 ][fill=myblack  ,fill opacity=1 ][line width=1.5]  (201.44,165) .. controls (201.46,165) and (201.48,165) .. (201.5,165) .. controls (207.3,165) and (212,169.48) .. (212,175) .. controls (212,180.52) and (207.3,185) .. (201.5,185) .. controls (201.47,185) and (201.43,185) .. (201.4,185) -- cycle ;

\draw  [color=myblack  ,draw opacity=1 ][line width=1.5]  (300.56,40) .. controls (300.54,40) and (300.52,40) .. (300.5,40) .. controls (294.7,40) and (290,35.52) .. (290,30) .. controls (290,24.48) and (294.7,20) .. (300.5,20) .. controls (300.53,20) and (300.57,20) .. (300.6,20) -- cycle ;
\draw  [color=myblack  ,draw opacity=1 ][fill=myblack  ,fill opacity=1 ][line width=1.5]  (301.44,20) .. controls (301.46,20) and (301.48,20) .. (301.5,20) .. controls (307.3,20) and (312,24.48) .. (312,30) .. controls (312,35.52) and (307.3,40) .. (301.5,40) .. controls (301.47,40) and (301.43,40) .. (301.4,40) -- cycle ;

\draw [line width=2.25]    (192.75,150) .. controls (163.63,123.18) and (162.47,108.29) .. (189.26,83.19) ;
\draw [shift={(192.75,80)}, rotate = 138.26] [fill=myblack  ][line width=0.08]  [draw opacity=0] (14.29,-6.86) -- (0,0) -- (14.29,6.86) -- cycle    ;

\draw (111,42) node [anchor=north west][inner sep=0.75pt]   [align=left] {\textbf{{\sizeFonex{x-1}}}};
\draw (211,42) node [anchor=north west][inner sep=0.75pt]   [align=left] {\textbf{{\sizeFonex{x}}}};
\draw (311,42) node [anchor=north west][inner sep=0.75pt]   [align=left] {\textbf{{\sizeFonex{x+1}}}};
\draw (111,187) node [anchor=north west][inner sep=0.75pt]   [align=left] {\textbf{{\sizeFonex{x-1}}}};
\draw (211,187) node [anchor=north west][inner sep=0.75pt]   [align=left] {\textbf{{\sizeFonex{x}}}};
\draw (311,187) node [anchor=north west][inner sep=0.75pt]   [align=left] {\textbf{{\sizeFonex{x+1}}}};
\draw (203.75,92.4) node [anchor=north west][inner sep=0.75pt]  [font=\LARGE]  {\sizeFoneAx{$A_{x}$}};

\end{tikzpicture}}\\
    {\footnotesize $A_x$ on a right moving particle.}
\end{subfigure}

\begin{subfigure}{0.49\textwidth}
    \captionsetup{justification=centering}
    \resizebox{\textwidth}{!}{\tikzset{every picture/.style={line width=0.75pt}} 

\begin{tikzpicture}[x=0.75pt,y=0.75pt,yscale=-1,xscale=1]

\draw [color=myblack  ,draw opacity=1 ][line width=1.5]    (928.46,47.46) -- (808.14,47.46) ;
\draw [color=border  ,draw opacity=1 ][line width=0.75]  [dash pattern={on 0.84pt off 2.51pt}]  (1060.81,47.46) -- (1133,47.46) ;
\draw [color=myblack  ,draw opacity=1 ][line width=1.5]    (1060.81,47.46) -- (940.49,47.46) ;
\draw [color=myblack  ,draw opacity=1 ][line width=1.5]    (86.22,47.46) -- (206.54,47.46) ;
\draw [color=myblack  ,draw opacity=1 ][line width=1.5]    (326.86,47.46) -- (206.54,47.46) ;
\draw [color=myblack  ,draw opacity=1 ][line width=1.5]    (326.86,47.46) -- (447.18,47.46) ;
\draw [color=myblack  ,draw opacity=1 ][line width=1.5]    (567.5,47.46) -- (447.18,47.46) ;
\draw [color=myblack  ,draw opacity=1 ][line width=1.5]    (567.5,47.46) -- (687.82,47.46) ;
\draw [color=myblack  ,draw opacity=1 ][line width=1.5]    (808.14,47.46) -- (687.82,47.46) ;
\draw [color=border  ,draw opacity=1 ][line width=0.75]  [dash pattern={on 0.84pt off 2.51pt}]  (2,47.46) -- (74.19,47.46) ;
\draw  [color=myblack  ,draw opacity=1 ][fill={rgb, 255:red, 255; green, 255; blue, 255 }  ,fill opacity=1 ][line width=1.5]  (310.29,47.46) .. controls (310.29,38.31) and (317.71,30.89) .. (326.86,30.89) .. controls (336.01,30.89) and (343.43,38.31) .. (343.43,47.46) .. controls (343.43,56.61) and (336.01,64.03) .. (326.86,64.03) .. controls (317.71,64.03) and (310.29,56.61) .. (310.29,47.46) -- cycle ;
\draw  [color=myblack  ,draw opacity=1 ][fill={rgb, 255:red, 255; green, 255; blue, 255 }  ,fill opacity=1 ][line width=1.5]  (69.65,47.46) .. controls (69.65,38.31) and (77.07,30.89) .. (86.22,30.89) .. controls (95.37,30.89) and (102.79,38.31) .. (102.79,47.46) .. controls (102.79,56.61) and (95.37,64.03) .. (86.22,64.03) .. controls (77.07,64.03) and (69.65,56.61) .. (69.65,47.46) -- cycle ;
\draw  [color=myblack  ,draw opacity=1 ][fill={rgb, 255:red, 255; green, 255; blue, 255 }  ,fill opacity=1 ][line width=1.5]  (189.97,47.46) .. controls (189.97,38.31) and (197.39,30.89) .. (206.54,30.89) .. controls (215.69,30.89) and (223.11,38.31) .. (223.11,47.46) .. controls (223.11,56.61) and (215.69,64.03) .. (206.54,64.03) .. controls (197.39,64.03) and (189.97,56.61) .. (189.97,47.46) -- cycle ;
\draw  [color=myblack  ,draw opacity=1 ][fill={rgb, 255:red, 255; green, 255; blue, 255 }  ,fill opacity=1 ][line width=1.5]  (430.61,47.46) .. controls (430.61,38.31) and (438.03,30.89) .. (447.18,30.89) .. controls (456.33,30.89) and (463.75,38.31) .. (463.75,47.46) .. controls (463.75,56.61) and (456.33,64.03) .. (447.18,64.03) .. controls (438.03,64.03) and (430.61,56.61) .. (430.61,47.46) -- cycle ;
\draw  [color=myblack  ,draw opacity=1 ][fill={rgb, 255:red, 255; green, 255; blue, 255 }  ,fill opacity=1 ][line width=1.5]  (550.93,47.46) .. controls (550.93,38.31) and (558.35,30.89) .. (567.5,30.89) .. controls (576.65,30.89) and (584.07,38.31) .. (584.07,47.46) .. controls (584.07,56.61) and (576.65,64.03) .. (567.5,64.03) .. controls (558.35,64.03) and (550.93,56.61) .. (550.93,47.46) -- cycle ;
\draw  [color=myblack  ,draw opacity=1 ][fill={rgb, 255:red, 255; green, 255; blue, 255 }  ,fill opacity=1 ][line width=1.5]  (671.25,47.46) .. controls (671.25,38.31) and (678.67,30.89) .. (687.82,30.89) .. controls (696.97,30.89) and (704.39,38.31) .. (704.39,47.46) .. controls (704.39,56.61) and (696.97,64.03) .. (687.82,64.03) .. controls (678.67,64.03) and (671.25,56.61) .. (671.25,47.46) -- cycle ;
\draw  [color=myblack  ,draw opacity=1 ][fill={rgb, 255:red, 255; green, 255; blue, 255 }  ,fill opacity=1 ][line width=1.5]  (791.57,47.46) .. controls (791.57,38.31) and (798.99,30.89) .. (808.14,30.89) .. controls (817.29,30.89) and (824.71,38.31) .. (824.71,47.46) .. controls (824.71,56.61) and (817.29,64.03) .. (808.14,64.03) .. controls (798.99,64.03) and (791.57,56.61) .. (791.57,47.46) -- cycle ;
\draw  [color=myblack  ,draw opacity=1 ][fill={rgb, 255:red, 255; green, 255; blue, 255 }  ,fill opacity=1 ][line width=1.5]  (911.89,47.46) .. controls (911.89,38.31) and (919.31,30.89) .. (928.46,30.89) .. controls (937.61,30.89) and (945.03,38.31) .. (945.03,47.46) .. controls (945.03,56.61) and (937.61,64.03) .. (928.46,64.03) .. controls (919.31,64.03) and (911.89,56.61) .. (911.89,47.46) -- cycle ;
\draw  [color=myblack  ,draw opacity=1 ][fill={rgb, 255:red, 255; green, 255; blue, 255 }  ,fill opacity=1 ][line width=1.5]  (1044.24,47.46) .. controls (1044.24,38.31) and (1051.66,30.89) .. (1060.81,30.89) .. controls (1069.96,30.89) and (1077.38,38.31) .. (1077.38,47.46) .. controls (1077.38,56.61) and (1069.96,64.03) .. (1060.81,64.03) .. controls (1051.66,64.03) and (1044.24,56.61) .. (1044.24,47.46) -- cycle ;
\draw  [color=myblack  ,draw opacity=1 ][fill=myblack  ,fill opacity=1 ][line width=1.5]  (928.48,64.2) .. controls (928.45,64.2) and (928.41,64.2) .. (928.38,64.2) .. controls (919.33,64.2) and (912,56.32) .. (912,46.6) .. controls (912,36.88) and (919.33,29) .. (928.38,29) .. controls (928.44,29) and (928.49,29) .. (928.55,29) -- cycle ;
\draw  [color=myblack  ,draw opacity=1 ][line width=1.5]  (928.52,29) .. controls (928.55,29) and (928.59,29) .. (928.62,29) .. controls (937.67,29) and (945,36.88) .. (945,46.6) .. controls (945,56.32) and (937.67,64.2) .. (928.62,64.2) .. controls (928.56,64.2) and (928.51,64.2) .. (928.45,64.2) -- cycle ;

\draw  [color=myblack  ,draw opacity=1 ][line width=1.5]  (206.8,63.2) .. controls (206.77,63.2) and (206.74,63.2) .. (206.7,63.2) .. controls (197.48,63.2) and (190,55.77) .. (190,46.6) .. controls (190,37.43) and (197.48,30) .. (206.7,30) .. controls (206.76,30) and (206.81,30) .. (206.87,30) -- cycle ;
\draw  [color=myblack  ,draw opacity=1 ][fill=myblack  ,fill opacity=1 ][line width=1.5]  (208.2,30) .. controls (208.23,30) and (208.26,30) .. (208.3,30) .. controls (217.52,30) and (225,37.43) .. (225,46.6) .. controls (225,55.77) and (217.52,63.2) .. (208.3,63.2) .. controls (208.24,63.2) and (208.19,63.2) .. (208.13,63.2) -- cycle ;

\draw (101,62) node [anchor=north west][inner sep=0.75pt]  [font=\Large] [align=left] {\textbf{{\sizeFone{0}}}};
\draw (221,65.22) node [anchor=north west][inner sep=0.75pt]  [font=\Large] [align=left] {\textbf{{\sizeFone{1}}}};
\draw (341.32,65.22) node [anchor=north west][inner sep=0.75pt]  [font=\Large] [align=left] {\textbf{{\sizeFone{2}}}};
\draw (461.64,65.22) node [anchor=north west][inner sep=0.75pt]  [font=\Large] [align=left] {\textbf{{\sizeFone{3}}}};
\draw (581.95,65.22) node [anchor=north west][inner sep=0.75pt]  [font=\Large] [align=left] {\textbf{{\sizeFone{4}}}};
\draw (702.27,65.22) node [anchor=north west][inner sep=0.75pt]  [font=\Large] [align=left] {\textbf{{\sizeFone{5}}}};
\draw (822.59,65.22) node [anchor=north west][inner sep=0.75pt]  [font=\Large] [align=left] {\textbf{{\sizeFone{6}}}};
\draw (942.91,65.22) node [anchor=north west][inner sep=0.75pt]  [font=\Large] [align=left] {\textbf{{\sizeFone{7}}}};
\draw (1063.23,65.22) node [anchor=north west][inner sep=0.75pt]  [font=\Large] [align=left] {\textbf{{\sizeFone{ 8}}}};

\end{tikzpicture}}
    \caption{$G$.}
\end{subfigure}
\hfill
\begin{subfigure}{0.49\textwidth}
    \captionsetup{justification=centering}
    \resizebox{\textwidth}{!}{\tikzset{every picture/.style={line width=0.75pt}} 

\begin{tikzpicture}[x=0.75pt,y=0.75pt,yscale=-1,xscale=1]

\draw [color=myblack  ,draw opacity=1 ][line width=1.5]    (928.46,47.46) -- (808.14,47.46) ;
\draw [color=border  ,draw opacity=1 ][line width=0.75]  [dash pattern={on 0.84pt off 2.51pt}]  (1060.81,47.46) -- (1133,47.46) ;
\draw [color=myblack  ,draw opacity=1 ][line width=1.5]    (1060.81,47.46) -- (940.49,47.46) ;
\draw [color=myblack  ,draw opacity=1 ][line width=1.5]    (86.22,47.46) -- (206.54,47.46) ;
\draw [color=myblack  ,draw opacity=1 ][line width=1.5]    (326.86,47.46) -- (206.54,47.46) ;
\draw [color=myblack  ,draw opacity=1 ][line width=1.5]    (326.86,47.46) -- (447.18,47.46) ;
\draw [color=myblack  ,draw opacity=1 ][line width=1.5]    (567.5,47.46) -- (447.18,47.46) ;
\draw [color=myblack  ,draw opacity=1 ][line width=1.5]    (567.5,47.46) -- (687.82,47.46) ;
\draw [color=myblack  ,draw opacity=1 ][line width=1.5]    (808.14,47.46) -- (687.82,47.46) ;
\draw [color=border  ,draw opacity=1 ][line width=0.75]  [dash pattern={on 0.84pt off 2.51pt}]  (2,47.46) -- (74.19,47.46) ;
\draw  [color=myblack  ,draw opacity=1 ][fill=myblack  ,fill opacity=1 ][line width=1.5]  (310.29,47.46) .. controls (310.29,38.31) and (317.71,30.89) .. (326.86,30.89) .. controls (336.01,30.89) and (343.43,38.31) .. (343.43,47.46) .. controls (343.43,56.61) and (336.01,64.03) .. (326.86,64.03) .. controls (317.71,64.03) and (310.29,56.61) .. (310.29,47.46) -- cycle ;
\draw  [color=myblack  ,draw opacity=1 ][fill={rgb, 255:red, 255; green, 255; blue, 255 }  ,fill opacity=1 ][line width=1.5]  (69.65,47.46) .. controls (69.65,38.31) and (77.07,30.89) .. (86.22,30.89) .. controls (95.37,30.89) and (102.79,38.31) .. (102.79,47.46) .. controls (102.79,56.61) and (95.37,64.03) .. (86.22,64.03) .. controls (77.07,64.03) and (69.65,56.61) .. (69.65,47.46) -- cycle ;
\draw  [color=myblack  ,draw opacity=1 ][fill={rgb, 255:red, 255; green, 255; blue, 255 }  ,fill opacity=1 ][line width=1.5]  (189.97,47.46) .. controls (189.97,38.31) and (197.39,30.89) .. (206.54,30.89) .. controls (215.69,30.89) and (223.11,38.31) .. (223.11,47.46) .. controls (223.11,56.61) and (215.69,64.03) .. (206.54,64.03) .. controls (197.39,64.03) and (189.97,56.61) .. (189.97,47.46) -- cycle ;
\draw  [color=myblack  ,draw opacity=1 ][fill={rgb, 255:red, 255; green, 255; blue, 255 }  ,fill opacity=1 ][line width=1.5]  (430.61,47.46) .. controls (430.61,38.31) and (438.03,30.89) .. (447.18,30.89) .. controls (456.33,30.89) and (463.75,38.31) .. (463.75,47.46) .. controls (463.75,56.61) and (456.33,64.03) .. (447.18,64.03) .. controls (438.03,64.03) and (430.61,56.61) .. (430.61,47.46) -- cycle ;
\draw  [color=myblack  ,draw opacity=1 ][fill={rgb, 255:red, 255; green, 255; blue, 255 }  ,fill opacity=1 ][line width=1.5]  (550.93,47.46) .. controls (550.93,38.31) and (558.35,30.89) .. (567.5,30.89) .. controls (576.65,30.89) and (584.07,38.31) .. (584.07,47.46) .. controls (584.07,56.61) and (576.65,64.03) .. (567.5,64.03) .. controls (558.35,64.03) and (550.93,56.61) .. (550.93,47.46) -- cycle ;
\draw  [color=myblack  ,draw opacity=1 ][fill={rgb, 255:red, 255; green, 255; blue, 255 }  ,fill opacity=1 ][line width=1.5]  (671.25,47.46) .. controls (671.25,38.31) and (678.67,30.89) .. (687.82,30.89) .. controls (696.97,30.89) and (704.39,38.31) .. (704.39,47.46) .. controls (704.39,56.61) and (696.97,64.03) .. (687.82,64.03) .. controls (678.67,64.03) and (671.25,56.61) .. (671.25,47.46) -- cycle ;
\draw  [color=myblack  ,draw opacity=1 ][fill={rgb, 255:red, 255; green, 255; blue, 255 }  ,fill opacity=1 ][line width=1.5]  (791.57,47.46) .. controls (791.57,38.31) and (798.99,30.89) .. (808.14,30.89) .. controls (817.29,30.89) and (824.71,38.31) .. (824.71,47.46) .. controls (824.71,56.61) and (817.29,64.03) .. (808.14,64.03) .. controls (798.99,64.03) and (791.57,56.61) .. (791.57,47.46) -- cycle ;
\draw  [color=myblack  ,draw opacity=1 ][fill={rgb, 255:red, 255; green, 255; blue, 255 }  ,fill opacity=1 ][line width=1.5]  (911.89,47.46) .. controls (911.89,38.31) and (919.31,30.89) .. (928.46,30.89) .. controls (937.61,30.89) and (945.03,38.31) .. (945.03,47.46) .. controls (945.03,56.61) and (937.61,64.03) .. (928.46,64.03) .. controls (919.31,64.03) and (911.89,56.61) .. (911.89,47.46) -- cycle ;
\draw  [color=myblack  ,draw opacity=1 ][fill={rgb, 255:red, 255; green, 255; blue, 255 }  ,fill opacity=1 ][line width=1.5]  (1044.24,47.46) .. controls (1044.24,38.31) and (1051.66,30.89) .. (1060.81,30.89) .. controls (1069.96,30.89) and (1077.38,38.31) .. (1077.38,47.46) .. controls (1077.38,56.61) and (1069.96,64.03) .. (1060.81,64.03) .. controls (1051.66,64.03) and (1044.24,56.61) .. (1044.24,47.46) -- cycle ;

\draw (101,62) node [anchor=north west][inner sep=0.75pt]  [font=\Large] [align=left] {\textbf{{\sizeFone{0}}}};
\draw (221,65.22) node [anchor=north west][inner sep=0.75pt]  [font=\Large] [align=left] {\textbf{{\sizeFone{1}}}};
\draw (341.32,65.22) node [anchor=north west][inner sep=0.75pt]  [font=\Large] [align=left] {\textbf{{\sizeFone{2}}}};
\draw (461.64,65.22) node [anchor=north west][inner sep=0.75pt]  [font=\Large] [align=left] {\textbf{{\sizeFone{3}}}};
\draw (581.95,65.22) node [anchor=north west][inner sep=0.75pt]  [font=\Large] [align=left] {\textbf{{\sizeFone{4}}}};
\draw (702.27,65.22) node [anchor=north west][inner sep=0.75pt]  [font=\Large] [align=left] {\textbf{{\sizeFone{5}}}};
\draw (822.59,65.22) node [anchor=north west][inner sep=0.75pt]  [font=\Large] [align=left] {\textbf{{\sizeFone{6}}}};
\draw (942.91,65.22) node [anchor=north west][inner sep=0.75pt]  [font=\Large] [align=left] {\textbf{{\sizeFone{7}}}};
\draw (1063.23,65.22) node [anchor=north west][inner sep=0.75pt]  [font=\Large] [align=left] {\textbf{{\sizeFone{8}}}};

\end{tikzpicture}}
    \caption{$A_{345671} G$.}
\end{subfigure}

\begin{subfigure}{0.49\textwidth}
    \captionsetup{justification=centering}
    \resizebox{\textwidth}{!}{\tikzset{every picture/.style={line width=0.75pt}} 

\begin{tikzpicture}[x=0.75pt,y=0.75pt,yscale=-1,xscale=1]

\draw [color=myblack  ,draw opacity=1 ][line width=1.5]    (928.46,47.46) -- (808.14,47.46) ;
\draw [color=border  ,draw opacity=1 ][line width=0.75]  [dash pattern={on 0.84pt off 2.51pt}]  (1060.81,47.46) -- (1133,47.46) ;
\draw [color=myblack  ,draw opacity=1 ][line width=1.5]    (1060.81,47.46) -- (940.49,47.46) ;
\draw [color=myblack  ,draw opacity=1 ][line width=1.5]    (86.22,47.46) -- (206.54,47.46) ;
\draw [color=myblack  ,draw opacity=1 ][line width=1.5]    (326.86,47.46) -- (206.54,47.46) ;
\draw [color=myblack  ,draw opacity=1 ][line width=1.5]    (326.86,47.46) -- (447.18,47.46) ;
\draw [color=myblack  ,draw opacity=1 ][line width=1.5]    (567.5,47.46) -- (447.18,47.46) ;
\draw [color=myblack  ,draw opacity=1 ][line width=1.5]    (567.5,47.46) -- (687.82,47.46) ;
\draw [color=myblack  ,draw opacity=1 ][line width=1.5]    (808.14,47.46) -- (687.82,47.46) ;
\draw [color=border  ,draw opacity=1 ][line width=0.75]  [dash pattern={on 0.84pt off 2.51pt}]  (2,47.46) -- (74.19,47.46) ;
\draw  [color=myblack  ,draw opacity=1 ][fill={rgb, 255:red, 255; green, 255; blue, 255 }  ,fill opacity=1 ][line width=1.5]  (310.29,47.46) .. controls (310.29,38.31) and (317.71,30.89) .. (326.86,30.89) .. controls (336.01,30.89) and (343.43,38.31) .. (343.43,47.46) .. controls (343.43,56.61) and (336.01,64.03) .. (326.86,64.03) .. controls (317.71,64.03) and (310.29,56.61) .. (310.29,47.46) -- cycle ;
\draw  [color=myblack  ,draw opacity=1 ][fill={rgb, 255:red, 255; green, 255; blue, 255 }  ,fill opacity=1 ][line width=1.5]  (69.65,47.46) .. controls (69.65,38.31) and (77.07,30.89) .. (86.22,30.89) .. controls (95.37,30.89) and (102.79,38.31) .. (102.79,47.46) .. controls (102.79,56.61) and (95.37,64.03) .. (86.22,64.03) .. controls (77.07,64.03) and (69.65,56.61) .. (69.65,47.46) -- cycle ;
\draw  [color=myblack  ,draw opacity=1 ][fill={rgb, 255:red, 255; green, 255; blue, 255 }  ,fill opacity=1 ][line width=1.5]  (189.97,47.46) .. controls (189.97,38.31) and (197.39,30.89) .. (206.54,30.89) .. controls (215.69,30.89) and (223.11,38.31) .. (223.11,47.46) .. controls (223.11,56.61) and (215.69,64.03) .. (206.54,64.03) .. controls (197.39,64.03) and (189.97,56.61) .. (189.97,47.46) -- cycle ;
\draw  [color=myblack  ,draw opacity=1 ][fill={rgb, 255:red, 255; green, 255; blue, 255 }  ,fill opacity=1 ][line width=1.5]  (430.61,47.46) .. controls (430.61,38.31) and (438.03,30.89) .. (447.18,30.89) .. controls (456.33,30.89) and (463.75,38.31) .. (463.75,47.46) .. controls (463.75,56.61) and (456.33,64.03) .. (447.18,64.03) .. controls (438.03,64.03) and (430.61,56.61) .. (430.61,47.46) -- cycle ;
\draw  [color=myblack  ,draw opacity=1 ][fill={rgb, 255:red, 255; green, 255; blue, 255 }  ,fill opacity=1 ][line width=1.5]  (550.93,47.46) .. controls (550.93,38.31) and (558.35,30.89) .. (567.5,30.89) .. controls (576.65,30.89) and (584.07,38.31) .. (584.07,47.46) .. controls (584.07,56.61) and (576.65,64.03) .. (567.5,64.03) .. controls (558.35,64.03) and (550.93,56.61) .. (550.93,47.46) -- cycle ;
\draw  [color=myblack  ,draw opacity=1 ][fill={rgb, 255:red, 255; green, 255; blue, 255 }  ,fill opacity=1 ][line width=1.5]  (671.25,47.46) .. controls (671.25,38.31) and (678.67,30.89) .. (687.82,30.89) .. controls (696.97,30.89) and (704.39,38.31) .. (704.39,47.46) .. controls (704.39,56.61) and (696.97,64.03) .. (687.82,64.03) .. controls (678.67,64.03) and (671.25,56.61) .. (671.25,47.46) -- cycle ;
\draw  [color=myblack  ,draw opacity=1 ][fill=myblack  ,fill opacity=1 ][line width=1.5]  (791.57,47.46) .. controls (791.57,38.31) and (798.99,30.89) .. (808.14,30.89) .. controls (817.29,30.89) and (824.71,38.31) .. (824.71,47.46) .. controls (824.71,56.61) and (817.29,64.03) .. (808.14,64.03) .. controls (798.99,64.03) and (791.57,56.61) .. (791.57,47.46) -- cycle ;
\draw  [color=myblack  ,draw opacity=1 ][fill={rgb, 255:red, 255; green, 255; blue, 255 }  ,fill opacity=1 ][line width=1.5]  (911.89,47.46) .. controls (911.89,38.31) and (919.31,30.89) .. (928.46,30.89) .. controls (937.61,30.89) and (945.03,38.31) .. (945.03,47.46) .. controls (945.03,56.61) and (937.61,64.03) .. (928.46,64.03) .. controls (919.31,64.03) and (911.89,56.61) .. (911.89,47.46) -- cycle ;
\draw  [color=myblack  ,draw opacity=1 ][fill={rgb, 255:red, 255; green, 255; blue, 255 }  ,fill opacity=1 ][line width=1.5]  (1044.24,47.46) .. controls (1044.24,38.31) and (1051.66,30.89) .. (1060.81,30.89) .. controls (1069.96,30.89) and (1077.38,38.31) .. (1077.38,47.46) .. controls (1077.38,56.61) and (1069.96,64.03) .. (1060.81,64.03) .. controls (1051.66,64.03) and (1044.24,56.61) .. (1044.24,47.46) -- cycle ;

\draw (101,62) node [anchor=north west][inner sep=0.75pt]  [font=\Large] [align=left] {\textbf{{\sizeFone{0}}}};
\draw (221,65.22) node [anchor=north west][inner sep=0.75pt]  [font=\Large] [align=left] {\textbf{{\sizeFone{1}}}};
\draw (341.32,65.22) node [anchor=north west][inner sep=0.75pt]  [font=\Large] [align=left] {\textbf{{\sizeFone{2}}}};
\draw (461.64,65.22) node [anchor=north west][inner sep=0.75pt]  [font=\Large] [align=left] {\textbf{{\sizeFone{3}}}};
\draw (581.95,65.22) node [anchor=north west][inner sep=0.75pt]  [font=\Large] [align=left] {\textbf{{\sizeFone{4}}}};
\draw (702.27,65.22) node [anchor=north west][inner sep=0.75pt]  [font=\Large] [align=left] {\textbf{{\sizeFone{5}}}};
\draw (822.59,65.22) node [anchor=north west][inner sep=0.75pt]  [font=\Large] [align=left] {\textbf{{\sizeFone{6}}}};
\draw (942.91,65.22) node [anchor=north west][inner sep=0.75pt]  [font=\Large] [align=left] {\textbf{{\sizeFone{7}}}};
\draw (1063.23,65.22) node [anchor=north west][inner sep=0.75pt]  [font=\Large] [align=left] {\textbf{{\sizeFone{8}}}};

\end{tikzpicture}}
    \caption{$A_{543217} G$.}
\end{subfigure}
\hfill
\begin{subfigure}{0,49\textwidth}
    \captionsetup{justification=centering}
    \resizebox{\textwidth}{!}{\tikzset{every picture/.style={line width=0.75pt}} 

\begin{tikzpicture}[x=0.75pt,y=0.75pt,yscale=-1,xscale=1]

\draw [color=myblack  ,draw opacity=1 ][line width=1.5]    (928.46,47.46) -- (808.14,47.46) ;
\draw [color=border  ,draw opacity=1 ][line width=0.75]  [dash pattern={on 0.84pt off 2.51pt}]  (1060.81,47.46) -- (1133,47.46) ;
\draw [color=myblack  ,draw opacity=1 ][line width=1.5]    (1060.81,47.46) -- (940.49,47.46) ;
\draw [color=myblack  ,draw opacity=1 ][line width=1.5]    (86.22,47.46) -- (206.54,47.46) ;
\draw [color=myblack  ,draw opacity=1 ][line width=1.5]    (326.86,47.46) -- (206.54,47.46) ;
\draw [color=myblack  ,draw opacity=1 ][line width=1.5]    (326.86,47.46) -- (447.18,47.46) ;
\draw [color=myblack  ,draw opacity=1 ][line width=1.5]    (567.5,47.46) -- (447.18,47.46) ;
\draw [color=myblack  ,draw opacity=1 ][line width=1.5]    (567.5,47.46) -- (687.82,47.46) ;
\draw [color=myblack  ,draw opacity=1 ][line width=1.5]    (808.14,47.46) -- (687.82,47.46) ;
\draw [color=border  ,draw opacity=1 ][line width=0.75]  [dash pattern={on 0.84pt off 2.51pt}]  (2,47.46) -- (74.19,47.46) ;
\draw  [color=myblack  ,draw opacity=1 ][fill={rgb, 255:red, 255; green, 255; blue, 255 }  ,fill opacity=1 ][line width=1.5]  (310.29,47.46) .. controls (310.29,38.31) and (317.71,30.89) .. (326.86,30.89) .. controls (336.01,30.89) and (343.43,38.31) .. (343.43,47.46) .. controls (343.43,56.61) and (336.01,64.03) .. (326.86,64.03) .. controls (317.71,64.03) and (310.29,56.61) .. (310.29,47.46) -- cycle ;
\draw  [color=myblack  ,draw opacity=1 ][fill={rgb, 255:red, 255; green, 255; blue, 255 }  ,fill opacity=1 ][line width=1.5]  (69.65,47.46) .. controls (69.65,38.31) and (77.07,30.89) .. (86.22,30.89) .. controls (95.37,30.89) and (102.79,38.31) .. (102.79,47.46) .. controls (102.79,56.61) and (95.37,64.03) .. (86.22,64.03) .. controls (77.07,64.03) and (69.65,56.61) .. (69.65,47.46) -- cycle ;
\draw  [color=myblack  ,draw opacity=1 ][fill={rgb, 255:red, 255; green, 255; blue, 255 }  ,fill opacity=1 ][line width=1.5]  (189.97,47.46) .. controls (189.97,38.31) and (197.39,30.89) .. (206.54,30.89) .. controls (215.69,30.89) and (223.11,38.31) .. (223.11,47.46) .. controls (223.11,56.61) and (215.69,64.03) .. (206.54,64.03) .. controls (197.39,64.03) and (189.97,56.61) .. (189.97,47.46) -- cycle ;
\draw  [color=myblack  ,draw opacity=1 ][fill={rgb, 255:red, 255; green, 255; blue, 255 }  ,fill opacity=1 ][line width=1.5]  (430.61,47.46) .. controls (430.61,38.31) and (438.03,30.89) .. (447.18,30.89) .. controls (456.33,30.89) and (463.75,38.31) .. (463.75,47.46) .. controls (463.75,56.61) and (456.33,64.03) .. (447.18,64.03) .. controls (438.03,64.03) and (430.61,56.61) .. (430.61,47.46) -- cycle ;
\draw  [color=myblack  ,draw opacity=1 ][fill={rgb, 255:red, 255; green, 255; blue, 255 }  ,fill opacity=1 ][line width=1.5]  (550.93,47.46) .. controls (550.93,38.31) and (558.35,30.89) .. (567.5,30.89) .. controls (576.65,30.89) and (584.07,38.31) .. (584.07,47.46) .. controls (584.07,56.61) and (576.65,64.03) .. (567.5,64.03) .. controls (558.35,64.03) and (550.93,56.61) .. (550.93,47.46) -- cycle ;
\draw  [color=myblack  ,draw opacity=1 ][fill={rgb, 255:red, 255; green, 255; blue, 255 }  ,fill opacity=1 ][line width=1.5]  (671.25,47.46) .. controls (671.25,38.31) and (678.67,30.89) .. (687.82,30.89) .. controls (696.97,30.89) and (704.39,38.31) .. (704.39,47.46) .. controls (704.39,56.61) and (696.97,64.03) .. (687.82,64.03) .. controls (678.67,64.03) and (671.25,56.61) .. (671.25,47.46) -- cycle ;
\draw  [color=myblack  ,draw opacity=1 ][fill={rgb, 255:red, 255; green, 255; blue, 255 }  ,fill opacity=1 ][line width=1.5]  (791.57,47.46) .. controls (791.57,38.31) and (798.99,30.89) .. (808.14,30.89) .. controls (817.29,30.89) and (824.71,38.31) .. (824.71,47.46) .. controls (824.71,56.61) and (817.29,64.03) .. (808.14,64.03) .. controls (798.99,64.03) and (791.57,56.61) .. (791.57,47.46) -- cycle ;
\draw  [color=myblack  ,draw opacity=1 ][fill={rgb, 255:red, 255; green, 255; blue, 255 }  ,fill opacity=1 ][line width=1.5]  (911.89,47.46) .. controls (911.89,38.31) and (919.31,30.89) .. (928.46,30.89) .. controls (937.61,30.89) and (945.03,38.31) .. (945.03,47.46) .. controls (945.03,56.61) and (937.61,64.03) .. (928.46,64.03) .. controls (919.31,64.03) and (911.89,56.61) .. (911.89,47.46) -- cycle ;
\draw  [color=myblack  ,draw opacity=1 ][fill={rgb, 255:red, 255; green, 255; blue, 255 }  ,fill opacity=1 ][line width=1.5]  (1044.24,47.46) .. controls (1044.24,38.31) and (1051.66,30.89) .. (1060.81,30.89) .. controls (1069.96,30.89) and (1077.38,38.31) .. (1077.38,47.46) .. controls (1077.38,56.61) and (1069.96,64.03) .. (1060.81,64.03) .. controls (1051.66,64.03) and (1044.24,56.61) .. (1044.24,47.46) -- cycle ;
\draw  [color=myblack  ,draw opacity=1 ][fill=myblack  ,fill opacity=1 ][line width=1.5]  (206.48,65.2) .. controls (206.45,65.2) and (206.41,65.2) .. (206.38,65.2) .. controls (197.33,65.2) and (190,57.32) .. (190,47.6) .. controls (190,37.88) and (197.33,30) .. (206.38,30) .. controls (206.44,30) and (206.49,30) .. (206.55,30) -- cycle ;
\draw  [color=myblack  ,draw opacity=1 ][line width=1.5]  (206.52,30) .. controls (206.55,30) and (206.59,30) .. (206.62,30) .. controls (215.67,30) and (223,37.88) .. (223,47.6) .. controls (223,57.32) and (215.67,65.2) .. (206.62,65.2) .. controls (206.56,65.2) and (206.51,65.2) .. (206.45,65.2) -- cycle ;

\draw  [color=myblack  ,draw opacity=1 ][line width=1.5]  (927.8,64.2) .. controls (927.77,64.2) and (927.74,64.2) .. (927.7,64.2) .. controls (918.48,64.2) and (911,56.77) .. (911,47.6) .. controls (911,38.43) and (918.48,31) .. (927.7,31) .. controls (927.76,31) and (927.81,31) .. (927.87,31) -- cycle ;
\draw  [color=myblack  ,draw opacity=1 ][fill=myblack  ,fill opacity=1 ][line width=1.5]  (929.2,31) .. controls (929.23,31) and (929.26,31) .. (929.3,31) .. controls (938.52,31) and (946,38.43) .. (946,47.6) .. controls (946,56.77) and (938.52,64.2) .. (929.3,64.2) .. controls (929.24,64.2) and (929.19,64.2) .. (929.13,64.2) -- cycle ;

\draw (101,62) node [anchor=north west][inner sep=0.75pt]  [font=\Large] [align=left] {\textbf{{\sizeFone{0}}}};
\draw (221,65.22) node [anchor=north west][inner sep=0.75pt]  [font=\Large] [align=left] {\textbf{{\sizeFone{1}}}};
\draw (341.32,65.22) node [anchor=north west][inner sep=0.75pt]  [font=\Large] [align=left] {\textbf{{\sizeFone{2}}}};
\draw (461.64,65.22) node [anchor=north west][inner sep=0.75pt]  [font=\Large] [align=left] {\textbf{{\sizeFone{3}}}};
\draw (581.95,65.22) node [anchor=north west][inner sep=0.75pt]  [font=\Large] [align=left] {\textbf{{\sizeFone{4}}}};
\draw (702.27,65.22) node [anchor=north west][inner sep=0.75pt]  [font=\Large] [align=left] {\textbf{{\sizeFone{5}}}};
\draw (822.59,65.22) node [anchor=north west][inner sep=0.75pt]  [font=\Large] [align=left] {\textbf{{\sizeFone{6}}}};
\draw (942.91,65.22) node [anchor=north west][inner sep=0.75pt]  [font=\Large] [align=left] {\textbf{{\sizeFone{7}}}};
\draw (1063.23,65.22) node [anchor=north west][inner sep=0.75pt]  [font=\Large] [align=left] {\textbf{{\sizeFone{8}}}};

\end{tikzpicture}}
    \caption{$A_{65432}A_{345671} G = A_{23456}A_{345671} G$.}
\end{subfigure}
\caption{{\em Confluence yet inconsistencies.} The local rule transports right-moving particles to the right and left-moving particles to the left, without interactions. $(a)$ We start with a left-moving particle in $7$ and a right-moving particle in $1$. $(b)\& (c)$ The point of collision between the two particles is not well defined, it depends on the evaluation strategy. Here $A_{71}G$ is short for $A_7(A_1 G)$. $(d)$. Still, the system is confluent, as the divergent configurations can be both evolve into the last.}
\label{fig : faster than light}
\end{figure}

We are interested in computational processes that generate not just one result, but a non-terminating succession of them.
Unfortunately, confluence alone hardly provides any guarantee in such cases as, for instance, rewrite rules yielding every possible result are computationally meaningless but trivially proven confluent. The use of labelling techniques \cite{LevyLabels} aiming at safeguarding the successive results from being rewritten can still save the day at this stage. But things will go out of hand if we are interested in modelling not just one computational process, but an entire network of interacting processes, each generating results that cannot be safeguarded, because they are being consumed by the neighbouring processes to produce their next result, etc., i.e. non-terminating distributed computation over linear resources, including dynamical systems, e.g. particle systems. For instance, Fig.~\ref{fig : faster than light} shows left and right-moving particles on a line, where left versus right is indicated by the ports along the edges. The local rule simply has them move, which means consuming the particle at vertex, to give it to another. A little thought shows that the system is confluent. Yet, depending upon the chosen order of applications, one finds that one particle propagates much faster than the other, yielding the collision to occur at very different places. Whilst this example is designed to emphasize these issues, it is clear that, when it comes to dynamical systems, asynchronous evaluation strategies may lead to inconsistent results, nonphysical effects (e.g. superluminal signalling), and pathological dynamical behaviours (e.g. over Boolean networks \cite{AsynchronousSene,AsynchronousHaar} and cellular automata \cite{AsynchronousRoos,AsynchronousSchabanel}). This is a major deterrent to their usage, and the notion of confluence is no fix to that.

For this purpose, we introduce {\em space-time determinism} in two forms.
In both forms, the idea is to move beyond the notion of confluence, which traditionally compares configurations globally while ignoring their spatial and local structures
, and introduce a spatially and locally refined notion instead. It will ensure well-defined, interrelated events within a global ``space-time diagram'' (in the sense of cellular automata or Physics).
The first approach, referred to as {\em weak space-time determinism}, focuses on maintaining consistent a certain class of events---that can be thought of as ``results''---across all possible evaluation strategies.
The second approach, called {\em (full) space-time determinism}, considers a broader class of events instead---that can be thought of as ``partial results'' with respect to the previous class.
Intuitively, it asks that two snapshots of the evolving graph that have consumed the same information in some localized area, must agree on the partial results in this area regardless of the evaluation strategy.

{\em Dynamical geometry.} We do not restrict ourselves to work over a fixed lattice or a fixed boolean network, here. Our processes start with some given neighbours, but they can then connect with the neighbours of their neighbours, as well as disconnect. The DAG of dependencies is both a constraint upon the evolution, and a subject of the evolution, allowing us to express intriguing effects such as time dilation, reminiscent of general relativity, see Fig.~\ref{fig : Space time diagram example 2}.

{\em Applications.} Whilst mainly of theoretical nature, we hope to have convinced the reader that this work may contribute to the efficient parallel schemes for the implementation of dynamical system and distributed systems---doing away with any expensive clock synchronisation mechanisms and replacing it with a cheap DAG. This was not our original motivation: the applications we pursue lie at the crossroad with Physics, as we seek for a mathematically sound, constructive framework for discrete models of general relativity \cite{Sorkin}. From a general philosophy of science point of view, we find it compelling to reconcile asynchronism and determinism. 

{\em Plan of the paper.} Sec.~\ref{sec:graphs} makes specific the kind of DAG we use, the way we name each vertex, the fact that edges are between ports of vertices, and what is meant by a local rule. Sec.~\ref{sec:exV1} provides an example of how to perform the asynchronous simulation of a particle system, at the cost of introducing ``metric'' information in the form of these DAGs so as to capture the relative advancement of the computation in a region with respect to another. Sec.~\ref{sec: Simulation} generalizes this construct to simulate any cellular automaton. 
Sec.~\ref{sec:exV2} shows how, having introduced this DAG, one can manipulate it to achieve time dilation effects. An analogy is drawn with general relativity, where the concepts of time covariance, metric, background-independence and dynamical geometry, lead to time dilation.  Sec.~\ref{sec:consistency} introduces weak and full consistency as well as the corresponding theorems, constituting our main technical contributions. Sec.~\ref{sec:conclusion} summarizes the results, compares them to the related works, and provides some perspectives.

\section{Graphs and local rules}\label{sec:graphs}

\begin{figure}[b]
\hfil
\begin{subfigure}{0.25\textwidth}
    \captionsetup{justification=centering}
    \centering
    \resizebox{\textwidth}{!}{\tikzset{every picture/.style={line width=0.75pt}} 

\begin{tikzpicture}[x=0.75pt,y=0.75pt,yscale=-1,xscale=1]

\draw [line width=2.25]    (220,185) -- (155,185) ;
\draw [shift={(150,185)}, rotate = 360] [fill=myblack  ][line width=0.08]  [draw opacity=0] (14.29,-6.86) -- (0,0) -- (14.29,6.86) -- cycle    ;
\draw [line width=2.25]    (100,285) -- (165,285) ;
\draw [shift={(170,285)}, rotate = 180] [fill=myblack  ][line width=0.08]  [draw opacity=0] (14.29,-6.86) -- (0,0) -- (14.29,6.86) -- cycle    ;
\draw [color=border  ,draw opacity=1 ][line width=2.25]  [dash pattern={on 2.53pt off 3.02pt}]  (135,170) -- (87.91,104.07) ;
\draw [shift={(85,100)}, rotate = 54.46] [fill=border  ,fill opacity=1 ][line width=0.08]  [draw opacity=0] (14.29,-6.86) -- (0,0) -- (14.29,6.86) -- cycle    ;
\draw [line width=2.25]    (185,270) -- (137.91,204.07) ;
\draw [shift={(135,200)}, rotate = 54.46] [fill=myblack  ][line width=0.08]  [draw opacity=0] (14.29,-6.86) -- (0,0) -- (14.29,6.86) -- cycle    ;
\draw [color=border  ,draw opacity=1 ][line width=2.25]  [dash pattern={on 2.53pt off 3.02pt}]  (235,170) -- (187.91,104.07) ;
\draw [shift={(185,100)}, rotate = 54.46] [fill=border  ,fill opacity=1 ][line width=0.08]  [draw opacity=0] (14.29,-6.86) -- (0,0) -- (14.29,6.86) -- cycle    ;
\draw [line width=2.25]    (285,270) -- (237.91,204.07) ;
\draw [shift={(235,200)}, rotate = 54.46] [fill=myblack  ][line width=0.08]  [draw opacity=0] (14.29,-6.86) -- (0,0) -- (14.29,6.86) -- cycle    ;
\draw [line width=2.25]    (185,270) -- (232.09,204.07) ;
\draw [shift={(235,200)}, rotate = 125.54] [fill=myblack  ][line width=0.08]  [draw opacity=0] (14.29,-6.86) -- (0,0) -- (14.29,6.86) -- cycle    ;
\draw [line width=2.25]    (235,170) -- (282.09,104.07) ;
\draw [shift={(285,100)}, rotate = 125.54] [fill=myblack  ][line width=0.08]  [draw opacity=0] (14.29,-6.86) -- (0,0) -- (14.29,6.86) -- cycle    ;
\draw  [color=myblack  ,draw opacity=1 ][fill={rgb, 255:red, 255; green, 255; blue, 255 }  ,fill opacity=1 ][line width=3]  (170,285) .. controls (170,276.72) and (176.72,270) .. (185,270) .. controls (193.28,270) and (200,276.72) .. (200,285) .. controls (200,293.28) and (193.28,300) .. (185,300) .. controls (176.72,300) and (170,293.28) .. (170,285) -- cycle ;
\draw  [color=myblack  ,draw opacity=1 ][fill={rgb, 255:red, 255; green, 255; blue, 255 }  ,fill opacity=1 ][line width=3]  (70,285) .. controls (70,276.72) and (76.72,270) .. (85,270) .. controls (93.28,270) and (100,276.72) .. (100,285) .. controls (100,293.28) and (93.28,300) .. (85,300) .. controls (76.72,300) and (70,293.28) .. (70,285) -- cycle ;
\draw  [color=myblack  ,draw opacity=1 ][fill={rgb, 255:red, 255; green, 255; blue, 255 }  ,fill opacity=1 ][line width=3]  (270,285) .. controls (270,276.72) and (276.72,270) .. (285,270) .. controls (293.28,270) and (300,276.72) .. (300,285) .. controls (300,293.28) and (293.28,300) .. (285,300) .. controls (276.72,300) and (270,293.28) .. (270,285) -- cycle ;
\draw  [color=myblack  ,draw opacity=1 ][fill={rgb, 255:red, 255; green, 255; blue, 255 }  ,fill opacity=1 ][line width=3]  (120,185) .. controls (120,176.72) and (126.72,170) .. (135,170) .. controls (143.28,170) and (150,176.72) .. (150,185) .. controls (150,193.28) and (143.28,200) .. (135,200) .. controls (126.72,200) and (120,193.28) .. (120,185) -- cycle ;
\draw  [color=myblack  ,draw opacity=1 ][fill={rgb, 255:red, 255; green, 255; blue, 255 }  ,fill opacity=1 ][line width=3]  (220,185) .. controls (220,176.72) and (226.72,170) .. (235,170) .. controls (243.28,170) and (250,176.72) .. (250,185) .. controls (250,193.28) and (243.28,200) .. (235,200) .. controls (226.72,200) and (220,193.28) .. (220,185) -- cycle ;
\draw  [color=border  ,draw opacity=1 ][fill={rgb, 255:red, 255; green, 255; blue, 255 }  ,fill opacity=1 ][dash pattern={on 3.38pt off 3.27pt}][line width=3]  (170,85) .. controls (170,76.72) and (176.72,70) .. (185,70) .. controls (193.28,70) and (200,76.72) .. (200,85) .. controls (200,93.28) and (193.28,100) .. (185,100) .. controls (176.72,100) and (170,93.28) .. (170,85) -- cycle ;
\draw  [color=border  ,draw opacity=1 ][fill={rgb, 255:red, 255; green, 255; blue, 255 }  ,fill opacity=1 ][dash pattern={on 3.38pt off 3.27pt}][line width=3]  (70,85) .. controls (70,76.72) and (76.72,70) .. (85,70) .. controls (93.28,70) and (100,76.72) .. (100,85) .. controls (100,93.28) and (93.28,100) .. (85,100) .. controls (76.72,100) and (70,93.28) .. (70,85) -- cycle ;
\draw  [color=myblack  ,draw opacity=1 ][fill={rgb, 255:red, 255; green, 255; blue, 255 }  ,fill opacity=1 ][line width=3]  (270,85) .. controls (270,76.72) and (276.72,70) .. (285,70) .. controls (293.28,70) and (300,76.72) .. (300,85) .. controls (300,93.28) and (293.28,100) .. (285,100) .. controls (276.72,100) and (270,93.28) .. (270,85) -- cycle ;

\draw (81,242.4) node [anchor=north west][inner sep=0.75pt]  [font=\sizeFtwo]  {$\mathbf{t_{0} .x_{0}}$};
\draw (201,252.4) node [anchor=north west][inner sep=0.75pt]  [font=\sizeFtwo]  {$\mathbf{t_{1} .x_{1}}$};
\draw (301,252.4) node [anchor=north west][inner sep=0.75pt]  [font=\sizeFtwo]  {$\mathbf{t_{2} .x_{2}}$};
\draw (141,145.4) node [anchor=north west][inner sep=0.75pt]  [font=\sizeFtwo]  {$\mathbf{t_{4} .x_{4}}$};
\draw (261,172.4) node [anchor=north west][inner sep=0.75pt]  [font=\sizeFtwo]  {$\mathbf{t_{3} .x_{3}}$};
\draw (102,88.4) node [anchor=north west][inner sep=0.75pt]  [font=\sizeFtwo]  {$\mathbf{t_{5} .x_{5}}$};
\draw (202,88.4) node [anchor=north west][inner sep=0.75pt]  [font=\sizeFtwo]  {$\mathbf{t_{6} .x_{6}}$};
\draw (301,95.4) node [anchor=north west][inner sep=0.75pt]  [font=\sizeFtwo]  {$\mathbf{t_{7} .x_{7}}$};
\draw (279,242.4) node [anchor=north west][inner sep=0.75pt]  [font=\sizeFtwop,color={rgb, 255:red, 155; green, 155; blue, 155 }  ,opacity=1 ]  {$:a$};
\draw (201,242.4) node [anchor=north west][inner sep=0.75pt]  [font=\sizeFtwop,color={rgb, 255:red, 155; green, 155; blue, 155 }  ,opacity=1 ]  {$:a$};
\draw (102,288.4) node [anchor=north west][inner sep=0.75pt]  [font=\sizeFtwop,color={rgb, 255:red, 155; green, 155; blue, 155 }  ,opacity=1 ]  {$:a$};
\draw (99,152.4) node [anchor=north west][inner sep=0.75pt]  [font=\sizeFtwop,color={rgb, 255:red, 155; green, 155; blue, 155 }  ,opacity=1 ]  {$:a$};
\draw (69,112.4) node [anchor=north west][inner sep=0.75pt]  [font=\sizeFtwop,color={rgb, 255:red, 155; green, 155; blue, 155 }  ,opacity=1 ]  {$:a$};
\draw (169,112.4) node [anchor=north west][inner sep=0.75pt]  [font=\sizeFtwop,color={rgb, 255:red, 155; green, 155; blue, 155 }  ,opacity=1 ]  {$:a$};
\draw (279,112.4) node [anchor=north west][inner sep=0.75pt]  [font=\sizeFtwop,color={rgb, 255:red, 155; green, 155; blue, 155 }  ,opacity=1 ]  {$:a$};
\draw (249,152.4) node [anchor=north west][inner sep=0.75pt]  [font=\sizeFtwop,color={rgb, 255:red, 155; green, 155; blue, 155 }  ,opacity=1 ]  {$:a$};
\draw (251,202.4) node [anchor=north west][inner sep=0.75pt]  [font=\sizeFtwop,color={rgb, 255:red, 155; green, 155; blue, 155 }  ,opacity=1 ]  {$:b$};
\draw (199,202.4) node [anchor=north west][inner sep=0.75pt]  [font=\sizeFtwop,color={rgb, 255:red, 155; green, 155; blue, 155 }  ,opacity=1 ]  {$:c$};
\draw (137,288.4) node [anchor=north west][inner sep=0.75pt]  [font=\sizeFtwop,color={rgb, 255:red, 155; green, 155; blue, 155 }  ,opacity=1 ]  {$:b$};
\draw (151,202.4) node [anchor=north west][inner sep=0.75pt]  [font=\sizeFtwop,color={rgb, 255:red, 155; green, 155; blue, 155 }  ,opacity=1 ]  {$:b$};
\draw (161,162.4) node [anchor=north west][inner sep=0.75pt]  [font=\sizeFtwop,color={rgb, 255:red, 155; green, 155; blue, 155 }  ,opacity=1 ]  {$:c$};
\draw (198,152.4) node [anchor=north west][inner sep=0.75pt]  [font=\sizeFtwop,color={rgb, 255:red, 155; green, 155; blue, 155 }  ,opacity=1 ]  {$:d$};
\draw (151,252.4) node [anchor=north west][inner sep=0.75pt]  [font=\sizeFtwop,color={rgb, 255:red, 155; green, 155; blue, 155 }  ,opacity=1 ]  {$:c$};

\end{tikzpicture}}
    \caption{$G$}
    \label{fig : Induced subgraph 1}
\end{subfigure}
\hfil
\hfil
\begin{subfigure}{0.41\textwidth}
    \captionsetup{justification=centering}
    \centering
    \resizebox{.5\textwidth}{!}{\tikzset{every picture/.style={line width=0.75pt}} 

\begin{tikzpicture}[x=0.75pt,y=0.75pt,yscale=-1,xscale=1]

\draw [color=border  ,draw opacity=1 ][line width=2.25]  [dash pattern={on 2.53pt off 3.02pt}]  (222,185) -- (157,185) ;
\draw [shift={(152,185)}, rotate = 360] [fill=border  ,fill opacity=1 ][line width=0.08]  [draw opacity=0] (14.29,-6.86) -- (0,0) -- (14.29,6.86) -- cycle    ;
\draw [color=border  ,draw opacity=1 ][line width=2.25]  [dash pattern={on 2.53pt off 3.02pt}]  (237,170) -- (189.91,104.07) ;
\draw [shift={(187,100)}, rotate = 54.46] [fill=border  ,fill opacity=1 ][line width=0.08]  [draw opacity=0] (14.29,-6.86) -- (0,0) -- (14.29,6.86) -- cycle    ;
\draw [line width=2.25]    (287,270) -- (239.91,204.07) ;
\draw [shift={(237,200)}, rotate = 54.46] [fill=myblack  ][line width=0.08]  [draw opacity=0] (14.29,-6.86) -- (0,0) -- (14.29,6.86) -- cycle    ;
\draw [color=border  ,draw opacity=1 ][line width=2.25]  [dash pattern={on 2.53pt off 3.02pt}]  (187,270) -- (234.09,204.07) ;
\draw [shift={(237,200)}, rotate = 125.54] [fill=border  ,fill opacity=1 ][line width=0.08]  [draw opacity=0] (14.29,-6.86) -- (0,0) -- (14.29,6.86) -- cycle    ;
\draw [line width=2.25]    (237,170) -- (284.09,104.07) ;
\draw [shift={(287,100)}, rotate = 125.54] [fill=myblack  ][line width=0.08]  [draw opacity=0] (14.29,-6.86) -- (0,0) -- (14.29,6.86) -- cycle    ;
\draw  [color=border  ,draw opacity=1 ][fill={rgb, 255:red, 255; green, 255; blue, 255 }  ,fill opacity=1 ][dash pattern={on 3.38pt off 3.27pt}][line width=3]  (172,285) .. controls (172,276.72) and (178.72,270) .. (187,270) .. controls (195.28,270) and (202,276.72) .. (202,285) .. controls (202,293.28) and (195.28,300) .. (187,300) .. controls (178.72,300) and (172,293.28) .. (172,285) -- cycle ;
\draw  [color=myblack  ,draw opacity=1 ][fill={rgb, 255:red, 255; green, 255; blue, 255 }  ,fill opacity=1 ][line width=3]  (272,285) .. controls (272,276.72) and (278.72,270) .. (287,270) .. controls (295.28,270) and (302,276.72) .. (302,285) .. controls (302,293.28) and (295.28,300) .. (287,300) .. controls (278.72,300) and (272,293.28) .. (272,285) -- cycle ;
\draw  [color=border  ,draw opacity=1 ][fill={rgb, 255:red, 255; green, 255; blue, 255 }  ,fill opacity=1 ][dash pattern={on 3.38pt off 3.27pt}][line width=3]  (122,185) .. controls (122,176.72) and (128.72,170) .. (137,170) .. controls (145.28,170) and (152,176.72) .. (152,185) .. controls (152,193.28) and (145.28,200) .. (137,200) .. controls (128.72,200) and (122,193.28) .. (122,185) -- cycle ;
\draw  [color=myblack  ,draw opacity=1 ][fill={rgb, 255:red, 255; green, 255; blue, 255 }  ,fill opacity=1 ][line width=3]  (222,185) .. controls (222,176.72) and (228.72,170) .. (237,170) .. controls (245.28,170) and (252,176.72) .. (252,185) .. controls (252,193.28) and (245.28,200) .. (237,200) .. controls (228.72,200) and (222,193.28) .. (222,185) -- cycle ;
\draw  [color=border  ,draw opacity=1 ][fill={rgb, 255:red, 255; green, 255; blue, 255 }  ,fill opacity=1 ][dash pattern={on 3.38pt off 3.27pt}][line width=3]  (172,85) .. controls (172,76.72) and (178.72,70) .. (187,70) .. controls (195.28,70) and (202,76.72) .. (202,85) .. controls (202,93.28) and (195.28,100) .. (187,100) .. controls (178.72,100) and (172,93.28) .. (172,85) -- cycle ;
\draw  [color=myblack  ,draw opacity=1 ][fill={rgb, 255:red, 255; green, 255; blue, 255 }  ,fill opacity=1 ][line width=3]  (272,85) .. controls (272,76.72) and (278.72,70) .. (287,70) .. controls (295.28,70) and (302,76.72) .. (302,85) .. controls (302,93.28) and (295.28,100) .. (287,100) .. controls (278.72,100) and (272,93.28) .. (272,85) -- cycle ;

\draw (203,252.4) node [anchor=north west][inner sep=0.75pt]  [font=\sizeFtwo]  {$\mathbf{t_{1} .x_{1}}$};
\draw (303,252.4) node [anchor=north west][inner sep=0.75pt]  [font=\sizeFtwo]  {$\mathbf{t_{2} .x_{2}}$};
\draw (141,142.4) node [anchor=north west][inner sep=0.75pt]  [font=\sizeFtwo]  {$\mathbf{t_{4} .x_{4}}$};
\draw (263,172.4) node [anchor=north west][inner sep=0.75pt]  [font=\sizeFtwo]  {$\mathbf{t_{3} .x_{3}}$};
\draw (204,88.4) node [anchor=north west][inner sep=0.75pt]  [font=\sizeFtwo]  {$\mathbf{t_{6} .x_{6}}$};
\draw (303,95.4) node [anchor=north west][inner sep=0.75pt]  [font=\sizeFtwo]  {$\mathbf{t_{7} .x_{7}}$};
\draw (281,242.4) node [anchor=north west][inner sep=0.75pt]  [font=\sizeFtwop,color={rgb, 255:red, 155; green, 155; blue, 155 }  ,opacity=1 ]  {$:a$};
\draw (203,242.4) node [anchor=north west][inner sep=0.75pt]  [font=\sizeFtwop,color={rgb, 255:red, 155; green, 155; blue, 155 }  ,opacity=1 ]  {$:a$};
\draw (171,112.4) node [anchor=north west][inner sep=0.75pt]  [font=\sizeFtwop,color={rgb, 255:red, 155; green, 155; blue, 155 }  ,opacity=1 ]  {$:a$};
\draw (281,112.4) node [anchor=north west][inner sep=0.75pt]  [font=\sizeFtwop,color={rgb, 255:red, 155; green, 155; blue, 155 }  ,opacity=1 ]  {$:a$};
\draw (251,152.4) node [anchor=north west][inner sep=0.75pt]  [font=\sizeFtwop,color={rgb, 255:red, 155; green, 155; blue, 155 }  ,opacity=1 ]  {$:a$};
\draw (253,202.4) node [anchor=north west][inner sep=0.75pt]  [font=\sizeFtwop,color={rgb, 255:red, 155; green, 155; blue, 155 }  ,opacity=1 ]  {$:b$};
\draw (201,202.4) node [anchor=north west][inner sep=0.75pt]  [font=\sizeFtwop,color={rgb, 255:red, 155; green, 155; blue, 155 }  ,opacity=1 ]  {$:c$};
\draw (163,162.4) node [anchor=north west][inner sep=0.75pt]  [font=\sizeFtwop,color={rgb, 255:red, 155; green, 155; blue, 155 }  ,opacity=1 ]  {$:c$};
\draw (200,152.4) node [anchor=north west][inner sep=0.75pt]  [font=\sizeFtwop,color={rgb, 255:red, 155; green, 155; blue, 155 }  ,opacity=1 ]  {$:d$};

\end{tikzpicture}}
    \caption{$G_{\{x_2,x_3,x_7\}}$}
    \label{fig : Induced subgraph 2}
\end{subfigure}
\hfil
\begin{subfigure}{0.25\textwidth}
    \captionsetup{justification=centering}
    \centering
    \resizebox{\textwidth}{!}{\tikzset{every picture/.style={line width=0.75pt}} 

\begin{tikzpicture}[x=0.75pt,y=0.75pt,yscale=-1,xscale=1]

\draw [line width=2.25]    (220,185) -- (155,185) ;
\draw [shift={(150,185)}, rotate = 360] [fill=myblack  ][line width=0.08]  [draw opacity=0] (14.29,-6.86) -- (0,0) -- (14.29,6.86) -- cycle    ;
\draw [line width=2.25]    (100,285) -- (165,285) ;
\draw [shift={(170,285)}, rotate = 180] [fill=myblack  ][line width=0.08]  [draw opacity=0] (14.29,-6.86) -- (0,0) -- (14.29,6.86) -- cycle    ;
\draw [color=border  ,draw opacity=1 ][line width=2.25]  [dash pattern={on 2.53pt off 3.02pt}]  (135,170) -- (87.91,104.07) ;
\draw [shift={(85,100)}, rotate = 54.46] [fill=border  ,fill opacity=1 ][line width=0.08]  [draw opacity=0] (14.29,-6.86) -- (0,0) -- (14.29,6.86) -- cycle    ;
\draw [line width=2.25]    (185,270) -- (137.91,204.07) ;
\draw [shift={(135,200)}, rotate = 54.46] [fill=myblack  ][line width=0.08]  [draw opacity=0] (14.29,-6.86) -- (0,0) -- (14.29,6.86) -- cycle    ;
\draw [color=border  ,draw opacity=1 ][line width=2.25]  [dash pattern={on 2.53pt off 3.02pt}]  (235,170) -- (187.91,104.07) ;
\draw [shift={(185,100)}, rotate = 54.46] [fill=border  ,fill opacity=1 ][line width=0.08]  [draw opacity=0] (14.29,-6.86) -- (0,0) -- (14.29,6.86) -- cycle    ;
\draw [line width=2.25]    (285,270) -- (237.91,204.07) ;
\draw [shift={(235,200)}, rotate = 54.46] [fill=myblack  ][line width=0.08]  [draw opacity=0] (14.29,-6.86) -- (0,0) -- (14.29,6.86) -- cycle    ;
\draw [line width=2.25]    (185,270) -- (232.09,204.07) ;
\draw [shift={(235,200)}, rotate = 125.54] [fill=myblack  ][line width=0.08]  [draw opacity=0] (14.29,-6.86) -- (0,0) -- (14.29,6.86) -- cycle    ;
\draw [line width=2.25]    (235,170) -- (282.09,104.07) ;
\draw [shift={(285,100)}, rotate = 125.54] [fill=myblack  ][line width=0.08]  [draw opacity=0] (14.29,-6.86) -- (0,0) -- (14.29,6.86) -- cycle    ;
\draw  [color=setBorder  ,draw opacity=1 ][fill={rgb, 255:red, 255; green, 255; blue, 255 }  ,fill opacity=1 ][line width=3]  (170,285) .. controls (170,276.72) and (176.72,270) .. (185,270) .. controls (193.28,270) and (200,276.72) .. (200,285) .. controls (200,293.28) and (193.28,300) .. (185,300) .. controls (176.72,300) and (170,293.28) .. (170,285) -- cycle ;
\draw  [color=myblack  ,draw opacity=1 ][fill={rgb, 255:red, 255; green, 255; blue, 255 }  ,fill opacity=1 ][line width=3]  (70,285) .. controls (70,276.72) and (76.72,270) .. (85,270) .. controls (93.28,270) and (100,276.72) .. (100,285) .. controls (100,293.28) and (93.28,300) .. (85,300) .. controls (76.72,300) and (70,293.28) .. (70,285) -- cycle ;
\draw  [color=internal  ,draw opacity=1 ][fill={rgb, 255:red, 255; green, 255; blue, 255 }  ,fill opacity=1 ][line width=3]  (270,285) .. controls (270,276.72) and (276.72,270) .. (285,270) .. controls (293.28,270) and (300,276.72) .. (300,285) .. controls (300,293.28) and (293.28,300) .. (285,300) .. controls (276.72,300) and (270,293.28) .. (270,285) -- cycle ;
\draw  [color=setBorder  ,draw opacity=1 ][fill={rgb, 255:red, 255; green, 255; blue, 255 }  ,fill opacity=1 ][line width=3]  (120,185) .. controls (120,176.72) and (126.72,170) .. (135,170) .. controls (143.28,170) and (150,176.72) .. (150,185) .. controls (150,193.28) and (143.28,200) .. (135,200) .. controls (126.72,200) and (120,193.28) .. (120,185) -- cycle ;
\draw  [color=internal  ,draw opacity=1 ][fill={rgb, 255:red, 255; green, 255; blue, 255 }  ,fill opacity=1 ][line width=3]  (220,185) .. controls (220,176.72) and (226.72,170) .. (235,170) .. controls (243.28,170) and (250,176.72) .. (250,185) .. controls (250,193.28) and (243.28,200) .. (235,200) .. controls (226.72,200) and (220,193.28) .. (220,185) -- cycle ;
\draw  [color=setBorder  ,draw opacity=1 ][fill={rgb, 255:red, 255; green, 255; blue, 255 }  ,fill opacity=1 ][dash pattern={on 3.38pt off 3.27pt}][line width=3]  (170,85) .. controls (170,76.72) and (176.72,70) .. (185,70) .. controls (193.28,70) and (200,76.72) .. (200,85) .. controls (200,93.28) and (193.28,100) .. (185,100) .. controls (176.72,100) and (170,93.28) .. (170,85) -- cycle ;
\draw  [color=border  ,draw opacity=1 ][fill={rgb, 255:red, 255; green, 255; blue, 255 }  ,fill opacity=1 ][dash pattern={on 3.38pt off 3.27pt}][line width=3]  (70,85) .. controls (70,76.72) and (76.72,70) .. (85,70) .. controls (93.28,70) and (100,76.72) .. (100,85) .. controls (100,93.28) and (93.28,100) .. (85,100) .. controls (76.72,100) and (70,93.28) .. (70,85) -- cycle ;
\draw  [color=internal  ,draw opacity=1 ][fill={rgb, 255:red, 255; green, 255; blue, 255 }  ,fill opacity=1 ][line width=3]  (270,85) .. controls (270,76.72) and (276.72,70) .. (285,70) .. controls (293.28,70) and (300,76.72) .. (300,85) .. controls (300,93.28) and (293.28,100) .. (285,100) .. controls (276.72,100) and (270,93.28) .. (270,85) -- cycle ;

\draw (81,242.4) node [anchor=north west][inner sep=0.75pt]  [font=\sizeFtwo]  {$\mathbf{t_{0} .x_{0}}$};
\draw (201,252.4) node [anchor=north west][inner sep=0.75pt]  [font=\sizeFtwo]  {$\mathbf{t_{1} .x_{1}}$};
\draw (301,252.4) node [anchor=north west][inner sep=0.75pt]  [font=\sizeFtwo]  {$\mathbf{t_{2} .x_{2}}$};
\draw (141,145.4) node [anchor=north west][inner sep=0.75pt]  [font=\sizeFtwo]  {$\mathbf{t_{4} .x_{4}}$};
\draw (261,172.4) node [anchor=north west][inner sep=0.75pt]  [font=\sizeFtwo]  {$\mathbf{t_{3} .x_{3}}$};
\draw (102,88.4) node [anchor=north west][inner sep=0.75pt]  [font=\sizeFtwo]  {$\mathbf{t_{5} .x_{5}}$};
\draw (202,88.4) node [anchor=north west][inner sep=0.75pt]  [font=\sizeFtwo]  {$\mathbf{t_{6} .x_{6}}$};
\draw (301,95.4) node [anchor=north west][inner sep=0.75pt]  [font=\sizeFtwo]  {$\mathbf{t_{7} .x_{7}}$};
\draw (279,242.4) node [anchor=north west][inner sep=0.75pt]  [font=\sizeFtwop,color={rgb, 255:red, 155; green, 155; blue, 155 }  ,opacity=1 ]  {$:a$};
\draw (201,242.4) node [anchor=north west][inner sep=0.75pt]  [font=\sizeFtwop,color={rgb, 255:red, 155; green, 155; blue, 155 }  ,opacity=1 ]  {$:a$};
\draw (102,288.4) node [anchor=north west][inner sep=0.75pt]  [font=\sizeFtwop,color={rgb, 255:red, 155; green, 155; blue, 155 }  ,opacity=1 ]  {$:a$};
\draw (99,152.4) node [anchor=north west][inner sep=0.75pt]  [font=\sizeFtwop,color={rgb, 255:red, 155; green, 155; blue, 155 }  ,opacity=1 ]  {$:a$};
\draw (69,112.4) node [anchor=north west][inner sep=0.75pt]  [font=\sizeFtwop,color={rgb, 255:red, 155; green, 155; blue, 155 }  ,opacity=1 ]  {$:a$};
\draw (169,112.4) node [anchor=north west][inner sep=0.75pt]  [font=\sizeFtwop,color={rgb, 255:red, 155; green, 155; blue, 155 }  ,opacity=1 ]  {$:a$};
\draw (279,112.4) node [anchor=north west][inner sep=0.75pt]  [font=\sizeFtwop,color={rgb, 255:red, 155; green, 155; blue, 155 }  ,opacity=1 ]  {$:a$};
\draw (249,152.4) node [anchor=north west][inner sep=0.75pt]  [font=\sizeFtwop,color={rgb, 255:red, 155; green, 155; blue, 155 }  ,opacity=1 ]  {$:a$};
\draw (251,202.4) node [anchor=north west][inner sep=0.75pt]  [font=\sizeFtwop,color={rgb, 255:red, 155; green, 155; blue, 155 }  ,opacity=1 ]  {$:b$};
\draw (199,202.4) node [anchor=north west][inner sep=0.75pt]  [font=\sizeFtwop,color={rgb, 255:red, 155; green, 155; blue, 155 }  ,opacity=1 ]  {$:c$};
\draw (137,288.4) node [anchor=north west][inner sep=0.75pt]  [font=\sizeFtwop,color={rgb, 255:red, 155; green, 155; blue, 155 }  ,opacity=1 ]  {$:b$};
\draw (151,202.4) node [anchor=north west][inner sep=0.75pt]  [font=\sizeFtwop,color={rgb, 255:red, 155; green, 155; blue, 155 }  ,opacity=1 ]  {$:b$};
\draw (161,162.4) node [anchor=north west][inner sep=0.75pt]  [font=\sizeFtwop,color={rgb, 255:red, 155; green, 155; blue, 155 }  ,opacity=1 ]  {$:c$};
\draw (198,152.4) node [anchor=north west][inner sep=0.75pt]  [font=\sizeFtwop,color={rgb, 255:red, 155; green, 155; blue, 155 }  ,opacity=1 ]  {$:d$};
\draw (151,252.4) node [anchor=north west][inner sep=0.75pt]  [font=\sizeFtwop,color={rgb, 255:red, 155; green, 155; blue, 155 }  ,opacity=1 ]  {$:c$};

\end{tikzpicture}}
    \caption{$G$ with $X$ and $X^-$.}
    \label{fig : Induced subgraph 3}
\end{subfigure}
\caption{
{\em Induced subgraph and borders (a)(b):} 
We consider a graph $G$ and its induced subgraph $G_{\{x_2,x_3,x_7\}}$. Graphs have borders, as shown pointed by the dashed lines for $(a)$ $G$ and $(b)$ $G_{\{x_2,x_3,x_7\}}$. 
{\em Interior of a set (c):} We consider a set of positions $X = \{x_i \mid x_i\notin \{x_5,x_0\}\}$. Positions in this set are either interior ($X^-$ in dark blue--e.g. $t_3.x_3$) or in the boundary ($X\setminus X^-$ in cyan).
}
\label{fig : Induced subgraph}
\end{figure}

\begin{figure}[t]
\hfil
\begin{subfigure}{0.24\textwidth}
    \captionsetup{justification=centering}
    \resizebox{\textwidth}{!}{\input{figs/Local_operator_1.tex}}
    \caption{$G$}
    \label{fig : Local operator 1}
\end{subfigure}
\hfil
\hfil
\begin{subfigure}{0.24\textwidth}
    \captionsetup{justification=centering}
    \resizebox{\textwidth}{!}{\input{figs/Local_operator_2.tex}}
    \caption{$A_{x_1}G$}
    \label{fig : Local operator 2}
\end{subfigure}
\hfil
\caption{{\em Action of a local rule $A_{(-)}$} centered on a vertex $u_1 = t_1.x_1$. It affects the vertices of $\zee \restr{x}{G}$ that are circled dark blue \& cyan. Dark blue vertices (e.g. $u_1,u_4$) can be almost arbitrarily modified whereas cyan vertices must have their names and external edges preserved. Internal states $\Sigma = \{0,1\}$ are represented by white and black.
One can set $u_1'$, $u_2'$, and $u_7'$ to $3.u_1$, $1.u_2$, and $4.u_7$ respectively for instance. } 
\label{fig : Local operator}
\end{figure}

We start by formally introducing the type of graph that we consider: coloured directed acyclic port graphs. Each vertex $t.x$ may be understood as computational process at position $x$ and time tag $t$. This naming is essential if we want to be able to refer to individual events and demand that they are well-determined. The colours stand for the internal states of the computational processes. The directed edges capture the dependencies between them. We use port graphs for expressivity, i.e. in order to be able to tell a neighbouring process from another, and to track relative time advancement with respect to those neighbouring processes.
\begin{definition}[Positions, ports, states, names]
Let $\calX$ be an infinite countable set of \emph{positions}.
Let $\pi$ be a finite set of \emph{ports} and $\Sigma$ be a finite set of \emph{states}.
Let us denote $\calV := \{\, t.x \mid (t,x) \in \Z \times \calX \,\}$ and call its elements \emph{names}.

For any subsets $U \subseteq \calV$ and $X\subseteq \mathcal{X}$, let us define $(U\!:\!\pi):=\{\,(u\!:\!a)\mid u\in U,\, a\in \pi\,\}$ \MC{and $\X(U) = \{\, x \mid t.x\in U \,\}$}.
Let us also denote $\overline{X} := \mathcal{X} \setminus X$, and $\zee X := \{\, t.x \mid t\in \Z,x\in X \,\}$. Given any $u = t.x \in \mathcal{V}$, let us denote $t'.u$ for $(t'+t).x \in \calV$.
\end{definition}

\begin{definition}[Graphs, Past]\label{def : graphs}
A \emph{graph} $G$ is given by a tuple $(\I_G, \B_G, \E_G, \sigma_G)$ where:
\begin{itemize}
    \item $\I_G \subseteq \calV$ has its elements called \emph{internal vertices of $G$},
    \item $\B_G \subseteq \calV$ has its elements called \emph{border vertices of $G$},
    \item $\E_G \subseteq ((\I_G \cup \B_G)\!:\!\pi )^{2} \setminus ( \B_{G}\!:\!\pi )^2$ has its elements called \emph{(oriented) edges}, and
    \item $\sigma_G :\I_G\rightarrow \Sigma$ maps each internal vertex to its states.
\end{itemize}
We denote $\V_G := \I_G\cup \B_G$. This tuple has to be such that:
\begin{itemize}
    \item{Vertex partitioning:} $\I_{G} \cap \B_{G} = \emptyset$,
    \item{Unicity of positions:} $\forall t.x, t'.x' \in {\V_{G}},\, x = x' \Rightarrow t = t'$,
    \item{Port non-saturation:} $\forall (u\!:\!a,v\!:\!b),(u'\!:\!a',v'\!:\!b') \in {\E_G},\, u\!:\!a\neq v'\!:\!b' \,\wedge\, u\!:\!a=u'\!:\!a' \Leftrightarrow v\!:\!b=v'\!:\!b'$,
    \item{Border attachment:} $\forall u \in \B_G,\, \exists (v\!:\!a,v'\!:\!a') \in \E_G,\, u \in \{\,v,v'\,\}$, and
    \item{Acyclicity:} \LM{$\forall n \in \mathbb{N},\, \forall \langle (u_i\!:\!a_i,v_i\!:\!b_i) \rangle_{i \in \{1,\dots,{n+1}\}} \in {\E_G}^{n+1} \text{ s.t. } (\forall i \in \{1,\dots,n\},\, v_i = u_{i+1}),\, u_1 \neq v_n$}.
\end{itemize}
We denote by $\mathcal{G}$ the set of all graphs, and by $\Past(G)\subseteq \V_G$ the vertices of $G$ with no incoming edges.

\end{definition}
Summarizing, each {\position} $x$ only appears once, each vertex port $:\!a$ can only be used once per vertex, and border vertices are only there to express dangling edges of internal vertices. The $\Past(G)$ vertices intuitively stand for computational processes that are no longer awaiting for results by others, and are therefore ready to be executed.
Next, we define the induced subgraph $G_U \sqsubseteq G$ as the graph whose internal vertices are $I_{G_U}=I_G\cap U$, and whose edges are all those edges of $G$ which touch a vertex in $I_{G_U}$. Thus, its border vertices $B_{G_U}$ are those nodes of $V_G\setminus I_{G_U}$ which lie at distance one of $I_{G_U}$ in $G$ (see Fig.~\ref{fig : Induced subgraph}). We also introduce the operation $G\sqcup H$, a union which is only defined if both $G$ and $H$ can be viewed as subgraphs of the same larger graph. In particular, this implies that if $u\in I_G$,  $v\in B_G$, $v\in I_H$ and $(u\!:\!a,v\!:\!b)\in E_G$, it must be the case that $(u\!:\!a,v\!:\!b)\in E_H$ and $u\in V_H$. All this can be done quite elegantly through partial order relations.

\begin{definition}[(Induced) subgraph and boundaries.]
We write $H\subseteq G$ and say that $H$ is a \emph{subgraph} of $G$ whenever:
$$\V_H \subseteq \V_G \, \wedge \, \I_H \subseteq \I_G \,\wedge\, \E_H \subseteq \E_G \,\wedge\, \sigma_H = \sigma_G \restriction_{\I_H},$$
where $\sigma_G \restriction_{\I_H}$ denotes the function $\sigma_G$ restricted to the domain $I_H$. This relation is a partial order on $\calG$. We denote by $G \cup G'$ and $G \cap G'$ the join and meet of $G$ and $G'$ in the $\subseteq$-order when they exist.

We write $H \sqsubseteq G$ and say that $H$ is an \emph{induced subgraph} of $G$ whenever $H \subseteq G$ and $H$ is the biggest subgraph of $G$ having this set of internal vertices, i.e. $\forall H' \in \calG, H' \subseteq G \wedge \I_{H'} = \I_H \implies H' \subseteq H$.
This is also a partial order on $\calG$ and we denote by $G \sqcup G'$ and $G \sqcap G'$ the join and meet of $G$ and $G'$ in the $\sqsubseteq$-order when they exist.

Given a subset $U \subseteq \mathcal{V}$, we write $G_U$ for the induced subgraph of $G$ such that $\I_{G_U} = U \cap \I_G$.
We also write $G_u$ for $G_{\{u\}}$.
For a subset $X\subseteq \mathcal{X}$, we write $G_X$ for the induced subgraph $G_{\zee X}$. We write $X^-(G)$ or just $X^-$ the set $\{x\in X \,|\, x\in \X(\I_G)\wedge\V_{G_x} \subseteq \zee X\}$, and call these the internal vertices of $X$ in $G$. 
The leftover vertices (i.e. $X\setminus X^-$) will be referred to as the boundary of $X$ in $G$, see Fig.\ref{fig : Induced subgraph}.
\end{definition}

Notice that $G = G_X \sqcup G_{\overline{X}}$.
This decomposition will allow us to define the action of our local rules.
We will proceed as follow. 
First we define \emph{operators} $A_{(-)}$ which, given a {\position} $x$, rewrite the graph $G$ as $A_x G$. Second we define a \emph{neighbourhood scheme} $\restri{}$ which, given a \position{}  $x$, selects a subset of nearby {\position}s $\restr{x}{G}$. Third we say that $A_{(-)}$ is \emph{$\restri{}$-local} if the action of $A_x$ is to replace just the left hand side of $G = G_{\restr{x}{G}} \sqcup G_{\overline{\restr{x}{G}}}$, independently of the right hand side---an operation which can likely be formalised as a double push-out \cite{HarmerFundamentals}.

\begin{definition}[Neighbourhood scheme]\label{def : chi-neighbourhood}
A \emph{neighbourhood scheme} $\restri{}$ is an operator from $\mathcal{P}(\mathcal{X}) \times \calG$ to $\mathcal{P}(\mathcal{X})$ mapping a pair $(\omega, G)$ to $\restr{\omega}{G}\subseteq \mathcal{X}$ such that all $x\in \restr{\omega}{G}$ is \textbf{reachable} from $\omega$, i.e.
there exists $p : \omega \rightarrow x$ a directed path in $\E_G$ from some element of $\omega$ to $x$.\\
The input $\omega$ is formally a set of vertices, however we will often apply $\restri{}$ to $\omega\in \mathcal{X}^*$ implicitly referring to the underlying set.
When $G$ is clear from context, we write $\restri{\omega}$ instead of $\restr{\omega}{G}$.
\end{definition}
Note that $\preci{\omega}$ refers to the internal vertices of $\restri{\omega}$ with respect to $G$ as in Def.~\ref{def : graphs}. Notice also that our definition of neighbourhood schemes is quite permissive. It even allows for unbounded, or infinite neighbourhood (which in a DAG do not equate to unbounded speed of propagation of information), even though for all practical purposes one is likely to use a bounded neighbourhood. This is all we need for some of our results (Props. \ref{prop:weakconsistency} and \ref{lemma : Unicity of space-time cuts}). For some others we must forbid that global criteria be used to decide whether some closeby vertex belongs to the neighbourhood or not. To forbid this to happen we define extensivity. It asks that the neighbourhood $\restr{\omega}{G}$ as computed by the function $\restri{}$ over a graph $G$, be the same as that computed over a big enough subgraph of $G$. This can be understood as a form of graph-locality of the neighbourhood scheme $\restri{}$ itself, for instance if $\restr{\omega}{G}$ is computed step by step starting from $\omega$ until it hits a `wall' i.e. a local ending criterion, then it will be extensive.

\begin{definition}[Extensivity]
    \LM{A neighbourhood scheme $N$ is \emph{extensive} if} $G_\restri{\omega}\sqsubseteq H \sqsubseteq G$ implies $\restr{\omega}{G} = \restr{\omega}{H}$ \LM{for any two graphs $G$ and $H$ and any $\omega \in \mathcal{P}(\mathcal{X})$}.
\end{definition}
Still, the notion of extensive neighbourhood schemes allows for unbounded neighbourhoods. In practice the user of the theory will likely choose a much smaller neighbourhood---our results will apply so long as the abstract conditions are met. 


Next we define the notion of a local rule $A_{(-)}$. Before we proceed, remember that our graphs can never contain both vertices $u=t.x$ and $u'=(t+1).x$. It follows that the local rule $A_x$ will act unambiguously on the unique vertex of the form $u=t.x$. Keep in mind also that our directed edges stand for dependencies between events, i.e. if $u=t.x$ points to $v=t'.y$, then $v$ is considered ahead in time of $u$, and thus frozen awaiting information from $u$. The action of some local rule $A_y$ on $v$ is therefore trivial, preventing $y$ to be computed too far ahead and thereby providing a weak synchronisation mechanism. The action of $A_x$ on $u$ is non-trivial only if $u$ is minimal aka Past. Such a $u$ can be thought of as lagging behind in time, and no longer awaiting for any information---it has reached its local normal form. The action of $A_x$ is to dispose of it by communicating its information to $v$ and other dependencies, and creating vertex $u'=(t+\Delta t).x$ in some provisional state.

\begin{definition}[Local rule]\label{def : locality}
    A \emph{local rule} is an operator over graphs $ A_{( -)} :\mathcal{X} \to ( \mathcal{G\to G})$ 
    which is $\restri{}$-local for some 
    neighborhood scheme $\restri{}$, i.e. 
    for all $G$ in $\calG$, for all $x$ in $X$,
    $${A_x G = 
    \begin{cases}
    (A_x G_\restri{x}) \sqcup G_\crestri{x}&\textrm{if $x\in \Past(G)$}\\
    G&\textrm{otherwise}
    \end{cases}}
    $$
    From now on $A_{yx}$ will stand for $A_y A_x$.
We say that $\omega \in \mathcal{X}^*$ is a valid sequence in $G$ if, for all $\omega_1,\omega_2\in \mathcal{X}^*$ such that $\omega = \omega_2 x \omega_1$, we have $x\in \X(\Past(A_{\omega_1}G))$.
We denote $\Omega_G(A) \subseteq \mathcal{X}^*$ the set of valid sequences in $G$.
\end{definition}
\begin{remark}\label{rk:locality excludes borders}
  \MC{Let $x\in \X(\Past(G))$.}
  \MC{The fact that $A_x(G_{\restri{x}}) \sqcup G_{\overline{\restri{x}}}$ needs to exist for any $G$ implies that $A_x(G_{\restri{x}})$ and $G_{\overline{\restri{x}}}$ must agree upon their ``common frontier''. As a consequence :}
    \begin{itemize}
        \item the border vertices $B_{G_{\restri{x}}}$ cannot be modified (dashed vertices in Fig.~\ref{fig : Local operator}).
        \item the border edges $\E_{G_{\restri{x}}}\setminus\ (I_{G_{\restri{x}}}\!:\!\pi)^2$ cannot be modified (dashed edges in Fig.~\ref{fig : Local operator}).
        \item new vertices $V_{A_x G}\setminus V_G$ are exclusively vertices in $\zee \prec{x}{G}$ with time tags \MC{modified} by $A_x$ (the vertices in dark blue (e.g. $u_1$, $u_2$) in Fig.~\ref{fig : Local operator}).
    \end{itemize}
\end{remark}
An additional property that can be required, without difficulty, of our operators, is renaming-invariance \cite{ArrighiNamesInQG}. We will skip this here.

Having defined the considered graphs and the kind of local transformations allowed on them, let us state our goal informally.
Consider a graph $G$, a local rule $A_{(-)}$ and all possible sequences $\omega, \omega', \ldots\in\Omega_G(A)$.
If one applies $A_\omega$, one obtains one possible evolution of the system.
But $A_{\omega'}$, $A_{\omega''}$,\ldots are equally legitimate different orders of local rule applications. We aim to define precisely what it means for all these possible evolutions to actually agree on a consistent common story, i.e. a \emph{consistent space-time diagram}. To motivate the remaining formalization, we consider examples.
Later we give sufficient conditions for a local rule to induce such consistent space-time diagrams.

\begin{definition}[Space-time diagram]
    Given a graph $G$ and a local rule $A_{(-)}$, \emph{space-time diagram} $\mathcal{M}_A(G) := \{\, A_{\omega} G\mid \omega \in \Omega_G(A) \,\}$ is the set of all generated graphs.
    We sometimes omit $A$ and $G$.
\end{definition}
To visualize how the graphs in $\mathcal{M}$ share common vertices and edges, some \emph{space-time backgrounds} are depicted (Figs~\ref{fig : We need consistency}, and~\ref{fig : Space time diagram example 2}).
These space-time backgrounds are ``pseudo-graphs'', i.e. graph in a common sense, without any of the port-constraints of Def. \ref{def : graphs} imposed. Each pseudo-graph $M$ is defined by:
\begin{align*}
    V_M &=\bigcup_{G\in \mathcal{M}} V_G & E_M &=\bigcup_{G\in \mathcal{M}} \E_G.
\end{align*}

\section{Particle system example}\label{sec:exV1}

We now show how asynchronous applications of a local rule can represent a dynamical system of left and right-moving particles, consistently, thereby fixing issues of Fig.~\ref{fig : faster than light}. 

\begin{figure}[h]
\hfil
\begin{subfigure}{0.2\textwidth}
    \captionsetup{justification=centering}
    \resizebox{.8\textwidth}{!}{\tikzset{every picture/.style={line width=0.75pt}} 

\begin{tikzpicture}[x=0.75pt,y=0.75pt,yscale=-1,xscale=1]

\draw [color={rgb, 255:red, 74; green, 74; blue, 74 }  ,draw opacity=1 ][line width=2.25]  [dash pattern={on 2.53pt off 3.02pt}]  (276,176) -- (321.96,104.74) ;
\draw [shift={(324.67,100.53)}, rotate = 122.82] [fill={rgb, 255:red, 74; green, 74; blue, 74 }  ,fill opacity=1 ][line width=0.08]  [draw opacity=0] (14.29,-6.86) -- (0,0) -- (14.29,6.86) -- cycle    ;
\draw [color={rgb, 255:red, 74; green, 74; blue, 74 }  ,draw opacity=1 ][line width=2.25]  [dash pattern={on 2.53pt off 3.02pt}]  (176,176) -- (126.52,104.64) ;
\draw [shift={(123.67,100.53)}, rotate = 55.26] [fill={rgb, 255:red, 74; green, 74; blue, 74 }  ,fill opacity=1 ][line width=0.08]  [draw opacity=0] (14.29,-6.86) -- (0,0) -- (14.29,6.86) -- cycle    ;
\draw [color={rgb, 255:red, 74; green, 74; blue, 74 }  ,draw opacity=1 ][line width=2.25]  [dash pattern={on 2.53pt off 3.02pt}]  (116,284) -- (162.19,202.88) ;
\draw [shift={(164.67,198.53)}, rotate = 119.66] [fill={rgb, 255:red, 74; green, 74; blue, 74 }  ,fill opacity=1 ][line width=0.08]  [draw opacity=0] (14.29,-6.86) -- (0,0) -- (14.29,6.86) -- cycle    ;
\draw [color={rgb, 255:red, 0; green, 0; blue, 0 }  ,draw opacity=1 ][line width=2.25]    (226,274) -- (189.91,202) ;
\draw [shift={(187.67,197.53)}, rotate = 63.38] [fill={rgb, 255:red, 0; green, 0; blue, 0 }  ,fill opacity=1 ][line width=0.08]  [draw opacity=0] (14.29,-6.86) -- (0,0) -- (14.29,6.86) -- cycle    ;
\draw [color={rgb, 255:red, 74; green, 74; blue, 74 }  ,draw opacity=1 ][line width=2.25]  [dash pattern={on 2.53pt off 3.02pt}]  (335,285) -- (291.05,203.93) ;
\draw [shift={(288.67,199.53)}, rotate = 61.54] [fill={rgb, 255:red, 74; green, 74; blue, 74 }  ,fill opacity=1 ][line width=0.08]  [draw opacity=0] (14.29,-6.86) -- (0,0) -- (14.29,6.86) -- cycle    ;
\draw [line width=2.25]    (226,274) -- (264.27,204.09) ;
\draw [shift={(266.67,199.7)}, rotate = 118.69] [fill={rgb, 255:red, 0; green, 0; blue, 0 }  ][line width=0.08]  [draw opacity=0] (14.29,-6.86) -- (0,0) -- (14.29,6.86) -- cycle    ;
\draw  [color={rgb, 255:red, 74; green, 74; blue, 74 }  ,draw opacity=1 ][fill={rgb, 255:red, 255; green, 255; blue, 255 }  ,fill opacity=1 ][dash pattern={on 3.38pt off 3.27pt}][line width=3]  (320,285) .. controls (320,276.72) and (326.72,270) .. (335,270) .. controls (343.28,270) and (350,276.72) .. (350,285) .. controls (350,293.28) and (343.28,300) .. (335,300) .. controls (326.72,300) and (320,293.28) .. (320,285) -- cycle ;
\draw  [color={rgb, 255:red, 74; green, 74; blue, 74 }  ,draw opacity=1 ][fill={rgb, 255:red, 255; green, 255; blue, 255 }  ,fill opacity=1 ][dash pattern={on 3.38pt off 3.27pt}][line width=3]  (102,284) .. controls (102,276.27) and (108.27,270) .. (116,270) .. controls (123.73,270) and (130,276.27) .. (130,284) .. controls (130,291.73) and (123.73,298) .. (116,298) .. controls (108.27,298) and (102,291.73) .. (102,284) -- cycle ;
\draw  [color={rgb, 255:red, 38; green, 105; blue, 185 }  ,draw opacity=1 ][fill={rgb, 255:red, 255; green, 255; blue, 255 }  ,fill opacity=1 ][line width=3]  (200,274) .. controls (200,259.64) and (211.64,248) .. (226,248) .. controls (240.36,248) and (252,259.64) .. (252,274) .. controls (252,288.36) and (240.36,300) .. (226,300) .. controls (211.64,300) and (200,288.36) .. (200,274) -- cycle ;
\draw  [color={rgb, 255:red, 80; green, 227; blue, 194 }  ,draw opacity=1 ][fill={rgb, 255:red, 255; green, 255; blue, 255 }  ,fill opacity=1 ][line width=3]  (152,176) .. controls (152,162.75) and (162.75,152) .. (176,152) .. controls (189.25,152) and (200,162.75) .. (200,176) .. controls (200,189.25) and (189.25,200) .. (176,200) .. controls (162.75,200) and (152,189.25) .. (152,176) -- cycle ;
\draw  [color={rgb, 255:red, 80; green, 227; blue, 194 }  ,draw opacity=1 ][fill={rgb, 255:red, 255; green, 255; blue, 255 }  ,fill opacity=1 ][line width=3]  (252,176) .. controls (252,162.75) and (262.75,152) .. (276,152) .. controls (289.25,152) and (300,162.75) .. (300,176) .. controls (300,189.25) and (289.25,200) .. (276,200) .. controls (262.75,200) and (252,189.25) .. (252,176) -- cycle ;
\draw  [color={rgb, 255:red, 74; green, 74; blue, 74 }  ,draw opacity=1 ][fill={rgb, 255:red, 255; green, 255; blue, 255 }  ,fill opacity=1 ][dash pattern={on 3.38pt off 3.27pt}][line width=3]  (100,85) .. controls (100,76.72) and (106.72,70) .. (115,70) .. controls (123.28,70) and (130,76.72) .. (130,85) .. controls (130,93.28) and (123.28,100) .. (115,100) .. controls (106.72,100) and (100,93.28) .. (100,85) -- cycle ;
\draw  [color={rgb, 255:red, 74; green, 74; blue, 74 }  ,draw opacity=1 ][fill={rgb, 255:red, 255; green, 255; blue, 255 }  ,fill opacity=1 ][dash pattern={on 3.38pt off 3.27pt}][line width=3]  (320,85) .. controls (320,76.72) and (326.72,70) .. (335,70) .. controls (343.28,70) and (350,76.72) .. (350,85) .. controls (350,93.28) and (343.28,100) .. (335,100) .. controls (326.72,100) and (320,93.28) .. (320,85) -- cycle ;
\draw [color={rgb, 255:red, 80; green, 227; blue, 194 }  ,draw opacity=1 ][line width=3]    (176,152) -- (176,200) ;
\draw [color={rgb, 255:red, 80; green, 227; blue, 194 }  ,draw opacity=1 ][line width=3]    (276,152) -- (276,200) ;
\draw [color={rgb, 255:red, 74; green, 107; blue, 226 }  ,draw opacity=1 ][line width=3]    (226,250) -- (226,298) ;

\draw (250,288.4) node [anchor=north west][inner sep=0.75pt]  [font=\sizeFfour]  {$\mathbf{u=t.x}$};
\draw (200,201.4) node [anchor=north west][inner sep=0.75pt]  [font=\sizeFfourp,color={rgb, 255:red, 155; green, 155; blue, 155 }  ,opacity=1 ]  {$:a'$};
\draw (230,200.4) node [anchor=north west][inner sep=0.75pt]  [font=\sizeFfourp,color={rgb, 255:red, 155; green, 155; blue, 155 }  ,opacity=1 ]  {$:b'$};
\draw (246,232.4) node [anchor=north west][inner sep=0.75pt]  [font=\sizeFfourp,color={rgb, 255:red, 155; green, 155; blue, 155 }  ,opacity=1 ]  {$:b$};
\draw (184,232.4) node [anchor=north west][inner sep=0.75pt]  [font=\sizeFfourp,color={rgb, 255:red, 155; green, 155; blue, 155 }  ,opacity=1 ]  {$:a$};
\draw (181,162) node [anchor=north west][inner sep=0.75pt]   [align=left] {\textbf{{\sizeFfours i}}};
\draw (261,162) node [anchor=north west][inner sep=0.75pt]   [align=left] {\textbf{{\sizeFfours j}}};

\end{tikzpicture}}
    \caption{$G_{\restri{x_1}}$.}
    \label{fig : Dynamic example 1; G}
\end{subfigure}
\hfil
\hfil
\begin{subfigure}{0.2\textwidth}
    \captionsetup{justification=centering}
    \resizebox{.8\textwidth}{!}{\tikzset{every picture/.style={line width=0.75pt}} 

\begin{tikzpicture}[x=0.75pt,y=0.75pt,yscale=-1,xscale=1]

\draw [color={rgb, 255:red, 74; green, 74; blue, 74 }  ,draw opacity=1 ][line width=2.25]  [dash pattern={on 2.53pt off 3.02pt}]  (276,176) -- (323.91,103.6) ;
\draw [shift={(326.67,99.43)}, rotate = 123.49] [fill={rgb, 255:red, 74; green, 74; blue, 74 }  ,fill opacity=1 ][line width=0.08]  [draw opacity=0] (14.29,-6.86) -- (0,0) -- (14.29,6.86) -- cycle    ;
\draw [color={rgb, 255:red, 74; green, 74; blue, 74 }  ,draw opacity=1 ][line width=2.25]  [dash pattern={on 2.53pt off 3.02pt}]  (176,176) -- (118.13,103.9) ;
\draw [shift={(115,100)}, rotate = 51.25] [fill={rgb, 255:red, 74; green, 74; blue, 74 }  ,fill opacity=1 ][line width=0.08]  [draw opacity=0] (14.29,-6.86) -- (0,0) -- (14.29,6.86) -- cycle    ;
\draw [color={rgb, 255:red, 74; green, 74; blue, 74 }  ,draw opacity=1 ][line width=2.25]  [dash pattern={on 2.53pt off 3.02pt}]  (116,284) -- (162.14,205.08) ;
\draw [shift={(164.67,200.77)}, rotate = 120.31] [fill={rgb, 255:red, 74; green, 74; blue, 74 }  ,fill opacity=1 ][line width=0.08]  [draw opacity=0] (14.29,-6.86) -- (0,0) -- (14.29,6.86) -- cycle    ;
\draw [color={rgb, 255:red, 0; green, 0; blue, 0 }  ,draw opacity=1 ][line width=2.25]    (276,176) -- (239.95,105.88) ;
\draw [shift={(237.67,101.43)}, rotate = 62.79] [fill={rgb, 255:red, 0; green, 0; blue, 0 }  ,fill opacity=1 ][line width=0.08]  [draw opacity=0] (14.29,-6.86) -- (0,0) -- (14.29,6.86) -- cycle    ;
\draw [color={rgb, 255:red, 74; green, 74; blue, 74 }  ,draw opacity=1 ][line width=2.25]  [dash pattern={on 2.53pt off 3.02pt}]  (335,285) -- (292.34,204.42) ;
\draw [shift={(290,200)}, rotate = 62.1] [fill={rgb, 255:red, 74; green, 74; blue, 74 }  ,fill opacity=1 ][line width=0.08]  [draw opacity=0] (14.29,-6.86) -- (0,0) -- (14.29,6.86) -- cycle    ;
\draw [line width=2.25]    (176,176) -- (212.32,107.85) ;
\draw [shift={(214.67,103.43)}, rotate = 118.05] [fill={rgb, 255:red, 0; green, 0; blue, 0 }  ][line width=0.08]  [draw opacity=0] (14.29,-6.86) -- (0,0) -- (14.29,6.86) -- cycle    ;
\draw  [color={rgb, 255:red, 74; green, 74; blue, 74 }  ,draw opacity=1 ][fill={rgb, 255:red, 255; green, 255; blue, 255 }  ,fill opacity=1 ][dash pattern={on 3.38pt off 3.27pt}][line width=3]  (320,285) .. controls (320,276.72) and (326.72,270) .. (335,270) .. controls (343.28,270) and (350,276.72) .. (350,285) .. controls (350,293.28) and (343.28,300) .. (335,300) .. controls (326.72,300) and (320,293.28) .. (320,285) -- cycle ;
\draw  [color={rgb, 255:red, 74; green, 74; blue, 74 }  ,draw opacity=1 ][fill={rgb, 255:red, 255; green, 255; blue, 255 }  ,fill opacity=1 ][dash pattern={on 3.38pt off 3.27pt}][line width=3]  (102,284) .. controls (102,276.27) and (108.27,270) .. (116,270) .. controls (123.73,270) and (130,276.27) .. (130,284) .. controls (130,291.73) and (123.73,298) .. (116,298) .. controls (108.27,298) and (102,291.73) .. (102,284) -- cycle ;
\draw  [color={rgb, 255:red, 38; green, 105; blue, 185 }  ,draw opacity=1 ][fill={rgb, 255:red, 255; green, 255; blue, 255 }  ,fill opacity=1 ][line width=3]  (200,76) .. controls (200,61.64) and (211.64,50) .. (226,50) .. controls (240.36,50) and (252,61.64) .. (252,76) .. controls (252,90.36) and (240.36,102) .. (226,102) .. controls (211.64,102) and (200,90.36) .. (200,76) -- cycle ;
\draw  [color={rgb, 255:red, 80; green, 227; blue, 194 }  ,draw opacity=1 ][fill={rgb, 255:red, 255; green, 255; blue, 255 }  ,fill opacity=1 ][line width=3]  (152,176) .. controls (152,162.75) and (162.75,152) .. (176,152) .. controls (189.25,152) and (200,162.75) .. (200,176) .. controls (200,189.25) and (189.25,200) .. (176,200) .. controls (162.75,200) and (152,189.25) .. (152,176) -- cycle ;
\draw  [color={rgb, 255:red, 80; green, 227; blue, 194 }  ,draw opacity=1 ][fill={rgb, 255:red, 255; green, 255; blue, 255 }  ,fill opacity=1 ][line width=3]  (252,176) .. controls (252,162.75) and (262.75,152) .. (276,152) .. controls (289.25,152) and (300,162.75) .. (300,176) .. controls (300,189.25) and (289.25,200) .. (276,200) .. controls (262.75,200) and (252,189.25) .. (252,176) -- cycle ;
\draw  [color={rgb, 255:red, 74; green, 74; blue, 74 }  ,draw opacity=1 ][fill={rgb, 255:red, 255; green, 255; blue, 255 }  ,fill opacity=1 ][dash pattern={on 3.38pt off 3.27pt}][line width=3]  (100,85) .. controls (100,76.72) and (106.72,70) .. (115,70) .. controls (123.28,70) and (130,76.72) .. (130,85) .. controls (130,93.28) and (123.28,100) .. (115,100) .. controls (106.72,100) and (100,93.28) .. (100,85) -- cycle ;
\draw  [color={rgb, 255:red, 74; green, 74; blue, 74 }  ,draw opacity=1 ][fill={rgb, 255:red, 255; green, 255; blue, 255 }  ,fill opacity=1 ][dash pattern={on 3.38pt off 3.27pt}][line width=3]  (320,85) .. controls (320,76.72) and (326.72,70) .. (335,70) .. controls (343.28,70) and (350,76.72) .. (350,85) .. controls (350,93.28) and (343.28,100) .. (335,100) .. controls (326.72,100) and (320,93.28) .. (320,85) -- cycle ;
\draw [color={rgb, 255:red, 80; green, 227; blue, 194 }  ,draw opacity=1 ][line width=3]    (176,152) -- (176,200) ;
\draw [color={rgb, 255:red, 80; green, 227; blue, 194 }  ,draw opacity=1 ][line width=3]    (276,152) -- (276,200) ;
\draw [color={rgb, 255:red, 74; green, 107; blue, 226 }  ,draw opacity=1 ][line width=3]    (226,50) -- (226,98) ;

\draw (241,92.4) node [anchor=north west][inner sep=0.75pt]  [font=\sizeFfour]  {$\mathbf{( t+1) .x}$};
\draw (251,112.4) node [anchor=north west][inner sep=0.75pt]  [font=\sizeFfourp,color={rgb, 255:red, 155; green, 155; blue, 155 }  ,opacity=1 ]  {$:a'$};
\draw (179,101.4) node [anchor=north west][inner sep=0.75pt]  [font=\sizeFfourp,color={rgb, 255:red, 155; green, 155; blue, 155 }  ,opacity=1 ]  {$:b'$};
\draw (195,141.4) node [anchor=north west][inner sep=0.75pt]  [font=\sizeFfourp,color={rgb, 255:red, 155; green, 155; blue, 155 }  ,opacity=1 ]  {$:b$};
\draw (236,141.4) node [anchor=north west][inner sep=0.75pt]  [font=\sizeFfourp,color={rgb, 255:red, 155; green, 155; blue, 155 }  ,opacity=1 ]  {$:a$};
\draw (231,62) node [anchor=north west][inner sep=0.75pt]   [align=left] {\textbf{{\sizeFfours i}}};
\draw (211,62) node [anchor=north west][inner sep=0.75pt]   [align=left] {\textbf{{\sizeFfours j}}};

\end{tikzpicture}}
    \caption{$A_{x_1}G_{\restri{x_1}}$.}
    \label{fig : Dynamic example 1; A_x G}
\end{subfigure}
\hfil        
\caption{{\em The local rule for the particle system.} $A_x$ acts by consuming the internal state $i=\sigma^r_G(v)$ of vertex $v$ (and symmetrically with $j=\sigma^l_G(w)$), thereby moving those particles at $x$ if they are present. It also updates ports from $a,b$ to $a',b'$, flips the arrows pointing to $x$, and increments its timetag, in order to move the vertex from past $u=t.x$ to future $u'= 1.u = (t+1).x$. Dashed edges and vertices do not influence the local rule.
}

\label{fig : Dynamic example 1}
\end{figure}
Again the vertices of our graphs have names of the form $u=t.x$, they must be thought of as events in space-time.
Here internal states are pairs of bits $\sigma_G(u)=(\sigma^l_G(u),\sigma^r_G(u))$, representing the presence of a left-moving particle or not, and of a right-moving particle or not. Each edge goes either from port $:\!a$ to $:\!a'$ or from port $:\!b$ to $:\!b'$, thereby indicating a spatial direction ($a$ versus $b$) and a temporal orientation (unprimed versus primed).
Altogether each graph must be thought of, not as a space-time diagram, but as a `space-like cuts' of a space-time diagram, i.e. a snapshot, cf. Fig.~\ref{fig:spacetimediagram}. 
\begin{figure}[t]
\centering
\begin{subfigure}{.435\textwidth}
    \resizebox{\textwidth}{!}{\begin{tikzpicture}


\foreach \i in {1,2,...,6}{
    \foreach \j in {1,2,3}{
    	\FPeval{\posx}{clip(\i*8)}
        \FPeval{\posy}{clip(\j*8)}
        \ifnum \j=2
            {\vertex{\posx}{\posy}{border}{0}{\j}{\i\j}}
        \else
            {\vertex{\posx}{\posy}{myblack}{0}{\j}{\i\j}}
        \fi
    }
}

\foreach \i in {1,2,...,5}{
    \foreach \j in {1,2,3}{
    	\FPeval{\posx}{clip(\i*8+4)}
        \FPeval{\posy}{clip(\j*8+4)}
        \ifnum \j=2
            {\ifnum \i=3
                {\vertex{\posx}{\posy}{border}{0}{u}{S\i\j}}
            \else
                {\vertex{\posx}{\posy}{border}{0}{\j}{S\i\j}}
            \fi
            }
        \else
            {\vertex{\posx}{\posy}{myblack}{0}{\j}{S\i\j}}
        \fi
    }
}


\foreach \i in {1,2,...,5}{
    \FPeval{\k}{clip(\i+1)}
    \edge{\i1}{S\i1}{myblack}
    \edge{\k1}{S\i1}{myblack}
    \edge{S\i1}{\i2}{border}
    \edge{S\i1}{\k2}{border}
    \edge{\i2}{S\i2}{border}
    \edge{\k2}{S\i2}{border}
    \edge{S\i2}{\i3}{border}
    \edge{S\i2}{\k3}{border}
    \edge{\i3}{S\i3}{myblack}
    \edge{\k3}{S\i3}{myblack}
}

\rightmoover{20}{12}{myblack}{first}
\rightmoover{36}{28}{myblack}{second}
\leftmoover{36}{12}{myblack}{third}
\leftmoover{20}{28}{myblack}{fourth}

\draw (7.5,29) node[right]{{\Huge \scalefont{2.5} $G'$}};
\draw (7.5,12) node[right]{{\Huge \scalefont{2.5} $G$}};
    
\end{tikzpicture}}
    \caption{Particle system example}
    \label{fig:spacetimediagram}
\end{subfigure}
\begin{subfigure}{.275\textwidth}
    \resizebox{\textwidth}{!}{\begin{tikzpicture}


\foreach \i in {1,2,...,4}{
    \foreach \j in {1,2,3}{
    	\FPeval{\posx}{clip(\i*8)}
        \FPeval{\posy}{clip(\j*8)}
        \ifnum \j=2
            {\vertex{\posx}{\posy}{myblack}{0}{}{\i\j}}
        \else
            {\vertex{\posx}{\posy}{border}{0}{}{\i\j}}
        \fi
    }
}

\foreach \i in {1,2,3}{
    \foreach \j in {1,2,3}{
    	\FPeval{\posx}{clip(\i*8+4)}
        \FPeval{\posy}{clip(\j*8+4)}
        \ifnum \j=2
            {\ifnum \i=3
                {\vertex{\posx}{\posy}{border}{0}{}{S\i\j}}
            \else
                {\vertex{\posx}{\posy}{myblack}{0}{}{S\i\j}}
            \fi
            }
        \else
            {\vertex{\posx}{\posy}{border}{0}{}{S\i\j}}
        \fi
    }
}


    \edge{11}{S11}{border}
    \edge{21}{S11}{border}
    \edge{S11}{12}{border}
    \edge{S11}{22}{border}
    
    \edge{12}{S12}{myblack}
    \edge{22}{S12}{myblack}
    \edge{S12}{13}{border}
    \edge{S12}{23}{border}
    
    \edge{13}{S13}{border}
    \edge{23}{S13}{border}


    \edge{21}{S21}{border}
    \edge{31}{S21}{border}
    \edge{S21}{22}{border}
    \edge{S21}{32}{border}
    
    \edge{22}{S22}{myblack}
    \edge{32}{S22}{myblack}
    \edge{S22}{23}{border}
    \edge{S22}{33}{border}
    
    \edge{23}{S23}{border}
    \edge{33}{S23}{border}

    \edge{31}{S31}{border}
    \edge{41}{S31}{border}
    \edge{S31}{32}{myblack}
    \edge{S31}{42}{myblack}
    \edge{32}{S32}{border}
    \edge{42}{S32}{border}
    \edge{S32}{33}{border}
    \edge{S32}{43}{border}
    \edge{33}{S33}{border}
    \edge{43}{S33}{border}

\rightmoover{20}{20}{myblack}{first}
\emptycell{28}{12}{myblack}{second}

\draw (29.5,13) node[right]{{\Huge \scalefont{2.5} $H$}};

\draw (21,20) node[right]{{\sizeFfive \scalefont{2.5} $v$}};
\draw (29,28) node[right]{{\sizeFfive \scalefont{2.5} $w$}};
\draw (13,12) node[right]{{\sizeFfive \scalefont{2.5} $u$}};

\draw (20.75,18) node {\sizeFfivep $:a'$};
\draw (19,18) node {\sizeFfivep $:b'$};
    
\end{tikzpicture}}
    \caption{Here $\sigma_H(v) = (0,1)$ $\dots$}
    \label{fig : We need consistency 1}
\end{subfigure}
\begin{subfigure}{.275\textwidth}
    \resizebox{\textwidth}{!}{\begin{tikzpicture}


\foreach \i in {1,2,...,4}{
    \foreach \j in {1,2,3}{
    	\FPeval{\posx}{clip(\i*8)}
        \FPeval{\posy}{clip(\j*8)}
        {\vertex{\posx}{\posy}{border}{0}{}{\i\j}}
    }
}

\foreach \i in {1,2,3}{
    \foreach \j in {1,2,3}{
    	\FPeval{\posx}{clip(\i*8+4)}
        \FPeval{\posy}{clip(\j*8+4)}
        {\vertex{\posx}{\posy}{border}{0}{}{S\i\j}}
    }
}


    \edge{11}{S11}{border}
    \edge{21}{S11}{border}
    \edge{S11}{12}{myblack}
    \edge{S11}{22}{myblack}
    
    \edge{12}{S12}{border}
    \edge{22}{S12}{border}
    \edge{S12}{13}{border}
    \edge{S12}{23}{border}
    
    \edge{13}{S13}{border}
    \edge{23}{S13}{border}


    \edge{21}{S21}{border}
    \edge{31}{S21}{border}
    \edge{S21}{22}{border}
    \edge{S21}{32}{border}
    
    \edge{22}{S22}{myblack}
    \edge{32}{S22}{border}
    \edge{S22}{23}{border}
    \edge{S22}{33}{myblack}
    
    \edge{23}{S23}{border}
    \edge{33}{S23}{border}

    \edge{31}{S31}{border}
    \edge{41}{S31}{border}
    \edge{S31}{32}{border}
    \edge{S31}{42}{border}
    \edge{32}{S32}{border}
    \edge{42}{S32}{border}
    \edge{S32}{33}{border}
    \edge{S32}{43}{border}
    \edge{33}{S33}{myblack}
    \edge{43}{S33}{myblack}

\rightmoover{28}{28}{myblack}{first}
\emptycell{24}{24}{myblack}{second}
\emptycell{20}{20}{myblack}{third}
\emptycell{16}{16}{myblack}{fourth}
\emptycell{12}{12}{myblack}{fifth}
\emptycell{8}{16}{myblack}{fifth}

\draw (23,27.5) node[right]{{\Huge \scalefont{2.5} $H'$}};

\draw (21,20) node[right]{{\sizeFfive \scalefont{2.5} $v$}};
\draw (29,28) node[right]{{\sizeFfive \scalefont{2.5} $w$}};
\draw (13,12) node[right]{{\sizeFfive \scalefont{2.5} $u$}};

\draw (20.75,22) node {\sizeFfivep $:a$};
\draw (19,18) node {\sizeFfivep $:b'$};
    
\end{tikzpicture}}
    \caption{$\dots$ there $\sigma_{H'}(v)=(0,0)$.}
    \label{fig : We need consistency 2}
\end{subfigure}
\caption{{\em Particle system example (a):} In black we highlight just the graphs $G$ and $G''$ belonging to the space-time diagram $\mathcal{M}_A(G)$. The local rule here moves the left-side particle towards the right and the right-side particle towards the left. Note how the problem raised by Fig.~\ref{fig : faster than light} is solved. The only point in space-time which can contain both particles is $u$.
{\em States depend on the cut (b)(c):} Both graphs $H$ and $H'$ belong to the same space-time diagram $\mathcal{M}_A(H_0)$ with $H_0$ containing a right-moving particle in $u$. They both contain vertex $v$. In $H'$, the particle that was present in $H$ has moved to point $w$. The state associated to $v$ thus depends on the way it is cut, which is captured by its set of incoming ports e.g. $\{a',b'\}$ versus just $\{a'\}$.\label{fig : We need consistency}
}
\end{figure}

The local rules is given by Fig.~\ref{fig : Dynamic example 1}. It allows us to evolve one graph $G$, understood as a space-like cut (a snapshot), into another later space-like cut $G''$, as in Fig.~\ref{fig:spacetimediagram}. But the point is that the graph $G$ can also be evolved asynchronously into $H$ or $H'$ as in Fig.~\ref{fig : We need consistency}. All of these graphs agree upon the trajectory of the particles.
More generally, this local rule is ``space-time deterministic'', i.e. it produces well-determined events in space-time. But notice that this property is a bit subtle to formulate, e.g. the state associated to the event $v$ seems different in $H$ and $H'$, as the particle got `consumed' from $v$ to $w$. The important point is that the state of every event $v$ (in terms of its internal state and connectivity) remains a function of how it is traversed in the space-like cut, represented here by its set of incoming ports. I.e. the example is fully consistent in the sense that the state of each vertex is fully determined by its set of incoming ports.\\ 
This `particle consumption mechanism' of this example is a design choice, willingly taken in order to illustrate three points: 1/ space-time determinism is really the idea that, \emph{given the angle of their space-like-cut}, all events are well-determined. E.g. in Physics the non-scalar quantum field associated to an event does depend on this angle. 2/ Ultimately this can be traced back to the nature of quantum information, which gets `consumed' as it gets `read out', because it cannot be `copied'. Our formalism is therefore compatible with the idea of linear evolutions, as will be required when we study reversible or quantum versions of these rules. 3/ The weak consistency criterion cares only for local normal vertices (aka Past), but in this example they always have $(0,0)$ as internal state (Fig.~\ref{fig : We need consistency}). This suggests that weak consistency is too weak a condition as it dangerously fails to capture the essence of this dynamical system: particles flying around. 
Altogether, this example shows that it is perfectly possible to consistently simulate a dynamical system by means of asynchronous rule applications. This sets the aim of the paper: to understand when asynchronism meets space-time determinism.

\section{Simulating synchronous cellular automata}\label{sec: Simulation}

Beyond the above particle system example, we now show that we can simulate any 1D cellular automata in spite of its synchonicity. 
The logic of this construction borrows to the well-studied `marching soldiers' scheme \cite{WeakConsistencyGacs}, as best formalised by \cite{MarchingNehaniv}
, but relies on the DAG of dependency rather than extra states as its mechanism for local, relative synchronisation.
Without loss of generality \cite{IbarraJiang}, we simulate radius half cellular automata.


\begin{figure}[t]
\centering
\begin{subfigure}{.3\textwidth}
    \captionsetup{justification=centering}
    \resizebox{\textwidth}{!}{\tikzset{every picture/.style={line width=0.75pt}} 

\begin{tikzpicture}[x=0.75pt,y=0.75pt,yscale=-1,xscale=1]

\draw [color={rgb, 255:red, 0; green, 0; blue, 0 }  ,draw opacity=1 ][line width=1.5]    (92.03,147.46) -- (212.35,147.46) ;
\draw [color={rgb, 255:red, 0; green, 0; blue, 0 }  ,draw opacity=1 ][line width=1.5]    (405,146.73) -- (212.35,147.46) ;
\draw  [color={rgb, 255:red, 0; green, 0; blue, 0 }  ,draw opacity=1 ][fill={rgb, 255:red, 255; green, 255; blue, 255 }  ,fill opacity=1 ][dash pattern={on 5.63pt off 4.5pt}][line width=1.5]  (80,147.46) .. controls (80,133.9) and (90.99,122.9) .. (104.55,122.9) .. controls (118.12,122.9) and (129.11,133.9) .. (129.11,147.46) .. controls (129.11,161.02) and (118.12,172.01) .. (104.55,172.01) .. controls (90.99,172.01) and (80,161.02) .. (80,147.46) -- cycle ;
\draw  [color={rgb, 255:red, 0; green, 0; blue, 0 }  ,draw opacity=1 ][fill={rgb, 255:red, 255; green, 255; blue, 255 }  ,fill opacity=1 ][dash pattern={on 5.63pt off 4.5pt}][line width=1.5]  (380,145) .. controls (380,131.19) and (391.19,120) .. (405,120) .. controls (418.81,120) and (430,131.19) .. (430,145) .. controls (430,158.81) and (418.81,170) .. (405,170) .. controls (391.19,170) and (380,158.81) .. (380,145) -- cycle ;
\draw [color={rgb, 255:red, 0; green, 0; blue, 0 }  ,draw opacity=1 ][line width=1.5]    (87.83,449.46) -- (208.15,449.46) ;
\draw [color={rgb, 255:red, 0; green, 0; blue, 0 }  ,draw opacity=1 ][line width=1.5]    (328.47,449.46) -- (208.15,449.46) ;
\draw [color={rgb, 255:red, 0; green, 0; blue, 0 }  ,draw opacity=1 ][line width=1.5]    (328.47,449.46) -- (420,450) ;
\draw  [color={rgb, 255:red, 0; green, 0; blue, 0 }  ,draw opacity=1 ][fill={rgb, 255:red, 255; green, 255; blue, 255 }  ,fill opacity=1 ][dash pattern={on 5.63pt off 4.5pt}][line width=1.5]  (380,445.41) .. controls (380,431.83) and (391.19,420.82) .. (405,420.82) .. controls (418.81,420.82) and (430,431.83) .. (430,445.41) .. controls (430,458.99) and (418.81,470) .. (405,470) .. controls (391.19,470) and (380,458.99) .. (380,445.41) -- cycle ;
\draw  [draw opacity=0][fill={rgb, 255:red, 208; green, 2; blue, 27 }  ,fill opacity=0.23 ] (54.72,448.39) -- (220.89,145.85) -- (271.18,145.85) -- (277.35,448.39) -- cycle ;
\draw  [color={rgb, 255:red, 208; green, 2; blue, 27 }  ,draw opacity=1 ][fill={rgb, 255:red, 208; green, 2; blue, 27 }  ,fill opacity=0.2 ] (160,301.19) .. controls (160,285.28) and (182.39,272.38) .. (210,272.38) .. controls (237.61,272.38) and (260,285.28) .. (260,301.19) .. controls (260,317.1) and (237.61,330) .. (210,330) .. controls (182.39,330) and (160,317.1) .. (160,301.19) -- cycle ;

\draw  [color={rgb, 255:red, 0; green, 0; blue, 0 }  ,draw opacity=1 ][fill={rgb, 255:red, 255; green, 255; blue, 255 }  ,fill opacity=1 ][line width=1.5]  (220.89,145.85) .. controls (220.89,132.51) and (231.88,121.7) .. (245.45,121.7) .. controls (259.01,121.7) and (270,132.51) .. (270,145.85) .. controls (270,159.19) and (259.01,170) .. (245.45,170) .. controls (231.88,170) and (220.89,159.19) .. (220.89,145.85) -- cycle ;
\draw  [color={rgb, 255:red, 0; green, 0; blue, 0 }  ,draw opacity=1 ][fill={rgb, 255:red, 255; green, 255; blue, 255 }  ,fill opacity=1 ][line width=1.5]  (75.8,449.46) .. controls (75.8,435.9) and (86.79,424.9) .. (100.35,424.9) .. controls (113.91,424.9) and (124.91,435.9) .. (124.91,449.46) .. controls (124.91,463.02) and (113.91,474.01) .. (100.35,474.01) .. controls (86.79,474.01) and (75.8,463.02) .. (75.8,449.46) -- cycle ;
\draw  [color={rgb, 255:red, 0; green, 0; blue, 0 }  ,draw opacity=1 ][fill={rgb, 255:red, 255; green, 255; blue, 255 }  ,fill opacity=1 ][line width=1.5]  (221.6,445.85) .. controls (221.6,432.51) and (232.6,421.7) .. (246.16,421.7) .. controls (259.72,421.7) and (270.71,432.51) .. (270.71,445.85) .. controls (270.71,459.19) and (259.72,470) .. (246.16,470) .. controls (232.6,470) and (221.6,459.19) .. (221.6,445.85) -- cycle ;

\draw (234,127.4) node [anchor=north west][inner sep=0.75pt]  [font=\sizehalfCA]  {$\sigma_{1}^{1}$};
\draw (87.6,432.4) node [anchor=north west][inner sep=0.75pt]  [font=\sizehalfCA]  {$\sigma_{0}^{0}$};
\draw (232.6,430.4) node [anchor=north west][inner sep=0.75pt]  [font=\sizehalfCA]  {$\sigma_{1}^{0}$};
\draw (198.14,279.61) node [anchor=north west][inner sep=0.75pt]  [font=\sizehalfCAb,color={rgb, 255:red, 208; green, 2; blue, 27 }  ,opacity=1 ]  {$f$};

\end{tikzpicture}}
    \caption{Radius one half locality \dots}
    \label{fig :one_half_radius_CA_1}
\end{subfigure}
\begin{subfigure}{.60\textwidth}
    \captionsetup{justification=centering}
    \resizebox{\textwidth}{!}{\input{figs/one_half_radius_CA_2}}
    \caption{\dots results in the above global evolution.}
    \label{fig:one_half_radius_CA_2}
\end{subfigure}
\caption{{\em Radius one half local rule (a):} A radius one half local rule $f$ takes two states as input ($\sigma_0^0,\sigma_1^0\in \Sigma_{CA}$) and outputs $\sigma_1^1 = f(\sigma_0^0,\sigma_1^0)$. {\em Global evolution (b):} The initial configuration is an infinite one dimensional array $(\sigma^0_i)_{i\in \mathbb{Z}}$. We consider the cellular automaton which apply $f$ homogeneously in space. In black we depicted the successive configurations computed by this automaton ($\sigma^i_k$ is just short for $f(\sigma^{i-1}_{k-1},\sigma^{i-1}_{k})$). In red we drew each time the local rule $f$ gets applied. We gave a name to each occurrence to introduce more easily the simulation scheme of Fig.\ref{fig:Simulation 2}. Finally we used red arrows to depict the dependencies between each occurrence of $f$, for example since $0.f_1$ needs the output $\sigma_1^1$ of $0.f_0$ to compute, we drew an arrow $0.f_0\rightarrow 0.f_1$.
}
\label{fig: one_half_radius_CA}
\end{figure}

\begin{figure}[h!]
\centering
\begin{subfigure}{.17\textwidth}
    \captionsetup{justification=centering}
    \resizebox{\textwidth}{!}{\tikzset{every picture/.style={line width=0.75pt}} 

\begin{tikzpicture}[x=0.75pt,y=0.75pt,yscale=-1,xscale=1]

\draw [color={rgb, 255:red, 0; green, 0; blue, 0 }  ,draw opacity=1 ][line width=2.25]    (226,274) -- (189.91,202) ;
\draw [shift={(187.67,197.53)}, rotate = 63.38] [fill={rgb, 255:red, 0; green, 0; blue, 0 }  ,fill opacity=1 ][line width=0.08]  [draw opacity=0] (14.29,-6.86) -- (0,0) -- (14.29,6.86) -- cycle    ;
\draw [line width=2.25]    (226,274) -- (264.27,204.09) ;
\draw [shift={(266.67,199.7)}, rotate = 118.69] [fill={rgb, 255:red, 0; green, 0; blue, 0 }  ][line width=0.08]  [draw opacity=0] (14.29,-6.86) -- (0,0) -- (14.29,6.86) -- cycle    ;
\draw  [color={rgb, 255:red, 38; green, 105; blue, 185 }  ,draw opacity=1 ][fill={rgb, 255:red, 255; green, 255; blue, 255 }  ,fill opacity=1 ][line width=3]  (200,274) .. controls (200,259.64) and (211.64,248) .. (226,248) .. controls (240.36,248) and (252,259.64) .. (252,274) .. controls (252,288.36) and (240.36,300) .. (226,300) .. controls (211.64,300) and (200,288.36) .. (200,274) -- cycle ;
\draw  [color={rgb, 255:red, 80; green, 227; blue, 194 }  ,draw opacity=1 ][fill={rgb, 255:red, 255; green, 255; blue, 255 }  ,fill opacity=1 ][line width=3]  (152,176) .. controls (152,162.75) and (162.75,152) .. (176,152) .. controls (189.25,152) and (200,162.75) .. (200,176) .. controls (200,189.25) and (189.25,200) .. (176,200) .. controls (162.75,200) and (152,189.25) .. (152,176) -- cycle ;
\draw  [color={rgb, 255:red, 80; green, 227; blue, 194 }  ,draw opacity=1 ][fill={rgb, 255:red, 255; green, 255; blue, 255 }  ,fill opacity=1 ][line width=3]  (252,176) .. controls (252,162.75) and (262.75,152) .. (276,152) .. controls (289.25,152) and (300,162.75) .. (300,176) .. controls (300,189.25) and (289.25,200) .. (276,200) .. controls (262.75,200) and (252,189.25) .. (252,176) -- cycle ;
\draw [color={rgb, 255:red, 80; green, 227; blue, 194 }  ,draw opacity=1 ][line width=3]    (176,152) -- (176,200) ;
\draw [color={rgb, 255:red, 80; green, 227; blue, 194 }  ,draw opacity=1 ][line width=3]    (276,152) -- (276,200) ;
\draw [color={rgb, 255:red, 74; green, 107; blue, 226 }  ,draw opacity=1 ][line width=3]    (226,250) -- (226,298) ;

\draw (258,282.4) node [anchor=north west][inner sep=0.75pt]  [font=\sizesimub]  {${0.f_{0}}$};
\draw (200,201.4) node [anchor=north west][inner sep=0.75pt]  [font=\sizesimup,color={rgb, 255:red, 155; green, 155; blue, 155 }  ,opacity=1 ]  {$:a'$};
\draw (230,200.4) node [anchor=north west][inner sep=0.75pt]  [font=\sizesimup,color={rgb, 255:red, 155; green, 155; blue, 155 }  ,opacity=1 ]  {$:b'$};
\draw (246,232.4) node [anchor=north west][inner sep=0.75pt]  [font=\sizesimup,color={rgb, 255:red, 155; green, 155; blue, 155 }  ,opacity=1 ]  {$:b$};
\draw (184,232.4) node [anchor=north west][inner sep=0.75pt]  [font=\sizesimup,color={rgb, 255:red, 155; green, 155; blue, 155 }  ,opacity=1 ]  {$:a$};
\draw (204.83,262.4) node [anchor=north west][inner sep=0.75pt]    {\sizesimu $\sigma_{0}^{0}$};
\draw (229.83,264) node [anchor=north west][inner sep=0.75pt]    {\sizesimu $\sigma_{1}^{0}$};
\draw (261,167.4) node [anchor=north west][inner sep=0.75pt]    {\sizesimu $\epsilon $};
\draw (182,166.4) node [anchor=north west][inner sep=0.75pt]    {\sizesimu $\epsilon $};

\end{tikzpicture}}
    \caption{$G_{\restri{f_0}}.$}
    \label{fig :Simulation 1 a}
\end{subfigure}
\begin{subfigure}{.17\textwidth}
    \captionsetup{justification=centering}
    \resizebox{\textwidth}{!}{\tikzset{every picture/.style={line width=0.75pt}} 

\begin{tikzpicture}[x=0.75pt,y=0.75pt,yscale=-1,xscale=1]

\draw [color={rgb, 255:red, 0; green, 0; blue, 0 }  ,draw opacity=1 ][line width=2.25]    (276,176) -- (239.95,105.88) ;
\draw [shift={(237.67,101.43)}, rotate = 62.79] [fill={rgb, 255:red, 0; green, 0; blue, 0 }  ,fill opacity=1 ][line width=0.08]  [draw opacity=0] (14.29,-6.86) -- (0,0) -- (14.29,6.86) -- cycle    ;
\draw [line width=2.25]    (176,176) -- (212.32,107.85) ;
\draw [shift={(214.67,103.43)}, rotate = 118.05] [fill={rgb, 255:red, 0; green, 0; blue, 0 }  ][line width=0.08]  [draw opacity=0] (14.29,-6.86) -- (0,0) -- (14.29,6.86) -- cycle    ;
\draw  [color={rgb, 255:red, 38; green, 105; blue, 185 }  ,draw opacity=1 ][fill={rgb, 255:red, 255; green, 255; blue, 255 }  ,fill opacity=1 ][line width=3]  (200,76) .. controls (200,61.64) and (211.64,50) .. (226,50) .. controls (240.36,50) and (252,61.64) .. (252,76) .. controls (252,90.36) and (240.36,102) .. (226,102) .. controls (211.64,102) and (200,90.36) .. (200,76) -- cycle ;
\draw  [color={rgb, 255:red, 80; green, 227; blue, 194 }  ,draw opacity=1 ][fill={rgb, 255:red, 255; green, 255; blue, 255 }  ,fill opacity=1 ][line width=3]  (152,176) .. controls (152,162.75) and (162.75,152) .. (176,152) .. controls (189.25,152) and (200,162.75) .. (200,176) .. controls (200,189.25) and (189.25,200) .. (176,200) .. controls (162.75,200) and (152,189.25) .. (152,176) -- cycle ;
\draw  [color={rgb, 255:red, 80; green, 227; blue, 194 }  ,draw opacity=1 ][fill={rgb, 255:red, 255; green, 255; blue, 255 }  ,fill opacity=1 ][line width=3]  (252,176) .. controls (252,162.75) and (262.75,152) .. (276,152) .. controls (289.25,152) and (300,162.75) .. (300,176) .. controls (300,189.25) and (289.25,200) .. (276,200) .. controls (262.75,200) and (252,189.25) .. (252,176) -- cycle ;
\draw [color={rgb, 255:red, 80; green, 227; blue, 194 }  ,draw opacity=1 ][line width=3]    (176,152) -- (176,200) ;
\draw [color={rgb, 255:red, 80; green, 227; blue, 194 }  ,draw opacity=1 ][line width=3]    (276,152) -- (276,200) ;
\draw [color={rgb, 255:red, 74; green, 107; blue, 226 }  ,draw opacity=1 ][line width=3]    (226,50) -- (226,98) ;

\draw (251,112.4) node [anchor=north west][inner sep=0.75pt]  [font=\sizesimup,color={rgb, 255:red, 155; green, 155; blue, 155 }  ,opacity=1 ]  {$:a'$};
\draw (179,101.4) node [anchor=north west][inner sep=0.75pt]  [font=\sizesimup,color={rgb, 255:red, 155; green, 155; blue, 155 }  ,opacity=1 ]  {$:b'$};
\draw (195,141.4) node [anchor=north west][inner sep=0.75pt]  [font=\sizesimup,color={rgb, 255:red, 155; green, 155; blue, 155 }  ,opacity=1 ]  {$:b$};
\draw (236,141.4) node [anchor=north west][inner sep=0.75pt]  [font=\sizesimup,color={rgb, 255:red, 155; green, 155; blue, 155 }  ,opacity=1 ]  {$:a$};
\draw (210,67.4) node [anchor=north west][inner sep=0.75pt]    {\sizesimu $\epsilon $};
\draw (233,67.4) node [anchor=north west][inner sep=0.75pt]    {\sizesimu $\epsilon $};
\draw (177.83,164) node [anchor=north west][inner sep=0.75pt]    {\sizesimu $\sigma_{1}^{1}$};
\draw (255.83,164) node [anchor=north west][inner sep=0.75pt]    {\sizesimu $\sigma_{1}^{1}$};
\draw (261,72.4) node [anchor=north west][inner sep=0.75pt]  [font=\sizesimub]  {$1{.f_{0}}$};

\end{tikzpicture}}
    \caption{$A_{f_0}G_{\restri{f_0}}.$}
    \label{fig :Simulation 1 b}
\end{subfigure}
\begin{subfigure}{.63\textwidth}
    \captionsetup{justification=centering}   
    \resizebox{\textwidth}{!}{\begin{tikzpicture}


\foreach \i in {1,2,...,6}{
    \foreach \j in {1,2}{
    	\FPeval{\posx}{clip(\i*8)}
        \FPeval{\posy}{clip(\j*8)}
        \ifnum \j=2
            {\ifnum \i=1
                {}
            \else
                \ifnum \i=6
                    {}
                \else
                    {\vertex{\posx}{\posy}{myblack}{0}{}{\i\j}}
                \fi
            \fi
            }
        \else
            {\ifnum \i<3
                {}
            \else
                {\vertex{\posx}{\posy}{myblack}{0}{}{\i\j}}
            \fi
            }
        \fi
    }
}

\foreach \i in {1,2,...,5}{
    \foreach \j in {1,2}{
    	\FPeval{\posx}{clip(\i*8+4)}
        \FPeval{\posy}{clip(\j*8+4)}
        \ifnum \j=1
            {\ifnum \i=1
                {}
            \else
                {\vertex{\posx}{\posy}{myblack}{0}{}{S\i\j}}
            \fi
            }
        \else
            \ifnum \i=5
                {}
            \else
                {\vertex{\posx}{\posy}{myblack}{0}{}{S\i\j}}
            \fi
        \fi
    }
}


\foreach \i in {1,2,...,5}{
    \FPeval{\k}{clip(\i+1)}
    \ifnum \i=1
        {\edge{\k2}{S\i2}{myblack}}
    \else
        \ifnum \i=2
            {\edge{\k1}{S\i1}{myblack}
            \edge{S\i1}{\i2}{border}
            \edge{S\i1}{\k2}{border}
            \edge{\i2}{S\i2}{myblack}
            \edge{\k2}{S\i2}{myblack}}
        \else
            \ifnum \i=5
                {\edge{\i1}{S\i1}{myblack}
                \edge{\k1}{S\i1}{myblack}
                \edge{S\i1}{\i2}{border}}
            \else
                \edge{\i1}{S\i1}{myblack}
                \edge{\k1}{S\i1}{myblack}
                \edge{S\i1}{\i2}{border}
                \edge{S\i1}{\k2}{border}
                \edge{\i2}{S\i2}{myblack}
                \edge{\k2}{S\i2}{myblack}
            \fi
        \fi
    \fi
}


\draw (26.3,8) node{{\sizebigsimub $0.f_0$}};
\draw (34.3,8) node{{\sizebigsimub $0.f_2$}};
\draw (42.3,8) node{{\sizebigsimub $0.f_4$}};
\draw (50.3,8) node{{\sizebigsimub $0.f_6$}};
\draw (22.3,12) node{{\sizebigsimub $0.f_{-1}$}};
\draw (30.3,12) node{{\sizebigsimub $0.f_{1}$}};
\draw (38.3,12) node{{\sizebigsimub $0.f_{3}$}};
\draw (46.3,12) node{{\sizebigsimub $0.f_{5}$}};
\draw (18.3,16) node{{\sizebigsimub $1.f_{-2}$}};
\draw (26.3,16) node{{\sizebigsimub $1.f_{0}$}};
\draw (34.3,16) node{{\sizebigsimub $1.f_{2}$}};
\draw (42.3,16) node{{\sizebigsimub $1.f_{4}$}};
\draw (14.3,20) node{{\sizebigsimub $1.f_{-3}$}};
\draw (22.3,20) node{{\sizebigsimub $1.f_{-1}$}};
\draw (30.3,20) node{{\sizebigsimub $1.f_{1}$}};
\draw (38.3,20) node{{\sizebigsimub $1.f_{3}$}};


\draw (24.4,8) node{{\sizebigsimu $\sigma_1^0$}};
\draw (32.4,8) node{{\sizebigsimu $\sigma_2^0$}};
\draw (40.4,8) node{{\sizebigsimu $\sigma_3^0$}};
\draw (48.4,8) node{{\sizebigsimu $\sigma_4^0$}};
\draw (20.4,12) node{{\sizebigsimu $\epsilon$}};
\draw (28.4,12) node{{\sizebigsimu $\epsilon$}};
\draw (36.4,12) node{{\sizebigsimu $\epsilon$}};
\draw (44.4,12) node{{\sizebigsimu $\epsilon$}};
\draw (16.4,16) node{{\sizebigsimu $\sigma_1^2$}};
\draw (24.4,16) node{{\sizebigsimu $\sigma_2^2$}};
\draw (32.4,16) node{{\sizebigsimu $\sigma_3^2$}};
\draw (40.4,16) node{{\sizebigsimu $\sigma_4^2$}};
\draw (12.4,20) node{{\sizebigsimu $\epsilon$}};
\draw (20.4,20) node{{\sizebigsimu $\epsilon$}};
\draw (28.4,20) node{{\sizebigsimu $\epsilon$}};
\draw (36.4,20) node{{\sizebigsimu $\epsilon$}};

\draw (23.5,8) node{{\sizebigsimu $\sigma_0^0$}};
\draw (31.5,8) node{{\sizebigsimu $\sigma_1^0$}};
\draw (39.5,8) node{{\sizebigsimu $\sigma_2^0$}};
\draw (47.5,8) node{{\sizebigsimu $\sigma_3^0$}};
\draw (19.5,12) node{{\sizebigsimu $\epsilon$}};
\draw (27.5,12) node{{\sizebigsimu $\epsilon$}};
\draw (35.5,12) node{{\sizebigsimu $\epsilon$}};
\draw (43.5,12) node{{\sizebigsimu $\epsilon$}};
\draw (15.5,16) node{{\sizebigsimu $\sigma_0^2$}};
\draw (23.5,16) node{{\sizebigsimu $\sigma_1^2$}};
\draw (31.5,16) node{{\sizebigsimu $\sigma_2^2$}};
\draw (39.5,16) node{{\sizebigsimu $\sigma_3^2$}};
\draw (11.5,20) node{{\sizebigsimu $\epsilon$}};
\draw (19.5,20) node{{\sizebigsimu $\epsilon$}};
\draw (27.5,20) node{{\sizebigsimu $\epsilon$}};
\draw (35.5,20) node{{\sizebigsimu $\epsilon$}};

\draw (10.5,17) node[right]{{\Huge \scalefont{2.5} $G'$}};
\draw (18.5,9) node[right]{{\Huge \scalefont{2.5} $G$}};
    
\end{tikzpicture}}
    \caption{Initial graph and space-time diagram.}
    \label{fig:Simulation 2}
\end{subfigure}
\caption{{\em The local rule for cellular automata simulation (a) and (b):} $A_{f_0}$ acts exactly as $A_x$ in Fig.~\ref{fig : Dynamic example 1}, except for the internal states. If $\sigma(0.f_0) = (\sigma_0^0,\sigma_1^0)$ does not contain $\epsilon$, it computes $\sigma_1^1 = f(\sigma_0^0,\sigma_1^0)$ and stores the result in both neighbours. {\em Space-time diagram of the simulation (c)} of the cellular automaton in Fig.~\ref{fig:one_half_radius_CA_2}. The highlighted $G$ encodes the initial cellular automaton configuration $\sigma^0$ of the CA, whilst $G'$ encodes $\sigma^2$ the configuration obtained after two time steps of the cellular automaton.
}
\end{figure}

First let us recall that a radius one half cellular automaton acting on an alphabet $\Sigma_{CA}$ is a global function $F:{\Sigma_{CA}}^{\mathbb{Z}}\to {\Sigma_{CA}}^{\mathbb{Z}}$ defined by the synchronous application of a local function $f: {\Sigma_{CA}}^2\to {\Sigma_{CA}}$ everywhere at once, cf. Fig.~\ref{fig :one_half_radius_CA_1}.
The dependencies between the different applications of $f$ are shown in Fig.~\ref{fig:one_half_radius_CA_2}, notice the similarity with Fig.~\ref{fig:spacetimediagram}. This suggests that the natural way to encode this cellular automaton in our model is to use vertices (a.k.a events) to represent each application of $f$.

Second we define $\Sigma = (\Sigma_{CA} \cup \epsilon)^2$ and pick the same neighbourhood scheme and ports as in the particle system example. We define the local rule as depicted in Fig.~\ref{fig :Simulation 1 a} and Fig.~\ref{fig :Simulation 1 b}. Its action can be understood as follow: 1/ applied to the vertex named $0.f_0$, it applies $f(\sigma_0^0,\sigma_1^0)$ 2/ and stores the result both in the right part of the vertex $0.f_{-1}$ and in the left part of the vertex $0.f_1$. 3/ Finally it creates the vertex $1.f_0$, which is awaiting for the output of the $0.f_{-1}$ and $0.f_1$ applications, as represented by the two incoming edges.

We encode the initial configuration $\sigma^0$ by a graph $G$ (depicted in Fig.~\ref{fig:Simulation 2}) which contains two copies of each internal state in $\sigma^0$. The dynamics will indeed generate every configuration $\sigma^i=F^i(\sigma^0)$. The space-time diagram is similar to that of Fig.~\ref{fig:spacetimediagram}.
The scheme can easily be generalised to $d$-dimensional cellular automata, and is likely to work for any local synchronous evolution in a broad sense ; the idea being to use vertices to encode applications of the local rule, and edges to represent how these applications are causally related.

\section{Beyond synchronous simulation}\label{sec:exV2}

The previous examples are simulations of synchronous systems. Beyond these, we are interested in expressing systems that remain deterministic but are genuinely asynchronous.
Such local rule still have consistent space-time diagrams: the only thing which is lost is the possibly to interprete these through the lens of a global clock.
A paradigmatic, physics-inspired example for this whole class of systems is the following time dilation example.

The time dilation example extends the internal state space $\Sigma$ used in Sec.~\ref{sec:exV1} with two more states: the \textrm{green} and \textrm{red} states. The same local rule is extended so that if one such green/red particle is found at position $x$, it will stay there, oscillating between both colours and altering the very texture of space-time, cf. Fig.~\ref{fig : Dynamic example 2}. This results in time dilation as can be seen from Fig.~\ref{fig : Space time diagram example 2}.
For instance, if two identically made clocks were modeled out of a signal oscillating from left to right and right to left between two neighbouring nodes, the clock lying on the right-hand-side of the green/red particle would tick twice as slower than the one lying on the left-hand-side. Yet, the very same local rule is being applied left and right of the green/red particle. Such a phenomenon is reminiscent of general relativity, e.g. time flows slightly slower on Earth than it does in the stratosphere, as measured by identically made atomic clocks. Yet, the same laws of Physics apply in the stratosphere and on Earth. How did we get there? 

We argue that, to some extent, the construction of this model of computation mimics some of the key steps of the derivation of general relativity theory from physical symmetries---as found in standard textbooks \cite{Gravitation}. The developped analogy can safely be skipped by the reader with lesser interest in Physics. Indeed, let us remind the reader that the textbook derivation proceeds by: 1/ Assuming the existence of a well-determined space-time. 2/ Requiring covariance, i.e. invariance under changes of coordinates, which implies a form of asynchronism as one can choose coordinates whereby one region of a space-like cut will evolve (large time lapse), but not the other (small time lapse). 3/ Concluding that in order to obtain covariance, one needs to provide extra causality structure at each point, namely the metric field. 4/ Assuming background-invariance, namely enabling the possibility that space-time be curved by the presence of this newly allowed metric. 5/ Providing a dynamic upon the metric itself. \\
Here, in the discrete, we: 1/ enforced a well-determined space-time, 2/ in spite of an asynchronous evaluation strategy, 3/ by introducing an extra causality structure, namely a DAG of dependencies. 4/ We then allowed ourselves to consider graphs with exotic such DAG, and 5/ rules manipulating them. 

\begin{figure}[h]
\centering
\begin{subfigure}{0.17\textwidth}
    \captionsetup{justification=centering}
    \resizebox{\textwidth}{!}{\begin{tikzpicture}


\redcellinternal{5}{5}{u};
\emptycell{1}{9}{setBorder}{v};
\emptycell{9}{13}{setBorder}{w};


\edge{u}{v}{myblack};
\edge{u}{w}{myblack};


\draw (6,4) node[right]{{\sizeFsix $u=t.x$}};
\draw (2.25,9) node[right]{{\sizeFsix $v$}};
\draw (6.5,13.5) node[right]{{\sizeFsix $w$}};

\draw (1.5,9) node {{\sizeFsixs $i$}};
\draw (8.5,13) node {{\sizeFsixs $j$}};

\draw (3.25,5.5) node[color = port] {{\sizeFsixp $:a$}};
\draw (6.75,6.5) node[color = port]  {{\sizeFsixp $:b$}};
\draw (3,8) node[color = port]  {{\sizeFsixp $:a'$}};
\draw (7,11) node[color = port] {{\sizeFsixp $:b'$}};

\end{tikzpicture}}
    \caption{$G_{\restri{x}}$.}
    \label{fig : Dynamic example 2; G}
\end{subfigure}
\hfill
\begin{subfigure}{0.17\textwidth}
    \captionsetup{justification=centering}
    \resizebox{\textwidth}{!}{\begin{tikzpicture}


\greencellinternal{5}{13}{u};
\emptycell{1}{9}{setBorder}{v};
\emptycell{9}{13}{setBorder}{w};


\edge{v}{u}{myblack};
\edge{u}{w}{myblack};


\draw (1,15) node[right]{{\sizeFsix $(t+1).x$}};
\draw (1.75,8) node[right]{{\sizeFsix $v$}};
\draw (8.5,11) node[right]{{\sizeFsix $w$}};


\draw (3,10.25) node[color = port] {{\sizeFsixp $:a$}};
\draw (6.5,13.5) node[color = port]  {{\sizeFsixp $:b$}};
\draw (3,12) node[color = port]  {{\sizeFsixp $:a'$}};
\draw (7.5,12.25) node[color = port] {{\sizeFsixp $:b''$}};

\end{tikzpicture}}
    \caption{$A_{x}(G_{\restri{x}})$.}
    \label{fig : Dynamic example 2; A_x G}
\end{subfigure}
\hfill
\begin{subfigure}{0.17\textwidth}
    \captionsetup{justification=centering}
    \resizebox{\textwidth}{!}{\begin{tikzpicture}


\greencellinternal{5}{5}{u};
\emptycell{1}{9}{setBorder}{v};
\emptycell{9}{9}{setBorder}{w};


\edge{u}{v}{myblack};
\edge{u}{w}{myblack};


\draw (6,4) node[right]{{\sizeFsix $u=t.x$}};
\draw (2.25,9) node[right]{{\sizeFsix $v$}};
\draw (6.25,9) node[right]{{\sizeFsix $w$}};

\draw (1.5,9) node {{\sizeFsixs $i$}};
\draw (8.5,9) node {{\sizeFsixs $j$}};

\draw (3.25,5.5) node[color = port] {{\sizeFsixp $:a$}};
\draw (6.5,5.5) node[color = port]  {{\sizeFsixp $:b$}};
\draw (3,8) node[color = port]  {{\sizeFsixp $:a'$}};
\draw (6.75,8) node[color = port] {{\sizeFsixp $:b'$}};

\end{tikzpicture}}
    \caption{$H_{\restri{x}}$.}
    \label{fig : Dynamic example 2; G green}
\end{subfigure}
\hfill
\begin{subfigure}{0.17\textwidth}
    \captionsetup{justification=centering}
    \resizebox{\textwidth}{!}{\begin{tikzpicture}


\redcellinternal{5}{13}{u};
\emptycell{1}{9}{setBorder}{v};
\emptycell{9}{9}{setBorder}{w};


\edge{v}{u}{myblack};
\edge{w}{u}{myblack};


\draw (6,13.5) node[right]{{\sizeFsix $(t+1).x$}};
\draw (2.25,9) node[right]{{\sizeFsix $v$}};
\draw (6.25,9) node[right]{{\sizeFsix $w$}};

\draw (1.5,9) node {{\sizeFsixs $i$}};
\draw (8.5,9) node {{\sizeFsixs $j$}};

\draw (3,10.25) node[color = port] {{\sizeFsixp $:b$}};
\draw (6.75,10.25) node[color = port]  {{\sizeFsixp $:a$}};
\draw (3,12) node[color = port]  {{\sizeFsixp $:b'$}};
\draw (6.75,12) node[color = port] {{\sizeFsixp $:a'$}};

\end{tikzpicture}}
    \caption{$A_{x}(H_{\restri{x}})$.}
    \label{fig : Dynamic example 2; A_x G green}
\end{subfigure}

\caption{{\em The local rule for time dilation.
} 
We define here the behaviour of $A_{x}$ when $u=t.x$ is colored---i.e. $\sigma(u)= green \lor \sigma(u) =red$. In both cases particles get destroyed when reaching the colored vertex and we flip the color in $x$. On a green vertex ($c$ \& $d$), $A_x$ behaves as we are used to. On a red vertex ($a$ \& $b$) $A_x$ creates an anomaly. The edge between $(t+1).x$ and $w$ is reversed, thus we will be forced to apply $A_x$ again before updating $w$. Note that this evolution is still port-decreasing (see Def.\ref{definition : Port Decreasing}) because $b'$ is replaced by a smaller port $b''$.
}
\label{fig : Dynamic example 2}
\end{figure}
\begin{figure}[t]
\centering
\resizebox{0.75\textwidth}{!}{\begin{tikzpicture}



\foreach \i in {1,2,3}{
    \foreach \j in {1,2,3}{
    	\FPeval{\posx}{clip(\i*8+4)}
        \FPeval{\posy}{clip(\j*8)}
        \ifnum \j=2
            {\vertex{\posx}{\posy}{border}{0}{\j}{\i\j}}
        \else
            {\vertex{\posx}{\posy}{myblack}{0}{\j}{\i\j}}
        \fi
    }
}

\foreach \i in {1,2,3,4}{
    \foreach \j in {1,2,3}{
    	\FPeval{\posx}{clip(\i*8)}
        \FPeval{\posy}{clip(\j*8+4)}
        \ifnum \j=2
            {\ifnum \i=4
            {\vertex{\posx}{\posy}{border}{0}{}{S\i\j}}
            \else
            \vertex{\posx}{\posy}{border}{0}{\j}{S\i\j}
            \fi}
        \else
            {\vertex{\posx}{\posy}{myblack}{0}{\j}{S\i\j}}
        \fi
    }
}


\foreach \i in {1,2,3}{
    \foreach \j in {1,2}{
    	\FPeval{\posx}{clip(\i*8+4*8-4+1.5)}
        \FPeval{\posy}{clip(\j*16-12)}
        {\vertex{\posx}{\posy}{myblack}{0}{\j}{Big\i\j}}
    }
}

\foreach \i in {1,2,3}{
    \foreach \j in {1,2}{
    	\FPeval{\posx}{clip(\i*8+4*8+1.5)}
        \FPeval{\posy}{clip(\j*16-4)}
        {\vertex{\posx}{\posy}{myblack}{0}{\j}{BigS\i\j}}
    }
}


\foreach \i in {1,2,3}{
    \FPeval{\k}{clip(\i+1)}
    \edge{\i1}{S\i1}{myblack}
    \edge{\i1}{S\k1}{myblack}
    
    \edge{S\i1}{\i2}{border}
    \edge{S\k1}{\i2}{border}
    \edge{\i2}{S\i2}{border}
    \edge{\i2}{S\k2}{border}
    
    \edge{S\i2}{\i3}{border}
    \edge{S\k2}{\i3}{border}
    \edge{\i3}{S\i3}{myblack}
    \edge{\i3}{S\k3}{myblack}
}

\foreach \i in {2,3}{
    \FPeval{\k}{clip(\i-1)}
    \edge{Big\i1}{BigS\i1}{myblack}
    \edge{Big\i1}{BigS\k1}{myblack}
    
    \edge{BigS\i1}{Big\i2}{border}
    \edge{BigS\k1}{Big\i2}{border}
    \edge{Big\i2}{BigS\i2}{myblack}
    \edge{Big\i2}{BigS\k2}{myblack}
}
\edge{Big11}{BigS11}{myblack}
\edge{Big11}{S41}{myblack}

\edge{BigS11}{Big12}{border}
\edge{S41}{Big12}{border}
\edge{S42}{Big12}{border}
\edge{Big12}{BigS12}{myblack}
\edge{Big12}{S43}{myblack}

\redcell{32}{12}{first}
\greencellgray{32}{20}{second}
\redcell{32}{28}{third}

\draw (7.5,24.5) node[right]{{\Huge \scalefont{2.5} $H'$}};
\draw (7.5,8.5) node[right]{{\Huge \scalefont{2.5} $H$}};
    
\end{tikzpicture}}
\caption{{\em Time dilation example.} In black we highlight two graphs $H$ and $H'$ belonging to the space-time diagram. We start with two particles, one on the left and the other on the right of the red. Notice how, even if the same local rule gets applied everywhere, times flows twice faster for the particle on the left. 
\label{fig : Space time diagram example 2}}
\end{figure}

\section{Obtaining space-time determinism}\label{sec:consistency}

The space-time determism of a local rule $A_{(-)}$ is the idea that the graphs that it generates from a seed $G$, which altogether form a space-time diagram $\mathcal{M}_A(G)$, are consistent between one another as regards the events that they describe. One might have hoped for a simple definition of this notion of consistency of $\mathcal{M}_A(G)$, whereby any two graphs of the space-time diagram must agree on the state of a vertex $u \in \calV$, if it so happens to appear in both of them. But the example in Sec.~\ref{sec:exV1} (Fig.~\ref{fig : We need consistency}) shows that things are more subtle. Its discussion motivates a definition of (full) space-time determinism whereby any two graphs of $\mathcal{M}$ must agree on the state of a vertex $v$ (in terms of its neighbourhood and internal state) whenever they agree on the set of incoming ports to that vertex $v$. In particular this implies that the state of $v$ should be the same for every graph of $\mathcal{M}$ such that $v\in \Past(G)$, which can be understood as stating that the ``result state'' (a.k.a. normal form) at $v$ is well-determined. We refer to this weaker demand as weak consistency (see Fig.~\ref{fig : consistency}).  

\begin{definition}[Consistency]\label{def : consistency}
Two graphs $G$ and $H$ are \emph{consistent} iff for all $v\in \I_H\cap \I_G$:
    $$\pof \Em_{G}(v)=\pof \Em_{H}(v) \Longrightarrow G_{v}= H_{v}.$$
where $\Em_{G}(v)$ is the set of edges ending in $v$, and $\pof \Em_{G}(v)$ denotes the set of incoming ports of $v$, i.e. $\pof E := \{\, b \mid (u\!:\!a,v\!:\!b) \in E \,\}$. Consistency is denoted $G\varparallel H$. The graphs $G$ and $H$ are called weakly consistent if they respect this condition in the special case where $\Em_{G}(v)= \emptyset$.
We say that a space-time diagram $\mathcal{M}$ is \emph{(resp. weakly) consistent}  if each pair of graphs in $\mathcal{M}$ is (resp. weakly) consistent.
Finally a local rule $A_{(-)}$ is said to be \emph{(weakly) space-time deterministic} if and only if all graphs $G$, $\mathcal{M}_A(G)$ is (weakly) consistent.
\end{definition}

\begin{figure}[t]
\hfil
\begin{subfigure}{0.15\textwidth}
    \captionsetup{justification=centering}
    \resizebox{\textwidth}{!}{\begin{tikzpicture}


\emptycell{5}{5}{myblack}{u};
\rightleft{1}{9}{myblack}{v};
\emptycell{9}{9}{myblack}{w};
\emptycell{5}{13}{myblack}{z};


\edge{u}{v}{myblack};
\edge{u}{w}{myblack};
\edge{w}{z}{myblack};
\edge{v}{z}{myblack};


\draw (6,4) node[right]{{\sizeFhuit $u$}};
\draw (6.75,9) node[right]{{\sizeFhuit $v$}};

\draw (3.25,5.5) node[color = port] {{\sizeFhuitp $:a$}};
\draw (6.5,5.5) node[color = port]  {{\sizeFhuitp $:b$}};

\draw (3,8) node[color = port]  {{\sizeFhuitp $:b$}};
\draw (6.75,8) node[color = port] {{\sizeFhuitp $:a$}};

\draw (3,10) node[color = port] {{\sizeFhuitp $:a$}};
\draw (7,10) node[color = port]  {{\sizeFhuitp $:b$}};

\draw (3,12) node[color = port] {{\sizeFhuitp $:b$}};
\draw (6.75,12) node[color = port]  {{\sizeFhuitp $:a$}};

\end{tikzpicture}}
    \caption{$G$}
    \label{fig : consistency 1}
\end{subfigure}
\hfil
\hfil
\begin{subfigure}{0.15\textwidth}
    \captionsetup{justification=centering}
    \resizebox{\textwidth}{!}{\begin{tikzpicture}


\emptycell{5}{5}{myblack}{u};
\rightleft{1}{9}{myblack}{v};
\emptycell{9}{9}{myblack}{w};
\rightleft{5}{13}{myblack}{z};


\edge{u}{v}{myblack};
\edge{u}{w}{myblack};
\edge{v}{z}{myblack};


\draw (6,4) node[right]{{\sizeFhuit $u$}};
\draw (6.75,9) node[right]{{\sizeFhuit $v$}};

\draw (3.25,5.5) node[color = port] {{\sizeFhuitp $:a$}};
\draw (6.5,5.5) node[color = port]  {{\sizeFhuitp $:b$}};

\draw (3,8) node[color = port]  {{\sizeFhuitp $:b$}};
\draw (6.75,8) node[color = port] {{\sizeFhuitp $:a$}};

\draw (3,10) node[color = port] {{\sizeFhuitp $:a$}};

\draw (3,12) node[color = port] {{\sizeFhuitp $:b$}};

\end{tikzpicture}}
    \caption{$H$}
    \label{fig : consistency 2}
\end{subfigure}
\hfil
\caption{{\em Weak consistency versus consistency.} The set of graphs $\{G,H\}$ is weakly consistent but not consistent. Indeed we have $G_{u} = H_{u}$ but consistency fails in $v$ because we have $\Em_{G}(v)=\Em_{H}(v) = \{\, :\!a\,\}$ but $G_{v} \neq H_{v}$.
}
\label{fig : consistency}
\end{figure}

We now embark in the quest for a set of properties ensuring that an $\restri{}$-local rule $A_{(-)}$ generates only consistent space-time diagrams.
We start with two properties.
The first one asks for timetags to only increase. 
The second one, akin to strong confluence or sequential independence in parallel graph transformations, states that a set of independent rule applications on a graph $G$ applied in any order should always lead to the same graph.
\begin{definition}[Time-increasing commutative local rules] A local rule $A_{(-)}$ is
\begin{itemize}
    \item \emph{time-increasing} 
    iff $\forall t.y \in V_G,\,\forall t'.y\in V_{A_x G},\;t \leq t'\text{, with }t<t' \textit{ if } y=x$;
    \item \emph{commutative} iff $\forall x,y \in \X(\Past( G)),\;A_x A_y (G) = A_y A_x (G)$.
\end{itemize}
\end{definition}
These properties place strong constraints on past vertices of a space-time diagram $\mathcal{M}_A(G)$: they entail weak consistency. Moreover each space-like cut is determined by its set of past elements.

\begin{restatable}{proposition}{weakconsistency}{\bf (Obtaining weak consistency).}\label{prop:weakconsistency}
    Let $A_{(-)}$ be a time-increasing commutative local rule. For all graphs $G$, $\mathcal{M}_A(G)$ is weakly consistent.
\end{restatable}
\withproofs{
\begin{proof}
    See Sec.\ref{proof : weak consistency}.
\end{proof}}

\begin{restatable}{proposition}{unicut}{\bf (Pasts determine space-like cuts)}\label{lemma : Unicity of space-time cuts}
    Let $A_{(-)}$ be a time-increasing commutative local rule. Let $G$ be a graph. Let $H,J\in \mathcal{M}_A(G)$. If $\Past(J) = \Past(H)$ then $J=H$.
\end{restatable}
\withproofs{
\begin{proof}
    See Sec.\ref{proof : unicut}.
\end{proof}}

\begin{figure}
\hfil
\begin{subfigure}{0.27\textwidth}
    \resizebox{0.75\textwidth}{!}{\tikzset{every picture/.style={line width=0.75pt}} 

\begin{tikzpicture}[x=0.75pt,y=0.75pt,yscale=-1,xscale=1]

\draw   (117.6,152.15) .. controls (117.6,75.24) and (179.94,12.9) .. (256.85,12.9) .. controls (333.76,12.9) and (396.1,75.24) .. (396.1,152.15) .. controls (396.1,229.06) and (333.76,291.4) .. (256.85,291.4) .. controls (179.94,291.4) and (117.6,229.06) .. (117.6,152.15) -- cycle ;
\draw  [color={rgb, 255:red, 248; green, 231; blue, 28 }  ,draw opacity=1 ][line width=2.25]  (225.97,131.93) .. controls (247.62,131.64) and (265.55,159.56) .. (266.01,194.31) .. controls (266.48,229.05) and (249.31,257.45) .. (227.66,257.74) .. controls (206.02,258.03) and (188.09,230.1) .. (187.63,195.36) .. controls (187.16,160.62) and (204.33,132.22) .. (225.97,131.93) -- cycle ;
\draw  [color={rgb, 255:red, 245; green, 166; blue, 35 }  ,draw opacity=1 ][line width=2.25]  (280.71,92.12) .. controls (310.24,91.72) and (334.66,128.21) .. (335.27,173.61) .. controls (335.88,219.02) and (312.45,256.15) .. (282.92,256.55) .. controls (253.4,256.94) and (228.98,220.46) .. (228.37,175.05) .. controls (227.76,129.65) and (251.19,92.52) .. (280.71,92.12) -- cycle ;
\draw  [color={rgb, 255:red, 208; green, 2; blue, 27 }  ,draw opacity=1 ][line width=2.25]  (254.25,42.65) .. controls (305.57,41.96) and (347.87,93.64) .. (348.73,158.08) .. controls (349.6,222.53) and (308.7,275.33) .. (257.39,276.02) .. controls (206.07,276.71) and (163.77,225.02) .. (162.9,160.58) .. controls (162.04,96.14) and (202.94,43.34) .. (254.25,42.65) -- cycle ;
\draw  [color={rgb, 255:red, 248; green, 231; blue, 28 }  ,draw opacity=1 ][line width=3]  (197,227.7) .. controls (197,218.48) and (210.47,211) .. (227.08,211) .. controls (243.7,211) and (257.17,218.48) .. (257.17,227.7) .. controls (257.17,236.92) and (243.7,244.4) .. (227.08,244.4) .. controls (210.47,244.4) and (197,236.92) .. (197,227.7) -- cycle ;
\draw  [color={rgb, 255:red, 245; green, 166; blue, 35 }  ,draw opacity=1 ][line width=3]  (248.17,223.4) .. controls (248.17,211.8) and (263.84,202.4) .. (283.17,202.4) .. controls (302.5,202.4) and (318.17,211.8) .. (318.17,223.4) .. controls (318.17,235) and (302.5,244.4) .. (283.17,244.4) .. controls (263.84,244.4) and (248.17,235) .. (248.17,223.4) -- cycle ;
\draw  [color={rgb, 255:red, 208; green, 2; blue, 27 }  ,draw opacity=1 ][line width=3]  (191.58,227.8) .. controls (191.58,207.7) and (220.82,191.4) .. (256.88,191.4) .. controls (292.93,191.4) and (322.17,207.7) .. (322.17,227.8) .. controls (322.17,247.9) and (292.93,264.2) .. (256.88,264.2) .. controls (220.82,264.2) and (191.58,247.9) .. (191.58,227.8) -- cycle ;

\draw (128,9.4) node [anchor=north west][inner sep=0.75pt]  [font=\huge]  {$\mathcal{C}$};
\draw (225,57.4) node [anchor=north west][inner sep=0.75pt]  [font=\huge]  {$\textcolor[rgb]{0.82,0.01,0.11}{N_{\omega }( G)}$};
\draw (242,117.4) node [anchor=north west][inner sep=0.75pt]  [font=\LARGE]  {$\textcolor[rgb]{0.96,0.65,0.14}{N_{\beta }( A_{\alpha } G)}$};
\draw (195,154.4) node [anchor=north west][inner sep=0.75pt]  [font=\LARGE]  {$\textcolor[rgb]{0.97,0.91,0.11}{N_{\alpha }( G)}$};
\draw (278,213.4) node [anchor=north west][inner sep=0.75pt]  [font=\LARGE]  {$\textcolor[rgb]{0.96,0.65,0.14}{\beta }$};
\draw (216,220.4) node [anchor=north west][inner sep=0.75pt]  [font=\LARGE]  {$\textcolor[rgb]{0.97,0.91,0.11}{\alpha }$};
\draw (243,242.4) node [anchor=north west][inner sep=0.75pt]  [font=\LARGE]  {$\textcolor[rgb]{0.82,0.01,0.11}{\omega }$};

\end{tikzpicture}}
    \caption{Monotony.}
    \label{fig : Monotony 1}
\end{subfigure}
\hfil
\begin{subfigure}{0.48 \textwidth}
    \resizebox{0.75\textwidth}{!}{\input{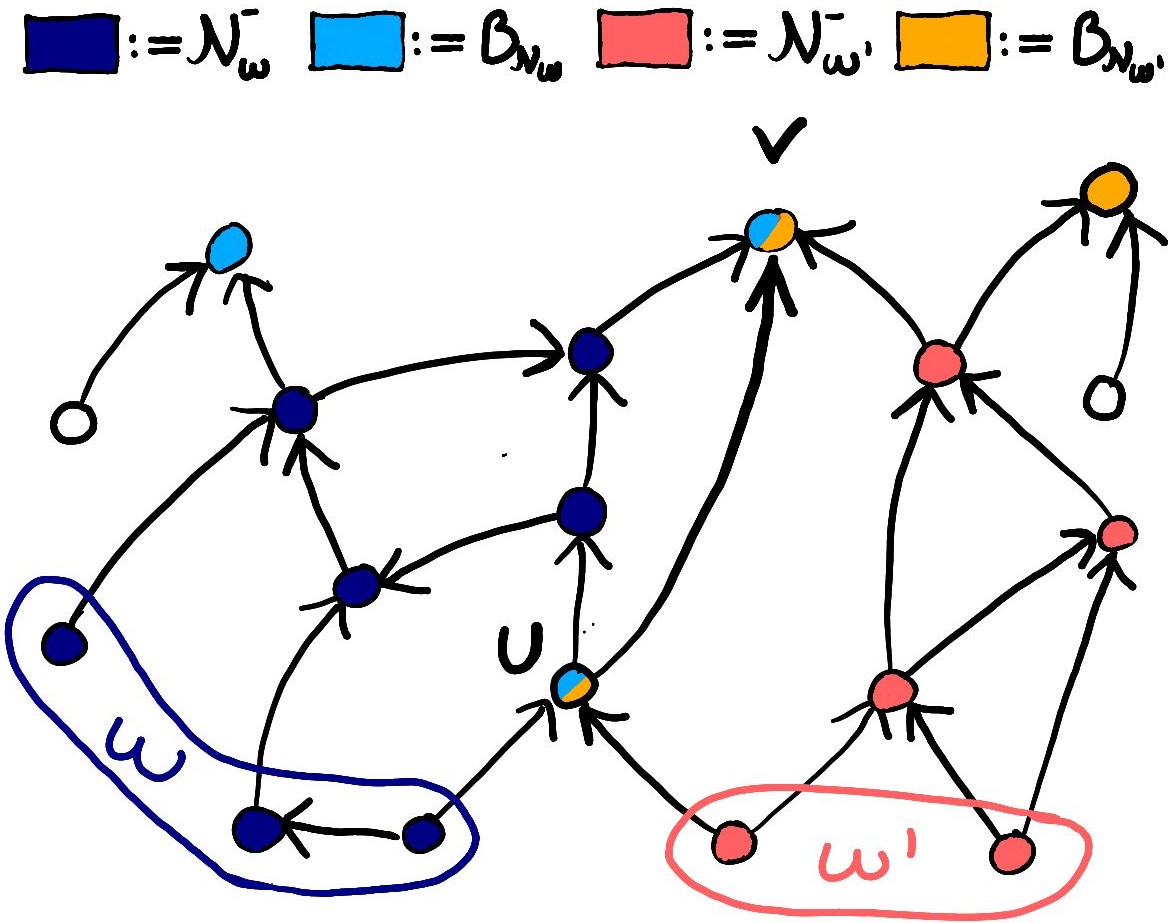}}
    \caption{Privacy and port decreasing.}
\end{subfigure}     
\hfil
\caption{{\em Monotony and privacy.} $(a)$ Monotony demands that, given a graph $G$ and a list of vertices $\omega = \beta \alpha$, the neighbourhood computed at the beginning $\restri{\omega}(G)$ (red) is larger   that $\restri{}_\alpha(G)$ (yellow) and that any future neighbourhood $\restri{}_\beta(A_\alpha G)$ (orange). Given a set of disjoint vertices $\omega'$ privacy demands that the neighbourhood $\restri{\omega}$ (blue) and $\restri{}_{\omega'}$ (red) only intersect on their boundaries (hatched vertices). $(b)$ The port decreasing condition demands that, in order to modify a vertex, $A_\omega$ must pay the price of decreasing one of its incoming private ports. Here, $x\in G_{\restri{\omega}}\cap G_{\restri{\omega'}}$ can be modified by both $A_\omega$ and $A_{\omega'}$ as each of them 
has private access to it.
}
\label{fig : Privacy}
\end{figure}

However Sec.~\ref{sec:exV1} discussed how weak consistency can be trivially realised even for non-trivial dynamics, and justified going for full consistency. 
To have sufficient condition for full consistency, we place new hypotheses on the neighbourhood scheme, namely extensivity, monotony and privacy. 

Note that these properties hold for the neighbourhood $\restri{}$ of examples of Secs~\ref{sec:exV1} and~\ref{sec:exV2} defined by:
\begin{align*}
    \forall x\in \mathcal{X},\restr{x}{G} = \V_{G_x} & &\forall \omega\in \mathcal{X}^*,\restr{\omega}{G} = \bigcup_{x\in \omega} \restr{x}{G}
\end{align*}

First we demand that $\restr{\omega}{G}$ be big enough to contain the neighbourhood of each vertex $x\in \omega$, at the time $A_x$ is to be computed. This closure property of $\restri{}$ is called monotony (see Fig.\ref{fig : Privacy}), it ensures that for any valid sequence $\omega$,  $A_\omega$ will not modify beyond $\restr{\omega}{G}$.

\begin{definition}[Monotony]  
    A neighbourhood scheme $\restri{}$ is monotonous iff for every sequence $\omega = \gamma  \beta  \alpha \in \Omega_G(A)$ we have 
    $\restr{\beta}{A_\alpha G} \subseteq \restr{\omega}{G}.$
\end{definition}
In particular we have $\restr{\beta}{A_\alpha G} \subseteq \restr{\beta  \alpha}{G}$ and  $\restr{\alpha}{G} \subseteq \restr{\omega}{G}$.

Second, privacy demands that any disjoint sequences $\omega, \omega'$ have their $\restr{\omega}{G}$ and $\restr{\omega'}{G}$ intersecting only on the vertices of their boundaries and the edges in between them (see Fig.~\ref{fig : Privacy}), a property akin to parallel independence \cite{StrongConfluenceEhrig}. 
Combining privacy and monotony ensures that concurrent influences happen only at these joint boundaries, an essential ingredient of full consistency, as illustrated in Fig.~\ref{fig : Non-private dismiss consistency 1} and~\ref{fig : Non-private dismiss consistency 2}. {They also makes easier to check whether two valid operator $A_x$ and $A_y$ commutes. Indeed when $\restr{y}{G} = \restr{y}{A_x G}$ and $\restr{x}{G} = \restr{x}{A_y G}$, privacy allows us to only check commutation on the boundaries to ensure it globally.}

\begin{definition}[Privacy]  
    A neighbourhood scheme $\restri{}$ is private, iff for any graph $G$ and any disjoint valid sequences $\omega, \omega'\in \Omega_G(A)$, we have 
    $\prec{\omega'}{G}\cap \restr{\omega}{G}=\emptyset.$
\end{definition}

\begin{figure}
\centering
\begin{subfigure}{0.14\textwidth}
    \resizebox{\textwidth}{!}{\usetikzlibrary{patterns}
\tikzset{every picture/.style={line width=0.75pt}} 

\begin{tikzpicture}[x=0.75pt,y=0.75pt,yscale=-1,xscale=1]

\draw  [color={rgb, 255:red, 0; green, 0; blue, 0 }  ,draw opacity=1 ][fill={rgb, 255:red, 255; green, 255; blue, 255 }  ,fill opacity=1 ][line width=3]  (331,240.44) .. controls (331,231.93) and (337.72,225.03) .. (346,225.03) .. controls (354.28,225.03) and (361,231.93) .. (361,240.44) .. controls (361,248.95) and (354.28,255.85) .. (346,255.85) .. controls (337.72,255.85) and (331,248.95) .. (331,240.44) -- cycle ;
\draw  [color={rgb, 255:red, 0; green, 0; blue, 0 }  ,draw opacity=1 ][fill={rgb, 255:red, 255; green, 255; blue, 255 }  ,fill opacity=1 ][line width=3]  (463,242) .. controls (463,233.49) and (469.72,226.59) .. (478,226.59) .. controls (486.28,226.59) and (493,233.49) .. (493,242) .. controls (493,250.51) and (486.28,257.41) .. (478,257.41) .. controls (469.72,257.41) and (463,250.51) .. (463,242) -- cycle ;
\draw [line width=2.25]    (467,229.4) -- (424.37,182.69) ;
\draw [shift={(421,179)}, rotate = 47.61] [fill={rgb, 255:red, 0; green, 0; blue, 0 }  ][line width=0.08]  [draw opacity=0] (14.29,-6.86) -- (0,0) -- (14.29,6.86) -- cycle    ;
\draw  [color={rgb, 255:red, 0; green, 0; blue, 0 }  ,draw opacity=1 ][fill={rgb, 255:red, 255; green, 255; blue, 255 }  ,fill opacity=1 ][line width=3][pattern=north east lines, pattern color=myblack]  (394,166.46) .. controls (394,157.95) and (400.72,151.05) .. (409,151.05) .. controls (417.28,151.05) and (424,157.95) .. (424,166.46) .. controls (424,174.97) and (417.28,181.87) .. (409,181.87) .. controls (400.72,181.87) and (394,174.97) .. (394,166.46) -- cycle ;
\draw  [color={rgb, 255:red, 0; green, 0; blue, 0 }  ,draw opacity=1 ][fill={rgb, 255:red, 255; green, 255; blue, 255 }  ,fill opacity=1 ][line width=3]  (394,70.49) .. controls (394,61.98) and (400.72,55.08) .. (409,55.08) .. controls (417.28,55.08) and (424,61.98) .. (424,70.49) .. controls (424,79) and (417.28,85.9) .. (409,85.9) .. controls (400.72,85.9) and (394,79) .. (394,70.49) -- cycle ;
\draw [line width=2.25]    (354,227.07) -- (395.57,183.04) ;
\draw [shift={(399,179.4)}, rotate = 133.35] [fill={rgb, 255:red, 0; green, 0; blue, 0 }  ][line width=0.08]  [draw opacity=0] (14.29,-6.86) -- (0,0) -- (14.29,6.86) -- cycle    ;
\draw [line width=2.25]    (409,151.05) -- (409,90.9) ;
\draw [shift={(409,85.9)}, rotate = 90] [fill={rgb, 255:red, 0; green, 0; blue, 0 }  ][line width=0.08]  [draw opacity=0] (14.29,-6.86) -- (0,0) -- (14.29,6.86) -- cycle    ;
\draw [color={rgb, 255:red, 208; green, 2; blue, 27 }  ,draw opacity=1 ][fill={rgb, 255:red, 208; green, 2; blue, 27 }  ,fill opacity=0.35 ][line width=3] [line join = round][line cap = round]   (495,268.4) .. controls (482.67,268.4) and (470.32,269.03) .. (458,268.4) .. controls (454.65,268.23) and (441.43,259.77) .. (438,258.4) .. controls (404.87,245.15) and (403.37,229.45) .. (384,200.4) .. controls (381.98,197.37) and (379.3,194.93) .. (377,193.4) .. controls (375.11,192.14) and (376.6,189.96) .. (376,188.4) .. controls (374.66,184.92) and (372.88,181.62) .. (371,178.4) .. controls (366.58,170.82) and (362.04,164.53) .. (359,155.4) .. controls (345.35,114.45) and (340.07,44.24) .. (397,42.4) .. controls (407.67,42.06) and (420.12,37.48) .. (429,43.4) .. controls (443.14,52.83) and (463.33,81.39) .. (469,98.4) .. controls (470.56,103.07) and (471.74,108.89) .. (474,113.4) .. controls (476.41,118.22) and (477.81,123.48) .. (480,128.4) .. controls (480.87,130.35) and (483.41,131.35) .. (484,133.4) .. controls (486.81,143.22) and (490.15,156.78) .. (494,166.4) .. controls (499.24,179.49) and (504.54,192.03) .. (509,205.4) .. controls (510.1,208.69) and (513.16,212.06) .. (514,215.4) .. controls (518.77,234.48) and (517.1,268.4) .. (496,268.4) ;
\draw [color={rgb, 255:red, 11; green, 11; blue, 108 }  ,draw opacity=1 ][fill={rgb, 255:red, 74; green, 144; blue, 226 }  ,fill opacity=0.35 ][line width=3] [line join = round][line cap = round]   (337,268.07) .. controls (346,268.07) and (355.02,268.69) .. (364,268.07) .. controls (372.09,267.51) and (392.85,247.88) .. (400,242.07) .. controls (410.77,233.31) and (417.69,223.37) .. (427,214.07) .. controls (429.24,211.83) and (429.74,208.33) .. (432,206.07) .. controls (452.38,185.69) and (473.63,161.36) .. (485,134.07) .. controls (496.97,105.34) and (481.89,77.95) .. (464,60.07) .. controls (461.77,57.84) and (460.08,54.61) .. (457,53.07) .. controls (440.71,44.92) and (427.6,44.34) .. (408,45.07) .. controls (404.18,45.21) and (399.5,48.33) .. (396,49.07) .. controls (383.17,51.77) and (377.17,52.48) .. (366,58.07) .. controls (361.11,60.51) and (356.88,65.74) .. (354,70.07) .. controls (353.04,71.51) and (345.52,77.94) .. (344,80.07) .. controls (328.23,102.15) and (314.31,134.39) .. (314,162.07) .. controls (313.75,184.33) and (292.88,268.07) .. (337,268.07) -- cycle ;

\draw (342,230) node [anchor=north west][inner sep=0.75pt]   [align=left] {\textbf{{\sizeFdixas $i$}}};
\draw (470,230) node [anchor=north west][inner sep=0.75pt]   [align=left] {\textbf{{\sizeFdixas $j$}}};
\draw (371,219) node [anchor=north west][inner sep=0.75pt]   [align=left] {{{\sizeFdixa $u$}}};
\draw (430,219) node [anchor=north west][inner sep=0.75pt]   [align=left] {{{\sizeFdixa $v$}}};
\draw (358,67) node [anchor=north west][inner sep=0.75pt]   [align=left] {{{\sizeFdixa $w$}}};
\draw (415,95.4) node [anchor=north west][inner sep=0.75pt]  [font=\sizeFdixap,color={rgb, 255:red, 0; green, 0; blue, 0 }  ,opacity=1 ]  {$:b$};

\end{tikzpicture}}
    \caption{$G$}
    \label{fig : Non-private dismiss consistency 1}
\end{subfigure}
\hfill
\begin{subfigure}{0.14\textwidth}
    \resizebox{0.76\textwidth}{!}{\tikzset{every picture/.style={line width=0.75pt}} 

\begin{tikzpicture}[x=0.75pt,y=0.75pt,yscale=-1,xscale=1]

\draw  [color={rgb, 255:red, 0; green, 0; blue, 0 }  ,draw opacity=1 ][fill={rgb, 255:red, 255; green, 255; blue, 255 }  ,fill opacity=1 ][line width=3]  (331,240.44) .. controls (331,231.93) and (337.72,225.03) .. (346,225.03) .. controls (354.28,225.03) and (361,231.93) .. (361,240.44) .. controls (361,248.95) and (354.28,255.85) .. (346,255.85) .. controls (337.72,255.85) and (331,248.95) .. (331,240.44) -- cycle ;
\draw  [color={rgb, 255:red, 0; green, 0; blue, 0 }  ,draw opacity=1 ][fill={rgb, 255:red, 255; green, 255; blue, 255 }  ,fill opacity=1 ][line width=3]  (463,242) .. controls (463,233.49) and (469.72,226.59) .. (478,226.59) .. controls (486.28,226.59) and (493,233.49) .. (493,242) .. controls (493,250.51) and (486.28,257.41) .. (478,257.41) .. controls (469.72,257.41) and (463,250.51) .. (463,242) -- cycle ;
\draw [line width=2.25]    (467,229.4) -- (424.37,182.69) ;
\draw [shift={(421,179)}, rotate = 47.61] [fill={rgb, 255:red, 0; green, 0; blue, 0 }  ][line width=0.08]  [draw opacity=0] (14.29,-6.86) -- (0,0) -- (14.29,6.86) -- cycle    ;
\draw  [color={rgb, 255:red, 0; green, 0; blue, 0 }  ,draw opacity=1 ][fill={rgb, 255:red, 255; green, 255; blue, 255 }  ,fill opacity=1 ][line width=3]  (394,166.46) .. controls (394,157.95) and (400.72,151.05) .. (409,151.05) .. controls (417.28,151.05) and (424,157.95) .. (424,166.46) .. controls (424,174.97) and (417.28,181.87) .. (409,181.87) .. controls (400.72,181.87) and (394,174.97) .. (394,166.46) -- cycle ;
\draw[color = myblack][line width=3] (410,68) circle (20);
\draw [line width=2.25]    (354,227.07) -- (395.57,183.04) ;
\draw [shift={(399,179.4)}, rotate = 133.35] [fill={rgb, 255:red, 0; green, 0; blue, 0 }  ][line width=0.08]  [draw opacity=0] (14.29,-6.86) -- (0,0) -- (14.29,6.86) -- cycle    ;
\draw [line width=2.25]    (409,151.05) -- (409,90.9) ;
\draw [shift={(409,85.9)}, rotate = 90] [fill={rgb, 255:red, 0; green, 0; blue, 0 }  ][line width=0.08]  [draw opacity=0] (14.29,-6.86) -- (0,0) -- (14.29,6.86) -- cycle    ;

\draw (398.5,58) node [anchor=north west][inner sep=0.75pt]   [align=left] {\textbf{{\sizeFdixas \textcolor[rgb]{0.11,0.35,0.65}{i}/\textcolor[rgb]{0.82,0.01,0.11}{j}}}};
\draw (371,219) node [anchor=north west][inner sep=0.75pt]   [align=left] {{{\sizeFdixa $u$}}};
\draw (430,219) node [anchor=north west][inner sep=0.75pt]   [align=left] {{{\sizeFdixa $v$}}};
\draw (353,67) node [anchor=north west][inner sep=0.75pt]   [align=left] {{{\sizeFdixa $w$}}};
\draw (415,95.4) node [anchor=north west][inner sep=0.75pt]  [font=\sizeFdixap,color={rgb, 255:red, 0; green, 0; blue, 0 }  ,opacity=1 ]  {$:a$};

\end{tikzpicture}}
    \caption{$A_u G$/$A_v G$}
    \label{fig : Non-private dismiss consistency 2}
\end{subfigure}
\hfill
\begin{subfigure}{0.10\textwidth}
    \resizebox{0.6\textwidth}{!}{\begin{tikzpicture}


\emptycell{0}{0}{myblack}{u};
\emptycell{0}{5}{myblack}{v};


\edge{u}{v}{myblack};


\draw (1,1) node[right]{{\sizeFdixc $u$}};
\draw (1,4) node[right]{{\sizeFdixc $v$}};

\draw (-0.85,3.5) node[color = border] {{\sizeFdixcp $:b$}};
\draw (-0.85,1.5) node[color = border] {{\sizeFdixcp $:a$}};

\end{tikzpicture}}
    \caption{$G$}
    \label{fig : Non-PD dismiss consistency 1}
\end{subfigure}
\hfill
\begin{subfigure}{0.10\textwidth}
    \resizebox{0.6\textwidth}{!}{\begin{tikzpicture}


\emptycell{0}{0}{myblack}{u};
\rightleft{0}{5}{myblack}{v};


\edge{u}{v}{myblack};


\draw (1,1) node[right]{{\sizeFdixc $u$}};
\draw (1,4) node[right]{{\sizeFdixc $v$}};

\draw (-0.85,3.5) node[color = border] {{\sizeFdixcp $:b$}};
\draw (-0.85,1.5) node[color = border] {{\sizeFdixcp $:a$}};

\end{tikzpicture}}
    \caption{$A_u G$}
    \label{fig : Non-PD dismiss consistency 2}
\end{subfigure}
\caption{\emph{Inconsistent dynamics examples.} The example of $(a)\&(b)$ relies on a neighbourhood scheme which is not private as $w\in \prec{u}{G}\cap \prec{v}{G}$. It follows that $A_u G$ and $A_v G$ disagree on the internal states of $w$, whilst both updating its incoming port in the same fashion ($b$ becomes $a<b$). 
The example of $(c)\&(d)$ is not port decreasing. It follows that $G$ and $A_u G$ disagree on $v$ but its set of incoming ports stay the same.}
\label{fig : Inconsistent dynamics}
\end{figure}

Lastly we want to forbid that a local rule modifies a vertex without altering its incoming ports, as this would immediately infringe full consistency (see Fig.~\ref{fig : Non-PD dismiss consistency 1} and~\ref{fig : Non-PD dismiss consistency 2}). Moreover the modification needs be decreasing---otherwise we could apply $A_x$ twice in a row and get to the exact same counterexample. This idea that $A_x$ should decrease the incoming ports of the vertex $v$ it modifies, can be understood as a natural way of `locally reflecting the progress of the computation of $v$', i.e. geometrically accounting for the fact that the dependency between $x$ and $v$ has been reduced, and hence their space-time relationship has changed. In Sec.~\ref{sec:exV1} the local rule is port decreasing because we suppress an incoming edge each time we modify a vertex. In Sec.~\ref{sec:exV2} the local rule is port decreasing for more subtle reasons in the case of a red vertex: we decrease $b'$ into $b''$ (see Fig.~\ref{fig : Dynamic example 2; G} and~\ref{fig : Dynamic example 2; A_x G}).

\begin{definition}[Port decreasing]\label{definition : Port Decreasing}
  An $\restri{}$-local rule $A_{(-)}$ is \emph{port decreasing} iff for any graph $G$, position $x\in \X(\Past(G))$ and vertex $u\in V_G \cap V_{A_x G}$ whenever $G_u^0\neq (A_xG)_u^0$ we have
    $$\pof(\E_G(V_G,u)\setminus\E_G(\cpreci{\omega},u)) > \pof(\E_{A_x G}(V_{A_x G},u)\setminus\E_G(\cpreci{\omega},u)) $$
    for a order $\leq$ over sets of ports which is for any $A,B,A',B'\in \mathcal{P}(\pi)$:
    \begin{itemize}
        \item \textbf{inclusive} i.e. $A \supseteq A' \Longrightarrow A\geq A'$
        \item \textbf{monotonous} i.e. $A\cap B = \emptyset\ \land\ A\geq A'\ \land\ B>B' \Longrightarrow A \cup B > A' \uplus B'$
    \end{itemize}
    where $\E_G(Y,u)$ denotes the set of edges in $G$ from any $v\in \zee Y$ to $u$, and $\pof E$ is as in Def.~\ref{def : consistency}.
\end{definition}

Note that we ask a port decreasing operator to decrease the port of a \emph{private edge}, i.e. an edge starting from $\restri{}^-$. Without this assumption we can give a counterexample by considering the graph of Fig.~\ref{fig : Privacy}. Say both $A_{\omega}$ and $A_{\omega'}$ were to modify the internal state of $v$ and decrease the port associated to the shared edge coming from $u$. Then we could have $\pof\Em_{A_\omega G}(v) = \pof\Em_{A_{\omega'} G}(v)$ whilst $\sigma_{A_{\omega'}G}(v) \neq \sigma_{A_\omega G}(v)$.
Finally we state our main result. 

\begin{restatable}{theorem}{generalcase}{\bf (Obtaining consistency)}\label{Th : General Case Consistency}
    Let $A_{(-)}$ be a port-decreasing time-increasing commutative $\restri{}$-local rule, with $\restri{}$ an extensive monotonous and private neighbourhood scheme.
    For all graphs $G$, $\mathcal{M}_A(G)$ is consistent.
\end{restatable}
\withproofs{
\begin{proof}
    See Sec.\ref{proof : General Case Consistency annex}.
\end{proof}}

The proof is quite intricated, but we can give the following intuition. On the one hand the commutation hypothesis ensures that $A_\omega G\varparallel A_{\omega'} G$ as long as $\omega'$ is a permutation of $\omega$. On the other hand as long as $\restri{}$ is extensive, monotonous and private two disjoint operators necessarily modify different subgraphs, albeit with a common boundary. Then, by decreasing the private incoming port of each modified vertex, we obtain full consistency by dismissing its premise.

\section{Conclusion} \label{sec:conclusion}

{\em Summary of results.} We introduced graphs that must be thought of as space-like cuts of space-time diagrams. The vertices have names of the form $u=t.x$, they must be thought of as events. The edges must be thought of as dependencies between events, i.e. if $u=t.x$ points to $v=t'.y$, then $v$ is ahead in time of $u$, awaiting for information from $u$. The action of a local rule $A_x$ on $u$ is non-trivial only if $u$ is minimal: it disposes of it by communicating its information to $v$ and other dependencies, and creates vertex $u'=(t+1).x$ in some provisional state.\\ 
We argued that the right notion of space-time determinism is the consistency of its space-time diagram: the state of each event (in terms of its internal state and connectivity) needs be a function of its set of incoming ports, as these represent the angle at which the space-like cut traverses the event. We gave sufficient conditions for the asynchronous applications of a local rule to be consistent: they must be commuting and port-decreasing, with respect to an extensive, monotonous, private notion of neighbourhood. We also looked at a weaker notion of consistency, requiring that only the normal forms of events across space-time be well-determined: commutation alone suffices then.\\ 
Throughout, we argued of the potential implications for distributed computation (weak consistency), asynchronous simulation of dynamical systems (full consistency and our first example), and discrete toy models of general relativity (our second example).  

\noindent {\em Related works.} Geometry is dynamical in our work. We thus hope it makes useful addition to the already wide literature on Graph Rewriting \cite{RozenbergBook,EhrigBook}. 
We are aware that the dominating vocabulary to describe them is now that of Category theory \cite{LoweAlgebraic,Taentzer,HarmerFundamentals}, in which ways of combining non-commuting rules \cite{EchahedCombiningRules} and notions of space-time diagrams have been developed \cite{UnfoldingAdhesivePawel,BehrStochastic,UnfoldingKonig,UnfoldingReckel}.\\ 
We instead use the vocabulary of dynamical systems, as we came to consider Graph Rewriting though a series of works generalising cellular automata to synchronous, causal graph dynamics \cite{ArrighiCayleyNesme}, and tilings to graph subshifts \cite{ArrighiSubshifts}. We are confident that abstracting away the essential features of our formalism could yield interesting categorical frameworks, e.g. à la \cite{Maignan}.\\
The closest works however turn out to come from varied communities. In algorithmic complexity \cite{AsynchronousSimsSync} uses a DAG of dependencies representation to reduce the synchronisation costs of simulating a class of synchronous algorithms---we use it in the more dynamical systems context and in order to relax synchronism altogether, whilst safeguarding space-time determinism. In computational Physics \cite{GorardWofram} promotes the lattice of dependencies of local rule applications to a notion space-time, and advocates a notion of `causal invariance' based on the unicity of this lattice, as formalised in the context of string rewriting \cite{CausalInvarianceOnStrings}---we formalise space-determinism for graph rewriting mathematically, without reference to this lattice of dependencies of applications, and provide local conditions to achieve it. In the network reliability community, \cite{WeakConsistencyGacs} obtains a result similar to Prop.~\ref{prop:weakconsistency}---our local rules are allowed to modify the neighbouring vertices and the graph per se, and we reach full consistency.

\noindent {\em Perspectives.} Clock synchronisation is an expensive overhead for parallel simulation of dynamical systems, as well as numerous distributed computation applications. The hereby developed theoretical framework says we can just do away with them and still obtain a strong form space-time determinism, provided that the local rule meets certain requirements. We are looking forward to see this being applied in practice. We, on the other hand, are likely to focus on the reversible and quantum regimes of these graph rewriting models. Another, important open problem is to understand whether space-times of fully-consistent local rules correspond to expansive graph subshifts \cite{ArrighiSubshifts}, in the same way that space-times of cellular automata correspond to expansive tilings. One of the difficulties in establishing such a result is that the naming conventions we adopt for our vertices systematically prevent our space-times from being cyclic, even when the dynamics described is periodic---whereas the corresponding graph subshift will in fact be cyclic.

\noindent {\bf Acknowledgements}~ We wish to thank Nicolas Behr for helpful discussions. This publication was made possible through the support of the PEPR integrated project EPiQ ANR-22-PETQ-0007 part of Plan France 2030, and the ID \#62312 grant from the John Templeton Foundation, as part of the \href{https://www.templeton.org/grant/the-quantum-information-structure-of-spacetime-qiss-second-phase}{‘The Quantum Information Structure of Spacetime’ Project (QISS) }. The opinions expressed in this project/publication are those of the author(s) and do not necessarily reflect the views of the John Templeton Foundation.

\bibliographystyle{plainurl}
\bibliography{ICGT}

\newpage

\appendix

\section{Valid sequences and commutative local-rule}

We notice that a past vertex must remains so, hence the well-definiteness of commutativity. 
\begin{lemma}[Validity of commutativity]\label{lem:2valid}\label{cor:commutativity}
    Let $A$ be a (not necessarily commutative) local rule, $G$ be a graph and $x, y \in \X(\Past(G))$.
    The sequence $yx$ is valid in $G$. Moreover, if A is \emph{commutative}, $A_{xy}G=A_{yx}G$.
\end{lemma}
\begin{proof}
    As there exists no path to $y$ in $G$, the reachability condition of neighbourhood schemes implies that $y \notin \restr{x}{G}$.
    By $\restri{}$-locality $A_x$ cannot add any incoming edge to $y$ which implies $y\in \X(\Past(A_{x} G))$, i.e. $yx$ is valid as wanted.  
\end{proof}

\begin{lemma}[*-validity of commutativity]\label{lemma : right commuting of a past vertex}
    Let $A$ be a commutative local rule.
    Let $\alpha$ be a valid sequence in $G$.
    Let $x\in \X(\Past(G))$ such that $x\notin \alpha$.
    We have $x\alpha$ and $\alpha x$ are valid in $G$ and $A_{x\alpha}G=A_{\alpha x}G$.
\end{lemma}
\begin{proof}
    We proceed by induction on the size of $\alpha$.
    When $|\alpha| = 0$ it is immediate.
    Let us suppose $\alpha = y\beta$.
    By induction hypothesis, we known that (1) $x\beta$ is valid in $G$, (2) $\beta x$ is valid in $G$ and (3) $A_{x\beta}G=A_{\beta x}G$.
    By (1), we have $x\in \X(\Past(A_\beta G))$.
    Since $\alpha = y\beta$ is supposed to be valid, we also have $y \in \X(\Past(A_{\beta}G))$.
    But Lem. \ref{cor:commutativity} tells us that $xy$ and $yx$ are both valid in $A_\beta G$ and $A_{xy}(A_\beta G)=A_{yx}(A_\beta G)$.
    Thus $xy\beta$ and $yx\beta$ are valid in $G$ and $A_{xy\beta} G= A_{yx\beta} G = A_{y}(A_{x\beta} G)$.
    Combining this with (2) and (3) finishes the induction and the proof.
\end{proof}
The two-letters commutativity condition is equivalent the ability to permute all letters of a valid sequence.
\begin{lemma}[$*$-commutation]\label{lemma : *-commutation}
    Let $A_{(-)}$ be a commutative local rule. Let $\omega$ be a valid sequence in $G$ and $\omega'$ be a permutation of $\omega$ also valid in $G$. We have 
    $$A_\omega G = A_{\omega'} G$$
\end{lemma}

\begin{proof}
    Let us denote $\omega = x_{n} \dots x_2 x_1$.
    We prove by induction that for all $i \in \{0,\dots,n\}$,
    $$A_{\omega'} G = A_{\beta}A_{x_i\dots x_1}G,$$
    with $\beta$ a permutation of $x_{n} \dots x_{i+1}$.
    The case $i=0$ is verified with $\beta = \omega'$.
    When $i= j+1$ we have $A_{\omega'} G = A_{\beta}A_{x_j\dots x_1}G$.
    By validity of $\omega$, we know that $x_{j+1} \in \X(\Past(A_{x_j\dots x_1} G))$.
    Consider the decomposition $\beta = \beta''x_{j+1}\beta'$ with $x_{j+1} \notin \beta'$. Lem. \ref{lemma : right commuting of a past vertex} applied to $\beta'$ and $x_{j+1}$ in the graph $A_{x_j\dots x_1} G$ gives us $A_{\omega'} G = A_{\beta''\beta'}A_{x_{j+1}\dots x_1}G$ as needed to finish the induction.
\end{proof}

\section{Weak consistency}

For the next proof we will need a notation for the rightmost sequence subtraction. Let $\omega\in \mathcal{X}^*$ and $\alpha \in \mathcal{X}^*$. We define recursively $(\omega \setminus \alpha)\in \mathcal{X}^*$ as:
\begin{equation*}
    \omega \setminus \alpha =
    \begin{cases}
        \omega & \text{if } |\alpha| = 0, \\
        \omega''\omega' & \text{if } |\alpha| = 1, \omega = \omega''\alpha\omega', \text{ and } \alpha \notin \omega'\\
        (\omega \setminus x) \setminus \alpha' & \text{if } \alpha = \alpha' x \text{ and } x \in \mathcal{X}.
    \end{cases}
\end{equation*}
For example if $\mathcal{X} = \{0,\dots ,9\}$, $\omega = 22159892$ and $\omega' = 28542$ we have 
$\omega\setminus \omega' = 2199$.
This operation preserves validity. 
\begin{lemma}[Validity of sequence substraction]\label{lemma : validity of omega - omega'}
    Let $A_{(-)}$ be a commutative local rule. Let $\omega$ and $\omega'$ be valid sequences. Then the sequence $\omega \setminus \omega'$ is valid in $A_{\omega'}G$. 
\end{lemma}

\begin{proof}
    We proceed by induction on the size of $\omega'$.
    It is immediate when $\omega'$ is empty.
    When $\omega' = x \alpha$, the validity of $\omega'$ gives us that $x\in \X(\Past(A_\alpha G))$.
    The induction hypothesis is that $\omega \setminus \alpha$ is valid in $A_{\alpha}G$.
    Taking $\omega \setminus \alpha = \gamma\beta$ with $\beta$ the longest suffix of $\omega \setminus \alpha$ such that $x \notin \beta$, we have $\beta$ valid in $A_{\alpha}G$.
    
    We can use Lem. \ref{lemma : right commuting of a past vertex} on $\beta$ and $x$ in $A_\alpha G$ to get validity of $\beta$ in $A_{x \alpha}G$ and $A_{x \beta \alpha} G = A_{\beta x \alpha}G$.
    If $\gamma = \varepsilon$ we can conclude immediately.
    Otherwise we have $\gamma = \beta' x$.
    Since $\beta'$ is valid in $A_{x \beta \alpha} G $ it is also valid in $A_{ \beta x\alpha} G $. This finishes to prove validity of $\beta' \beta = \omega\setminus x\alpha$ in $A_{x \alpha} G$.
\end{proof}
Whilst confluence is not the primary aim of this paper, we obtain it as a corollary. 
\begin{corollary}[Confluence]\label{lemma : confluence}
    Let $A_{(-)}$ be a commutative local rule. Let $\omega$ and $\omega'$ be valid sequences. Then the sequences $\omega'\setminus \omega$ and $\omega \setminus \omega'$ are valid respectively in $A_\omega G$ and $A_{\omega'}G$. This enforces:
    $$A_{(\omega'\setminus \omega)}A_\omega G = A_{(\omega\setminus \omega')}A_{\omega'} G$$
\end{corollary}
\begin{proof}
    We get each validity by one application of Lem. \ref{lemma : validity of omega - omega'}. Then we get the equality by applying Lem. \ref{lemma : *-commutation}.
\end{proof}

In any space-time diagram ${\cal M}_A(G)$ for $A$ time-increasing commutative, the state of a past vertex is well-determined.  

\begin{lemma}[Common past vertices have been equally updated]\label{lemma : common past vertices}
    Let $A_{(-)}$ be a time-increasing commutative local rule. Let $\omega$ and $\omega'$ be valid sequences in $G$. If there exists $t.x \in \Past(A_\omega G) \cap \Past(A_{\omega'}G)$ then $\omega'$ contains as much $x$ as $\omega$. 
\end{lemma}

\begin{proof}
    We suppose without loss of generality that there is at least as much $x$ in $\omega'$ than in $\omega$. Then because $x\notin (\omega \setminus \omega')$ and $t.x\in \Past(A_{\omega'}G)$ we have $t.x\in \Past(A_{(\omega\setminus\omega')\omega'}G)$. Using Lem.~\ref{lemma : confluence} this implies $t.x\in \Past(A_{(\omega'\setminus\omega)\omega}G)$. Since $t.x\in \Past(A_\omega G)$ this implies by the time-increasing condition that $\omega' \setminus \omega$ does not contain any $x$.
\end{proof}

\weakconsistency*\label{proof : weak consistency}
\begin{proof}
    Let $\omega$ and $\omega'$ be valid sequences in $G$. Let $t.x \in \Past(A_\omega G)\cap \Past(A_{\omega'}G)$.
    Lem.~\ref{lemma : common past vertices} tells us that there is as much $x$ in $\omega'$ than in $\omega$.
    This means $x\notin (\omega \setminus \omega')$ and $x\notin (\omega' \setminus \omega)$.
    Then, using locality and Lem.~\ref{lemma : confluence}, we get the following equalities
     $$(A_{\omega'}G)_{t.x} = (A_{(\omega \setminus \omega') \omega'}G)_{t.x}=(A_{(\omega'\setminus \omega) \omega}G)_{t.x}= (A_{\omega}G)_{t.x}.$$
\end{proof}

Moreover, in any space-time diagram ${\cal M}_A(G)$ for $A$ time-increasing commutative, fixing set of past vertices fixes the entire space-like cut.

\unicut*\label{proof : unicut}

\begin{proof}
    We consider two graphs $A_\omega G$ and $A_{\omega'}G$ such that $\Past(A_\omega G) = \Past(A_{\omega'}G)$. Let us prove by contradiction that $\omega \setminus \omega'$ is empty. We suppose it non empty and call $x$ its right-most letter. Then validity (coming from Lem.~\ref{lemma : confluence}) gives us $t.x\in \Past(A_{\omega'}G) = \Past(A_\omega G)$.
    By Lem.~\ref{lemma : common past vertices}, we have that $\omega$ and $\omega'$ contains as much $x$ as the other.
    So $\omega \setminus \omega'$ does not contain $x$, a contradiction.

    Thus we have $\omega \setminus \omega' = \emptyset$. Symmetrically $\omega'\setminus \omega = \emptyset$.
    Using Lem.~\ref{lemma : *-commutation} this proves $A_\omega G = A_{\omega'} G$.
\end{proof}

\section{Full consistency}

 The two hypotheses of Th.~\ref{Th : General Case Consistency} ($\restri{}$-locality and the port decreasing condition) handle only size $1$ valid sequences. In order to prove the theorem we first have to establish that the two hypotheses for any valid sequence. We start by extending the notion locality, to account not just for a single position, but valid sequences of them.
\begin{definition}[$*$-$\restri{}$-Locality]
    Consider a neighbourhood scheme $\restri{}$. An operator $A_{(-)}$ is said to be $\omega$-$\restri{}$-local iff, for all $G$ for which $\omega\in \Omega_G(A)$ we have
    $$A_\omega G = A_\omega (G_\restri{\omega}) \sqcup G_\crestri{\omega}.$$
    An operator said to be $\restri{}$-local iff this holds for any $\omega\in \Omega_G(A)$ such that $|\omega|=1$, we then say it is a local rule. It is said to be $*$-$\restri{}$-local if this holds for any $\omega\in \Omega_G(A)$. 
\end{definition}
Let us prove that any $\restri{}$-local operator for an extensive and monotonous neighbourhood scheme is also $*$-$\restri{}$-local.
\begin{lemma}[$\restri{}\subseteq X$-locality implies $X$-locality]\label{lemme : petit local implique grand local version générale}
    Consider an extensive neighbourhood scheme $\restri{}$, a graph $G$ and $\omega \in \mathcal{X}^*$, and any set $X\subseteq \mathcal{X}$ such that $\restr{\omega}{G}\subseteq X$. If $A_{(-)}$ is $\omega$-$\restri{}$-local then we have:
    $$A_\omega G = (A_\omega G_{X})\sqcup G_{\overline{X}}$$
\end{lemma}
\begin{proof}
In order to lighten the notations of this proof we temporarily write $\restri{}$ instead of $\restri{\omega}$.
By extensivity of $\restri{}$, $G_{\restr{}{G}}\sqsubseteq G_{X} \sqsubseteq G$ implies $\restr{}{G}=\restr{}{G_{X}}$.
\begin{align*}
    A_\omega G_{X}&= A_\omega G_{X{}\restr{}{G_{X}}}\sqcup G_{X{}\crestr{}{G_{X}}}\\
    &= A_\omega G_{X{}\restr{}{G}}\sqcup G_{X\crestr{}{G}}\textrm{ by extensivity}\\
    &= A_\omega G_{\restr{}{G}}\sqcup G_{X\cap\crestr{}{G}}\textrm{ since }\restr{}{G}\subseteq X\\
    A_\omega G_{X}\sqcup G_{\overline{X}}&=A_\omega G_{\restr{}{G}}\sqcup G_{X\cap\crestr{}{G}} \sqcup G_{\overline{X}}\\
    &=A_\omega G_{\restr{}{G}}\sqcup G_{X\cap\crestr{}{G}} \sqcup G_{\overline{X}\cap\crestr{}{G}} \textrm{ since } \overline{X}\subseteq\crestr{}{G}\\
    &=A_\omega G_{\restri{}(G)}\sqcup G_{\crestri{}(G)}\\
    &= A_\omega G
\end{align*}
\end{proof}

\begin{lemma}[$\restri{}$-locality implies $*$-$\restri{}$-locality]\label{lem:extloc}
    If an operator $A_{(-)}$ is $\restri{}$-local for a monotonous and extensive neighbourhood scheme $\restri{}$, then it is also $*$-$\restri{}$-local.
\end{lemma}
\begin{proof}
    By recurrence on the size of $\omega$. The base case $\omega=u$ is $\restri{}$-locality. Let $\omega = \beta \alpha$ with $\beta,\alpha$ non-empty. By monotony we have that for all $G$, $\restr{\alpha}{G}\subseteq \restr{\omega}{G}$. We can thus use Lem.~\ref{lemme : petit local implique grand local version générale} to get:  
\begin{align*}
A_\alpha G &= (A_\alpha G_{\restri{\omega}})\sqcup G_{\crestri{\omega}} \\
A_\beta A_\alpha G &= A_\beta ((A_\alpha G_{\restri{\omega}})\sqcup G_{\crestri{\omega}})
\end{align*}
By monotony we also have that $\restr{\beta}{A_\alpha G} \subseteq \restr{\omega}{G}$. We can thus use Lem.~\ref{lemme : petit local implique grand local version générale} again with $X = \restr{\omega}{G}$: 
\begin{align*}
A_\beta A_\alpha G &= \left(A_\beta \left(((A_\alpha G_{\restri{\omega}}) \sqcup G_{\crestri{\omega}})_{ \restr{\omega}{G}}\right)\right) \sqcup ((A_\alpha G_{\restri{\omega}})\sqcup G_{\crestri{\omega}})_{\overline{ \restri{\omega}}(G)} \\
&= \left(A_\beta \left(((A_\alpha G_{\restri{\omega}}) \phantom{ \sqcup G_{\crestri{\omega}}})_{ \restr{\omega}{G}}\right)\right) \sqcup (( \phantom{A_\alpha G_{\restri{\omega}}) \sqcup }G_{\crestri{\omega}})_{\overline{ \restri{\omega}}(G)}\textrm{ by 
Rk~\ref{rk:locality excludes borders}} \\
&=(A_\beta A_\alpha G_{\restri{\omega}}) \sqcup  G_{\crestri{\omega}}
\end{align*}
\end{proof}

In a similar manner we can define $*$-port decreasing operators, by considering any valid sequence $\omega$, i.e. applying $A_\omega$ instead of $A_u$ in definition~\ref{definition : Port Decreasing}. Let us show that any port decreasing operator for a monotonous neighbourhood scheme is also $*$-port decreasing.

\begin{remark}\label{Remark : Useful decomposition}

    The LHS of the port-decreasing condition (Def.~\ref{definition : Port Decreasing}) could have simply been written $\pof\E_G(\preci{\omega},x)$. The RHS however divides up into remainder originally private edges and new edges, i.e.
    $\pof\big((\E_{A_\omega G}(x)\cap\E_G(\preci{\omega},x))\, \uplus\, (\E_{A_\omega G}(x)\setminus\E_G(x))\big)$.\\
    Thus a way of understanding the port decreasing condition is the following: the subset of $\P_{A_u G}(x)$ containing  the remainder originally private ports and the new ports must be smaller than the originally private ports.
\end{remark}

\begin{restatable}{lemma}{omegaPD}{\bf (Port decreasing implies $*$-port decreasing)}\label{lemma : omega port decreasing} If $A_{(-)}$ is port decreasing for a monotonous neighbourhood scheme $\restri{}$, then $A_{(-)}$ is $*$-port decreasing for $\restri{}$.
\end{restatable}

\begin{proof}
    We show this by complete induction on the length of $\omega$.

    If $\omega$ contains only one letter it comes directly from Def.~\ref{definition : Port Decreasing}.

    Otherwise we write $\omega = \beta \alpha$ with $\beta$ and $\alpha$ non empty. We want to show that $A_\omega$ is port decreasing--- i.e. $$\Px{G_\preci{\omega}} > (\Px{A_\omega G}\cap \Px{G_\preci{\omega}})\cup (\Px{A_\omega G} 
 \setminus \Px{G}). $$
by supposing $A_\alpha$ and $A_\beta$ port-decreasing.

    We proceed in two steps, we will prove these inequalities:
    $$\pof(\E_G(x)\setminus\E_G(\cpreci{\omega},x)) > \pof(\E_{A_\alpha G}(x)\setminus\E_G(\cpreci{\omega},x)) >  \pof(\E_{A_\omega G}(x)\setminus\E_G(\cpreci{\omega},x)) $$
    
    We start by the left inequality. Using Rk~\ref{Remark : Useful decomposition}, the LHS is equal to $\E_G(\preci{\omega},x)$. Using $\restri{\alpha}\subseteq\restri{\omega}$ as stated by monotony, we decompose it according to privacy along $\alpha$.
    $$\E_G(\preci{\omega},x) =\E_G(G_\preci{\alpha},x) \uplus {(\E_G(G_\preci{\omega})\setminus\E_G(G_\preci{\alpha},x))} $$

    Using Rk~\ref{Remark : Useful decomposition}, the RHS is equal to

    $$ (\E_{A_\alpha G}(x)\cap\E_G(\preci{\omega},x)) \uplus (\E_{A_\alpha G}(x) \setminus\E_G(x))   $$
    
    We also decompose the RHS according to privacy along $\alpha$, using the same formula, and obtain:

 \begin{align*}
 &\big(\E_{A_\alpha G}(x)\cap  (\E_G(\preci{\alpha},x) \uplus {(\E_G(\preci{\omega})\setminus\E_G(\preci{\alpha},x))})   \big) \uplus (\E_{A_\alpha G}(x) \setminus\E_G(x))\\
 &= (\E_{A_\alpha G}(x)\cap\E_G(\preci{\alpha},x)) \uplus \big(\E_{A_\alpha G}(x)\cap{(\E_G(\preci{\omega})\setminus\E_G(\preci{\alpha},x))}  \big) \uplus (\E_{A_\alpha G}(x) \setminus\E_G(x))\\
\end{align*}

The following inequality holds because one set is included in the other: 
$$\pof\left({(\E_G(\preci{\omega})\setminus\E_G(\preci{\alpha},x))}\right) \geq \pof \left(\E_{A_\alpha G}(x)\cap{(\E_G(\preci{\omega})\setminus\E_G(\preci{\alpha},x))}\right)$$
The remaining part of the LHS is also greater than the remaining part of the RHS
$$
\pof\left(\E_G(\preci{\alpha},x)\right)  > \pof\left(
(\E_{A_\alpha G}(x)\cap\E_G(\preci{\alpha},x)) \uplus (\E_{A_\alpha G}(x) \setminus\E_G(x)) \right)
$$
because this is the Rk~\ref{Remark : Useful decomposition} version of the $A_\alpha$ port-decreasing recurrence hypothesis. Since the order is monotonous, this implies that the LHS is greater than the RHS.

Secondly we prove the right inequality. First we notice 
$$\E_{A_\alpha G}(\preci{\beta}(A_\alpha G),x)\setminus\E_G(\cpreci{\omega}(G),x)=E_{A_\alpha G}(\preci{\beta}(A_\alpha G),x).$$
because by monotony $\preci{\beta}(A_\alpha G)\subseteq  \preci{\omega}(G)$, thus $\preci{\beta}(A_\alpha G)\setminus \cpreci{\omega}(G) = \preci{\beta}(A_\alpha G)$.

We can therefore decompose the LHS of the right inequality as 
\begin{align*}
   \E_{A_\alpha G}(x)\setminus\E_G(\cpreci{\omega},x) &=  (\E_{A_\alpha G}(\preci{\beta}(A_\alpha G)) 
 \setminus\E_G(\cpreci{\omega},x))\uplus (\E_{A_\alpha G}(\cpreci{\beta}(A_\alpha G))   \setminus\E_G(\cpreci{\omega},x) ) \\
 &= \E_{A_\alpha G}(\preci{\beta}(A_\alpha G))\uplus (\E_{A_\alpha G}(\cpreci{\beta}(A_\alpha G))   \setminus\E_G(\cpreci{\omega},x) ) 
\end{align*}
We now decompose the RHS:
\begin{align*}
 \E_{A_\omega G}(x)&=(\E_{A_\omega G}(x)\cap\E_{A_\alpha G}(x)) \uplus ((\E_{A_\omega G}(x)\setminus\E_{A_\alpha G}(x))  \\
 &=(\E_{A_\omega G}(x)\cap\E_{A_\alpha G}(\preci{\beta}(A_\alpha G),x)) \\
   &\uplus (\E_{A_\omega G}(x)\cap\E_{A_\alpha G}(\cpreci{\beta}(A_\alpha G),x)) \\
   &\uplus (\E_{A_\omega G}(x)\setminus\E_{A_\alpha G}(x)) \\
  \E_{A_\omega G}(x) \setminus \E_G(\cpreci{\omega},x) &=(\E_{A_\omega G}(x)\cap\E_{A_\alpha G}(\preci{\beta}(A_\alpha G),x)) \setminus \E_G(\cpreci{\omega},x)\\
   &\uplus (\E_{A_\omega G}(x)\cap\E_{A_\alpha G}(\cpreci{\beta}(A_\alpha G),x)) \setminus \E_G(\cpreci{\omega},x)\\
   &\uplus (\E_{A_\omega G}(x)\setminus\E_{A_\alpha G}(x))\setminus \E_G(\cpreci{\omega},x)\\
\E_{A_\omega G}(x) \setminus \E_G(\cpreci{\omega},x) &=E_{A_\omega G}(x)\cap\E_{A_\alpha G}(\preci{\beta}(A_\alpha G),x)\\
   &\uplus (\E_{A_\omega G}(x)\cap\E_{A_\alpha G}(\cpreci{\beta}(A_\alpha G),x)) \setminus \E_G(\cpreci{\omega},x)\\
   &\uplus (\E_{A_\omega G}(x)\setminus\E_{A_\alpha G}(x))\setminus \E_G(\cpreci{\omega},x)\\
\end{align*}
 The following inequality holds because one set is included in the other:
 $$
\pof\left(\E_{A_\alpha G}(\cpreci{\beta}(A_\alpha G))   \setminus\E_G(\cpreci{\omega},x)\right)  > \pof\left((\E_{A_\omega G}(x)\cap\E_{A_\alpha G}(\cpreci{\beta}(A_\alpha G),x)) \setminus \E_G(\cpreci{\omega},x)\right)
$$
Now by the Rk~\ref{Remark : Useful decomposition} version of the $A_\beta$ port-decreasing recurrence hypothesis: 
\begin{align*}
   \E_{A_\alpha G}(\preci{\beta}(A_\alpha G)) 
    >&\ \E_{A_\omega G}(x)\cap\E_{A_\alpha G}(\preci{\beta}(A_\alpha G),x)\\
    &\uplus\E_{A_\omega G}(x)\setminus\E_{A_\alpha G}(x)\\
    \geq&\ \E_{A_\omega G}(x)\cap\E_{A_\alpha G}(\preci{\beta}(A_\alpha G),x)\\
    &\uplus (\E_{A_\omega G}(x)\setminus\E_{A_\alpha G}(x))\setminus \E_G(\cpreci{\omega},x)
\end{align*}
Since the order is monotonous, this implies that the LHS is greater than the RHS.
\end{proof}

\begin{remark}\label{Remark : signification of strict inequality}
    If we have a strict inequality between two sets of ports $A$ and $A'$ for an inclusive order then the bigger set $A'$ necessarily contains at least one element which does not belong to $A$ i.e.:
    $$A>A' \Longrightarrow \exists p\in A,\ p\notin A  $$
    Indeed otherwise we would have $A\subseteq A'$ which would imply the contradiction $A \leq A'$.
\end{remark}

We are now ready to tackle the case of disjoint sequences of the main theorem. 
\begin{proposition}[Disjoint case]\label{Lemma : Consistency , Disjoint case, radius R}
    Let $\restri{}$ be a private extensive monotonous neighbourhood scheme.
    Let $A_{(-)}$ be a time-increasing and commutative $\restri{}$-local rule. Let $G$ be a graph. Let $\omega,\omega'\in \Omega_G(A)$. If $\omega\cap\omega'=\emptyset$ and $A_{(-)}$ is port decreasing, then $A_{\omega}G\ \varparallel\ A_{\omega'}G$.
\end{proposition}

\begin{proof}
    Let us show consistency around $u \in \V_{ A_{\omega'}G}\cap \V_{A_{\omega}G}$.
        
    First we notice that $u$ must be a vertex of $G$. Indeed, due to 
    Rk~\ref{rk:locality excludes borders}
    vertices of $A_{\omega}G$ are either vertices of $G$ or of the form $(t+\Delta t).x$ with $\Delta t>0$ and $t.x\in V_G$. If $u = (t+\Delta t).x$, then by $*$-$\restri{}$-locality (comming from Lem.~\ref{lem:extloc}) and again Rk~\ref{rk:locality excludes borders}, $x$ belongs to $\preci{\omega}(G)\cap \preci{\omega'}(G)$. This contradicts privacy of $\restri{}$.
    
    There are three cases:
    \begin{itemize}
        \item The neighbourhood of $u$ is not modified by $A_\omega$ and neither by $A_{\omega'}$, i.e. $( A_\omega G)_{u}= G_u = ( A_{\omega'} G)_{u}$. This immediately implies consistency.
        \item Case $( A_\omega G)_{u}^0\neq G_u$.  We start by decomposing $\E_G(u)$ according to $\omega$ and $\omega'$ privacy:
        $$\E_G(u) =\E_G(\preci{\omega},u) \uplus\E_G(\preci{\omega'},u) \uplus (\E_G\setminus \E_G(\preci{\omega}\cup\preci{\omega'},u) )).$$
        Indeed since $\restri{}$ is private we have $\preci{\omega}\cap \preci{\omega'} =\emptyset$, therefore this decomposition is a partition. Moreover by $*$-$\restri{}$-locality and Rk~\ref{rk:locality excludes borders}, $A_{\omega'}$ cannot modify the private edges of $\omega$:
        $$\E_G(\preci{\omega},u) \subseteq\E_{A_{\omega'}G}(u).$$
        Since $A_\omega$ is $*$-port-decreasing, the action of $A_{(-)}$ gives us the following equations:
          $$\pof(\E_G(u)\setminus\E_G(\cpreci{\omega},u)) > \pof(\E_{A_\omega G}(u)\setminus\E_G(\cpreci{\omega},u)) $$
        This implies by Rk~\ref{Remark : signification of strict inequality} the existence of a port $p\in \pof(\E_G(\preci{\omega},u))$ such that $p\notin \pof(\E_{A_\omega G}(u)\setminus\E_G(\cpreci{\omega},u))$. Since $A_{\omega'}$ cannot modify private edges of $\omega$ we also have $p\in \P_{A_{\omega'}G}(u)$ and since ports are never repeated $p\notin\E_G(\cpreci{\omega},u))$ thus $p\notin \pof \Em_{A_\omega G}(u)$.
        Since  $p\in \P_{A_{\omega'}G}(u)$ and $p\notin \pof \Em_{A_\omega G}(u)$ the sets differ: consistency is ensured by dismissal of its premise.
         \item Case $( A_{\omega'} G)_{u}\neq G_u$ is symmetrical.
    \end{itemize}
\end{proof}

\generalcase*\label{proof : General Case Consistency annex}

\begin{proof}
    Let $\omega$ and $\omega'$ be valid sequences in $G$. Let us prove that $A_{\omega}G\ \varparallel\ A_{\omega'}G$.

    We show this result by strong recurrence on the length of $\omega$ and $\omega'$. Let us suppose that this property holds for all words of size smaller or equal to $n$. Let $\omega$ and $\omega'$ be words of size at most $n+1$.
    
    If $\omega \cap \omega'=\emptyset$ the result follows from Lem.~\ref{Lemma : Consistency , Disjoint case, radius R}.

    Otherwise we call $x\in\omega\cap\omega'$, a position such that:
    $$\omega = \omega_2 x \omega_1 \text{ s.t. } x \notin \omega_1$$
    $$\omega' = \omega_2' x \omega_1' \text{ s.t. } x \notin \omega_1'$$
    $$\omega_1 \cap \omega_1'=\emptyset$$

    First we prove that $x\in \X(\Past(G))$. 
    If $x\notin \X(\Past(G))$ we would have by validity of $\omega$ that $x\in X (\Past(A_{\omega_1}G))\setminus \X(\Past(G))$ which implies that $G_x\neq (A_{\omega_1} G)_x$. 
    Since $A_{\omega_1}$ is $*$-$\restri{}$-port-decreasing, there exists $e=(u\!:a, t.x\!:b)\in\E_G(\preci{\omega_1},x)$. By privacy $\preci{\omega_1}\cap \restri{\omega'_1}=\emptyset$ and so $u\notin I_{G_{\restri{\omega_1'}}}$. Thus $e$ is at best a border edge in $G_{\restri{\omega_1'}}$. By locality and Rk~\ref{rk:locality excludes borders} this edge cannot be modified by $A_{\omega_1'}$. But by validity of $A_{\omega'}$, we have $x\in \X(\Past(A_{\omega'_1}G))$ meaning that $e$ was removed, a contradiction.

    Now we can apply Lem. \ref{lemma : right commuting of a past vertex} on $\omega_1$ and $x$ to get
    $$A_\omega G = A_{\omega_2}A_{\omega_1}A_xG.$$
    Similarly,
    $$A_{\omega'} G= A_{\omega_2'}A_{\omega_1'}A_xG$$
    which concludes the proof because we can apply the induction hypothesis on $A_x G$.
\end{proof}

\end{document}